\documentclass[12pt]{article}
\usepackage{amsthm,amsmath,amssymb}

\usepackage{caption}
\usepackage{booktabs}
\usepackage{tabularx}
\usepackage{multirow}
\usepackage{times}
\usepackage[colorlinks]{hyperref}
\usepackage{graphicx}
\usepackage{colortbl}
\usepackage{xcolor}

 \usepackage{geometry}
\newcounter{exo}
\makeatletter
\newenvironment{example}[1][]
    {\refstepcounter{exo}%
\begin{center}
   \begin{tabular}{|>{\columncolor{gray!10}}p{0.967\textwidth}|}
    \hline \\[-3mm] {\bfseries{Example \theexo}: #1}\\ \hline 
    }
    { 
    \\[1.5mm] \hline
  \end{tabular} 
    \end{center}
    }
\newcommand{\overbar}[1]{\mkern 1.05mu\overline{\mkern-1.05mu#1\mkern-1.05mu}\mkern 1.05mu}
\newcommand{\downVdash}{\overbar{\top}}

\newtheorem{theorem}{Theorem}
\newtheorem{corollary}{Corollary}%[theorem]

\newtheorem{lemma}{Lemma}
\newcommand{\vset}{{{U}}} %% vertex set
\newcommand{\eset}{{{F}}} %% edge set
\newcommand{\vere}{{{u}}} %% vertex
\newcommand{\ede}{{{e}}}  %% edge
\newcommand{\elec}{\mathcal{E}} %%% election
\newcommand{\gamend}{$h$-amendment}
\newcommand{\Gamend}{$h$-Amendment}
\newcommand{\famend}{full-amendment}
\newcommand{\Famend}{Full-Amendment}
\newcommand{\bigo}[1]{O(#1)}
\newcommand{\smallo}[1]{o(#1)}
\newcommand{\bigos}[1]{O^*(#1)}

\newcommand{\yes}{\mbox{Yes}}
\newcommand{\no}{\mbox{No}}
\newcommand{\yesins}{{\yes}-instance}
\newcommand{\noins}{{\no}-instance}
\newcommand{\poly}{\textsf{P}}
\newcommand{\np}{\textsf{\mbox{NP}}}
\newcommand{\nph}{{\np}-{hard}}
\newcommand{\nphns}{\np-hardness}
\newcommand{\npc}{{\np}-{complete}}

\newcommand{\wa}{\textsf{\mbox{W[1]}}}
\newcommand{\wb}{\textsf{W[2]}}
\newcommand{\wah}{{\textsf{W[1]}}-hard}
\newcommand{\wahshort}{{\textsf{W[1]}}-h}
\newcommand{\wbh}{{\textsf{W[2]}}-hard}
\newcommand{\wbhshort}{{\textsf{W[2]}}-h}
\newcommand{\wac}{\wa-complete}
\newcommand{\wbc}{\wb-complete}
\newcommand{\wbhns}{{\textsf{W[2]}}-hardness}
\newcommand{\fpt}{\textsf{\mbox{FPT}}}
\newcommand{\paranph}{para{\nph}}
\newcommand{\paranphns}{para{\nphns}}
\newcommand{\xp}{\textsf{XP}}
\newcommand{\prob}[1]{{\textsc{#1}}}
\newcommand{\setmid}{:}
\newcommand{\abs}[1]{|#1|}
\newcommand{\edge}[2]{\{#1, #2\}}
\newcommand{\Succ}{\,\ }
\newcommand{\wahns}{\textsf{W[1]}-hardness}
\newcommand{\kav}{k_\text{AV}}
\newcommand{\kdv}{k_{\text{DV}}}
\newcommand{\kac}{k_{\text{AC}}}
\newcommand{\kdc}{k_{\text{DC}}}

\newcommand{\memph}[1]{\emph{#1}}
\newcommand{\dtime}{\textsf{DTIME}}

\newcommand{\arc}[2]{(#1, #2)}

\newcommand{\muplus}{\cup}

\newcommand{\EP}[3]{
\begin{center}
\smallskip
{\small
\begin{tabularx}{0.98\columnwidth}{ll}
\toprule
\multicolumn{2}{l}{\textsc{#1}} \\
\midrule
{\bf Given:}   & \parbox[t]{0.85\columnwidth}{#2\vspace*{2mm}}  \\
{\bf Question:}& \parbox[t]{0.85\columnwidth}{#3\vspace*{0mm}} \\
\bottomrule
\end{tabularx}
}
\smallskip
\end{center}
}

\newtheorem{definition}{Definition}
\theoremstyle{thmstylefive}%
\newtheorem{reductionrule}{Reduction Rule}%
\begin{document}
	
%	\begin{frontmatter}
		
\title{On the Parameterized Complexity of Controlling Amendment and Successive Winners\thanks{A $3$-page extended abstract summarizing the results for the amendment procedure and the successive procedure has appeared in the Proceedings of the 21st International Conference on Autonomous Agents and Multi-Agent Systems (AAMAS~2022)~\protect\cite{DBLP:conf/atal/aamas2022}.}}
%\tnoteref{mytitlenote}}

%\tnotetext[mytitlenote]{A $3$-page extended abstract summarizing the results for the amendment procedure and the successive procedure has appeared in the Proceedings of the 21st International Conference on Autonomous Agents and Multi-Agent Systems (AAMAS~2022)~\protect\cite{DBLP:conf/atal/aamas2022}.}

%\author[]{Yongjie Yang}
%\ead{yyongjiecs@gmail.com}
		
%\address{Chair of Economic Theory, Saarland University, Saarbr\"{u}cken, Germany}
\author{Yongjie Yang}
%\affiliation{
%  \institution{Chair of Economic Theory, Saarland University}
 % \city{Saarb\"{u}cken}
 % \country{Germany}}
%\email{yyongjiecs@gmail.com}

\date{\small{Chair of Economic Theory, Saarland University, Saarb\"{u}cken 66123, Germany \\
yyongjiecs@gmail.com}}

\maketitle

\begin{abstract}
The successive and the amendment procedures have been widely employed in parliamentary and legislative decision making and have undergone extensive study in the literature from various perspectives. However, investigating them through the lens of computational complexity theory has not been as thoroughly conducted as for many other  prevalent voting procedures heretofore. To the best of our knowledge, there is only one paper which explores the complexity of several strategic voting problems under these two procedures, prior to our current work. 
To provide a better understanding of to what extent the two procedures resist strategic behavior, we study the parameterized complexity of constructive/destructive control by adding/deleting voters/candidates for both procedures. To enhance the generalizability of our results, we also examine a more generalized form of the amendment procedure. Our exploration yields a comprehensive (parameterized) complexity landscape of these problems with respect to numerous parameters.
\end{abstract}
		
%\begin{keyword}
%	amendment procedure\sep successive procedure\sep election control\sep {\nph}\sep {\wah}\sep {\wbh}\sep {\fpt}
%\end{keyword}
		
%\end{frontmatter}
	
%%%%%%%%%%%%%%%%%%%%%%%%%%%%%%%%%%%%%%%%%%%%%%%%%%%%%%%%%%%%%%%%%%%%%%%%

\section{Introduction}
The amendment and the successive voting procedures are two fundamental methods commonly employed in parliamentary and legislative decision-making. The amendment procedure is particularly prevalent in countries such as the United States, the United Kingdom, Finland, and Switzerland, among others. Meanwhile, the successive procedure is widely utilized across many European nations, including the Czech Republic, Denmark, Germany, Hungary, and Iceland. For a more detailed discussion, refer to~\cite{Rasch2000}. Since the pioneering works of Black~\cite{Black1958} and Farquharson~\cite{Farquharson1969}, these two procedures have been extensively and intensively investigated in the literature 
(see, e.g., the works of Apesteguia, Ballester, and Masatlioglu~\cite{DBLP:journals/geb/ApesteguiaBM14}, Horan~\cite{Horan2021}, Miller~\cite{Miller1977}, Rasch~\cite{Rasch2000,Rasch2014}).
Each of these procedures takes as input a set of candidates, a group of voters with linear preferences over these candidates, and an agenda (also referred to as a priority), which is specified as a linear order of the candidates. The output is a single candidate selected as the winner. 
The agenda specifies the order in which candidates are considered in the decision making process. 
The amendment procedure takes~$m$ rounds to determine the winner, each determining a temporary winner, where~$m$ denotes the number of candidates. The winner of the first round is the first candidate in the agenda. For the $i^{\text{th}}$ round, the winner is either the~$i^{\text{th}}$ candidate in the agenda or the winner of the $(i-1)^{\text{th}}$ round, determined by the head-to-head comparison between the two candidates, with the winning candidate becoming the winner of the $i^{\text{th}}$ round. The amendment winner is the winner of the last round.
For the successive procedure, the winner is the first candidate in the agenda for whom there is a majority of voters, with each voter preferring this candidate to all successors of the candidate.

In practice, several factors can influence the outcome of an election. For instance, it is known that given the same voting profile, different agendas may result in different winners. Notably, according to the works of Black~\cite{Black1958} and Miller~\cite{Miller1977}, the successive procedure is more vulnerable to agenda control than the amendment procedure, in the sense that for the same profile, any candidate that can be made the amendment winner by some agenda can  also be made the successive winner by some agenda (see also the work of Barber\'{a} and Gerber~\cite{BarberaGerbeer2017} for an extension of this result). To gain a more nuanced understanding of how election outcomes under the amendment and the successive procedures can be affected by different factors,  Bredereck~et~al.~\cite{DBLP:journals/jair/BredereckCNW17} studied the complexity of several related combinatorial problems under these two procedures. Concretely, they studied problems such as {\prob{Agenda Control}}, {\prob{Coalition Manipulation}}, {\prob{Possible Winner}}, {\prob{Necessary Winner}}, and  weighted variants of these problems. Their results reveal that the amendment procedure is more resistant to agenda control  and manipulation than the successive procedure.

To the best of our knowledge, the work of Bredereck~et~al.~\cite{DBLP:journals/jair/BredereckCNW17} is so far the sole exploration into the complexity of strategic problems under these two procedures, leaving many other types of strategic operations for these two procedures unexplored. In an effort to fill these gaps and significantly expand our understanding of the extent to which the two procedures resist other types of strategic operations, we investigate several standard control problems. Particularly, we delve into the problems of constructive control by adding/deleting voters/candidates, initially proposed in the pioneering paper by Bartholdi, Tovey, and Trick~\cite{Bartholdi92howhard}. These problems model scenarios where we are given a number of registered voters and candidates, and an election controller aims to make a distinguished candidate the winner by either adding a limited number of unregistered voters/candidates or deleting a limited number of registered voters/candidates. Additionally, we study the destructive counterparts of these problems, first proposed by Hemaspaandra, Hemaspaandra, and Rothe~\cite{DBLP:journals/ai/HemaspaandraHR07}. In the destructive control problems, the goal of the controller is to prevent the distinguished candidate from becoming the winner. 

To make our results more general, in addition to the two procedures, we also study a generalization of the amendment procedure, which we term $h$-amendment, where~$h$ is a positive integer.  
Generally speaking, this procedure involves multiple rounds, where in each round, a candidate~$c$ is compared with her next~$h$ successors to determine whether~$c$ or the next~$h$ successors of~$c$ in the agenda are eliminated. The winner is the one who survives the last round. (We refer to Section~\ref{sec-pre} for the formal definition.) 
Notably, the standard amendment procedure is exactly the $1$-amendment procedure. 
Furthermore, another intriguing special case of the $h$-amendment procedures arises when~$h$ equals the number of candidates minus one. 
We refer to this special case as the full-amendment procedure. An advantage of the full-amendment procedure is that it selects a superior winner compared to the amendment procedure and the successive procedure. Specifically, for a fixed election and a given agenda, the full-amendment winner either coincides with or is preferred by a majority of voters over the amendment/successive winner.

For these procedures, we provide a comprehensive understanding of the parameterized complexity of election control problems, including many intractability results ({\nphns} results, {\wahns} results, {\wbhns} results, and {\paranphns} results), and numerous tractability results ({\poly}-results and {\fpt}-results). The parameters examined in the paper include the number of predecessors/successors of the distinguished candidate in the agenda, the number of added/deleted voters/candidates, the number of registered voters/candidates, and the number of  voters/candidates not deleted. 
Table~\ref{tab-resulst-summary} summarizes our main results. We also examine two natural parameters:  the number of voters and the number of candidates. Many of the results concerning these parameters are either straightforward to observe or implicitly derived from the results for the aforementioned parameters, with only a few cases remaining open. For more details, we refer the reader to Subsection~\ref{subsec-parameters}. 
Additionally, we note that beyond the results presented in Table~\ref{tab-resulst-summary}, our hardness reductions  also imply many kernelization and approximation lower bounds, which shall be discussed in detail in Section~\ref{sec-lowerbounds}. 

The general message conveyed by our study is that the amendment procedure and the successive procedure exhibit distinct behaviours regarding their resistance to the eight standard control problems. We will discuss this issue at the end of the paper in greater detail. 

As a side note, for readers interested in the topic, it is worth mentioning that similar control problems have been investigated in many other settings, such as  multiwinner voting (see, e.g., the works of Meir~et~al.~\cite{DBLP:journals/jair/MeirPRZ08}, Yang~\cite{DBLP:conf/ijcai/Yang19}), judgment  aggregation (see, e.g., the works of Baumeister~et~al.~\cite{DBLP:conf/stairs/BaumeisterEER12}, Baumeister~et~al.~\cite{DBLP:journals/jcss/BaumeisterEERS20}), group identification (see, e.g., the works of Erd{\'{e}}lyi, Reger, and Yang~\cite{DBLP:journals/aamas/ErdelyiRY20}, Yang and Dimitrov~\cite{DBLP:journals/aamas/YangD18}), tournament solutions (see, e.g., the work of Brill, Schmidt-Kraepelin, and Suksompong~\cite{DBLP:journals/ai/BrillSS22}). 

\subsection{Motivation}
Our motivation for the study is twofold. 

First, the control types (operations) of adding/deleting voters/candidates are highly relevant for practical applications. In the context of parliamentary decision-making, for example, members of parliament have the right to abstain from voting. Deleting voters represents the scenario where a strategic controller may persuade particular voters to abstain, though these voters initially planned not to do so. On the other hand, adding voters represents a scenario where the strategic controller persuades certain voters who initially planned to abstain to participate in the voting procedure. Similarly, adding or deleting candidates is relevant in scenarios where a strategic controller has the power to decide which motions, bills, proposals, etc., are eligible to be put on the agenda. 

Second, investigating the complexity of the aforementioned election control problems for many traditional voting procedures (those that do not require an agenda to determine a winner) had partially dominated the early development of computational social choice for several years, culminating in an almost complete complexity landscape of these problems (see the book chapters~\cite{BaumeisterR2016,handbookofcomsoc2016Cha7FR} for important progress up to 2016, and the papers~\cite{DBLP:journals/aamas/ErdelyiNRRYZ21,DBLP:journals/ai/NevelingR21,DBLP:conf/atal/Yang17,AAMAS17YangBordaSinlgePeaked,DBLP:conf/ecai/000120} for recent results).  The motivation behind the extensive research is rooted in the belief that complexity serves as a robust barrier against strategic behavior. Given the broad applicability of the amendment and the successive procedures, we believe it is of great importance to address the gaps in this regard for the two procedures.  

It is worth noting that several experimental studies have demonstrated that many computationally hard election problems can often be solved efficiently on practical instances, providing empirical evidence that mitigates the perception of complexity as a barrier to strategic voting. 
However, we emphasize that establishing the theoretical complexity of a problem is a foundational step toward a comprehensive understanding of its nature and limitations. This theoretical groundwork is indispensable for systematically addressing the challenges and opportunities presented by these problems in both theoretical and practical domains.

\subsection{Organization}
In Section~\ref{sec-pre}, we provide important notations used in our study. Following this, we unfold our main results in Section~\ref{sec-main-results}, which is further divided into three subsections respectively covering our results for the $h$-amendment procedures, the $(m-h)$-amendment procedures where~$m$ is the total number of candidates and~$h$ is a constant, and the successive procedure. Moving on, we discuss algorithmic lower bounds implied by our reductions in Section~\ref{sec-lowerbounds}. Finally, in Section~\ref{sec-conclusion}, we briefly conclude our results, discuss their implications on the usage of these procedures, and lay out several topics for future research.   

\begin{table}
\renewcommand{\tabcolsep}{0.6mm}
%\renewcommand\arraystretch{1.8}
%\captionsetup{singlelinecheck=off}
\caption{A summary of the parameterized complexity of election control problems under the amendment procedure, the full-amendment procedure, and the successive procedure. 
Results with the superscripts~$\top$ and~$\bot$ respectively mean that they hold when the distinguished candidate is the first and the last candidate in the agenda.  
Results with the superscript~${\downVdash}$ mean that they hold as long as the distinguished candidate is not the first one in the agenda. 
{\fpt}-results with the superscripts~$\leftarrow$ and~$\rightarrow$ are respectively with respect to the number of predecessors and the number of successors of the distinguished candidate. 
All {\wahns} and {\wbhns} results for {\prob{CCAV}}-$\tau$ and {\prob{DCAV}}-$\tau$ are with respect to the combined parameter of the number of added votes and the number of registered votes. 
All {\wahns} and {\wbhns} results for {\prob{CCDV}}-$\tau$ and {\prob{DCDV-$\tau$}} are with respect to both the number of deleted votes and the number of remaining votes. 
The {\wahns} of \prob{CCDC}-Successive and \prob{DCDC}-Successive holds with respect to both the number of deleted candidates and the number of remaining candidates. 
All {\wbhns} results of {\prob{CCAC}}-$\tau$ are with respect to the number of added candidates. 
The {\wbhns} of {\prob{CCAC}}-Successive holds even if we add one additional restriction that there are only two registered candidates. 
The {\wbhns} of {\prob{DCAC}}-Successive holds even when the distinguished candidate is the only registered candidate. 
}
\centering
{
\begin{tabular}{lllll} \toprule
&{\prob{CCAV}}-$\tau$
&{\prob{CCDV}}-$\tau$
&{\prob{CCAC-$\tau$}}
&{\prob{CCDC-$\tau$}}
\\ \midrule

{amendment}
&{\wahshort}$^{\top}$ \cite{DBLP:journals/tcs/LiuFZL09}
&{\wahshort}$^{\top}$ \cite{DBLP:journals/tcs/LiuFZL09,DBLP:journals/tcs/LiuZ13}
&immune$^{\top}$ \cite{Bartholdi92howhard}
&{{\poly} (Thm.~\ref{thm-ccdc-amd-p})}
\\

&{\wbhshort}$^{\bot}$ (Thm.~\ref{thm-ccav-amd-np})
&{\wbhshort}$^{\bot}$ (Thms.~\ref{thm-ccdv-amd-np},~\ref{thm-ccdv-amd-wbh-remaining-votes})
&{\poly} (Thm.~\ref{thm-ccac-amd-p})
& \\ \midrule

{full-}
&{\wahshort}$^{\top}$ \cite{DBLP:journals/tcs/LiuFZL09}
&{\wahshort}$^{\top}$ \cite{DBLP:journals/tcs/LiuFZL09,DBLP:journals/tcs/LiuZ13}
&immune$^{\top}$ \cite{Bartholdi92howhard}
&{\poly} (Thm.~\ref{thm-ccdc-suc-p})
\\

amendment
&{\wahshort}$^{\bot}$  (Thm.~\ref{thm-ccav-suc-np})
&{\wbhshort}$^{\bot}$ (Thms.~\ref{thm-ccdv-suc-np},~\ref{thm-ccdv-famend-wbh-remaining})
&{\fpt}$^{\leftarrow}$ (Thm.~\ref{thm-ccac-suc-fpt})
& \\

&
&
&{\wbhshort}$^{\bot}$ (Thm.~\ref{thm-ccac-suc-np})
&\\ \midrule

{successive}
&{\fpt}$^{\leftarrow}$ (Cor.~\ref{cor-ccav-ccdv-suc-fpt-predecessors})
&{\fpt}$^{\leftarrow}$ (Cor.~\ref{cor-ccav-ccdv-suc-fpt-predecessors})
& immune$^{\top}$ (Cor.~\ref{cor-immune-suc-first})
& {\wahshort}$^{\top}$ (Thms.~\ref{thm-ccdc-suc-wah-solution-size}, \ref{label-ccdc-suc-wah-dual-ssize})
\\

&{\wahshort}$^{\bot}$ (Thm.~\ref{thm-ccav-suc-wah-last})
&{\wbhshort}$^{\bot}$ (Thms.~\ref{thm-ccdv-suc-wah-last},~\ref{thm-ccdv-suc-wah-remainning-last})
&{\wbhshort}$^{\downVdash}$ (Thm.~\ref{thm-ccac-suc-wbh-solution-size})
&{\fpt}$^{\rightarrow}$ (Thm.~\ref{thm-ccdc-suc-fpt-sucessors}) \\ \bottomrule \toprule
%%%%%%%%%%%%%%%%%%%%%%%%%%%%
%%%%%%%%%%%%%%%%%%%%%%%%%%
%%%%%%%%%%%%%%%%%%%%%%%%%%
&\prob{DCAV-$\tau$}
&\prob{DCDV-$\tau$}
&\prob{DCAC-$\tau$}
&\prob{DCDC-$\tau$}
\\ \midrule

{amendment}
&{\fpt}$^{\leftarrow}$ (Cor.~\ref{cor-dcav-dcdv-amd-fpt})
&{\fpt}$^{\leftarrow}$ (Cor.~\ref{cor-dcav-dcdv-amd-fpt})
&{{\poly} (Cor.~\ref{cor-dcac-dcdc-amd-p})}
&immune$^{\top}$ \cite{Bartholdi92howhard}\\

&{\wahshort}$^{\bot}$ (Thm.~\ref{thm-dcav-amd-np})
&{\wbhshort}$^{\bot}$ (Thms.~\ref{thm-dcdv-amd-np},~\ref{thm-dcdv-amd-wbh-remaining})
&
&{\poly} (Cor.~\ref{cor-dcac-dcdc-amd-p}) \\ \midrule

full-
&{{\poly}$^{\top}$  (Cor.~\ref{cor-dcav-dcdv-amd-suc-p})}
&{{\poly}$^{\top}$ (Cor.~\ref{cor-dcav-dcdv-amd-suc-p})}
&\multirow{2}{*}{{\poly} (Thm.~\ref{thm-dcac-suc-p})}
&immune$^{\top}$ \cite{Bartholdi92howhard}\\

amendment
&{\wahshort}$^{\downVdash}$ (Thm.~\ref{thm-dcav-suc-np})
&{\wbhshort}$^{\downVdash}$ (Thms.~\ref{thm-dcdv-suc-np},~\ref{thm-dcdv-amend-wbh-remaining})
&
&{\poly} (Cor.~\ref{cor-dcdc-suc-p})
\\ \midrule

{successive}
&{\poly~(Cor.~\ref{cor-dcav-dcdv-suc-p})}
&{\poly~(Cor.~\ref{cor-dcav-dcdv-suc-p})}
&{\wbhshort}$^{\top}$ (Thm.~\ref{thm-dcac-suc-wbh})
&immune$^{\top}$ (Cor.~\ref{cor-suc-immune-dcdc-first})
\\

&
&
&{\fpt}$^{\rightarrow}$ (Thm.~\ref{thm-dcac-suc-fpt-successors})
&{\wahshort}$^{\downVdash}$ (Thms.~\ref{thm-dcdc-suc-wah-dis-last-solution},~\ref{thm-dcdc-suc-wah-remaining-one-predecessor})
\\  \bottomrule
\end{tabular}
}
\label{tab-resulst-summary}
\end{table}

\section{Preliminaries}
\label{sec-pre}
We assume that the reader has a basic understanding of (parameterized) complexity theory and graph theory. This section provides the necessary notations for our study, while terms not explicitly defined in the paper are referenced in the monographs edited by Bang-Jensen and Gutin~\cite{DBLP:books/sp/BG2018} and the work of West~\cite{Douglas2000}. For a gentle introduction to complexity theory, we recommend the paper by Tovey~\cite{DBLP:journals/interfaces/Tovey02}.

For an integer~$i$,~$[i]$ denotes the set of all positive integers less than or equal to~$i$.

\subsection{Elections}
An election is defined as a tuple $\elec = (C, V)$, where~$C$ is a set of candidates and $V$ is a multiset of votes cast by the voters. 
In particular, each $\succ \in V$ is a linear order on $C$, representing the preference of a corresponding voter over the candidates in~$C$. Specifically, $a \succ b$ for $a, b \in C$ indicates that the voter prefers $a$ to $b$. We also say that the vote~$\succ$ ranks~$a$ before~$b$ if $a\succ b$. 
For notational brevity, we sometimes write a preference in the format of a sequence of candidates from the most preferred one to the least preferred one.  For instance, stating a vote with the preference $a\Succ b\Succ c$ means that, according to the vote,~$a$ is ranked before~$b$, and~$b$ ranked before~$c$.
For two disjoint subsets $X, Y\subseteq C$, $X\succ Y$ means that all candidates in~$X$ are ranked before all candidates in~$Y$ in the vote~$\succ$. 
For $C'\subseteq C$, $V|_{C'}$ is the multiset of votes obtained from those in~$V$ by removing all candidates in $C\setminus C'$. Thus, $(C', V|_{C'})$ is the election $(C, V)$ restricted to~$C'$. For notational convenience, we use $(C', V)$ to denote $(C', V|_{C'})$.

A subset $C'\subseteq C$ is called a block with respect to $V$ if its members are ranked consecutively in all votes from~$V$, i.e., for each $\succ\in V$ and each $c\in C\setminus C'$, it holds that either $\{c\}\succ C'$ or $C'\succ \{c\}$. An ordered block on a subset~$C'$ of candidates is a linear order on~$C'$ such that~$C'$ is a block and, moreover, all votes of~$V$ restricted to~$C'$ are consistent with the linear order. For instance, if $(a, b, c)$ is an ordered block, all votes rank $\{a, b, c\}$ together, and they all prefer~$a$ to~$b$ to~$c$.

Let~$n_{V}(a, b)$ denote the number of votes in~$V$ ranking~$a$ before~$b$, i.e., $n_V(a, b)=\abs{\{\succ\in V \setmid a\succ b\}}$. We say that~$a$ beats~$b$ with respect to~$V$ if $n_V(a, b) > n_V(b, a)$, and that~$a$ ties with~$b$ with respect to~$V$ if $n_V(a, b) = n_V(b, a)$.   
For a candidate $a\in C$ and a subset $C'\subseteq C\setminus \{a\}$, we say that~$a$ majority-dominates~$C'$ with respect to~$V$ if there exists $V'\subseteq V$ such that $\abs{V'}> \abs{V}/2$, and $\{a\}\succ C'$ holds for all $\succ \in V'$. For uniformity, we define that every candidate $a\in C$ majority-dominates the empty set~$\emptyset$.

A candidate~$c\in C$ is the Condorcet winner of~$(C, V)$ if~$c$ beats every candidate from $C\setminus \{c\}$ with respect to~$V$. A candidate~$c$ is a weak Condorcet winner of~$(C, V)$ if~$c$ is not beaten by any other candidate from $C$ with respect to~$V$. A voting procedure is Condorcet-consistent if it always selects the Condorcet winner of an election as the winner whenever the election admits a Condorcet winner. 

An oriented graph is a directed graph (digraph) in which there is at most one arc between any two vertices. The majority graph of an election $(C, V)$ is an oriented graph with vertex set~$C$, where there is an arc from candidate $a$ to candidate $b$ if and only if $a$ beats $b$ with respect to $V$. An arc from vertex $a$ to vertex $b$ in a digraph is denoted by $\arc{a}{b}$. By assigning a weight of $n_{V}(a, b)$ to each arc $\arc{a}{b}$ in the majority graph, we obtain the weighted majority graph of $(C, V)$.

An agenda is a linear order~$\rhd$ on~$C$. We use~$\rhd[i]$ to denote the~$i^{\text{th}}$ candidate in~$\rhd$. Let $m=\abs{C}$. For two integers~$i$ and~$j$, we define $\rhd[i, j] = \{\rhd[x] \setmid x \in [m], i \leq x \leq j\}$. Note that if $j > m$, then $\rhd[i, j] = \rhd[i, m]$, and if $i > j$, then $\rhd[i, j] = \emptyset$. 
%Besides, we use $\overrightarrow{\rhd}[i,j]$ to denote~$\rhd$ restricted to~$\rhd[i,j]$. For a subset $C'\subseteq C$, $\overrightarrow{\rhd}[C']$ denotes~$\rhd$ restricted to~$C'$. 
For a candidate $c\in C$, we call candidates before~$c$ in the agenda the predecessors of~$c$, and call those after~$c$ her successors. In particular, for two candidates~$\rhd[i]$ and~$\rhd[j]$ such that $j>i$, we call~$\rhd[j]$ the $(j-i)^{\text{th}}$-successor of $\rhd[i]$, or say that~$\rhd[j]$ is a $t_{\leq}^{\text{th}}$-successor of $\rhd[i]$ where~$t$ is an integer such that $t\geq j-i$. The number of predecessors of~$c$ plus one is called the position of~$c$ in the agenda~$\rhd$.
The successive  winner and the {\gamend} winner with respect to $(C, V)$ and~$\rhd$ are determined as follows.
\begin{itemize}
\item {\bf{Successive}}\footnote{The successive procedure has also been studied under the name Euro-Latin procedure (see, e.g., the work of Farquharson~\cite{Farquharson1969}).} Let $i\in [m]$ be the smallest integer such that~$\rhd[i]$ majority-dominates $\rhd[i+1,m]$ with respect to~$V$. 
    The successive winner is~$\rhd[i]$.

\item {\bf{\Gamend}} This procedure comprises several rounds, where in each round either the first candidate in the current agenda is eliminated, or up to~$h$ candidates following the first candidate in the agenda are eliminated. Specifically, in round~$r$ where $r = 1, 2, \dots$, if the first candidate~$\rhd[1]$ in the current agenda beats every candidate in $\rhd[2, h+1]$, then all candidates in $\rhd[2, h+1]$ are eliminated from~$\rhd$. Otherwise,~$\rhd[1]$ is eliminated from~$\rhd$. The procedure continues until only one candidate remains in the agenda; this candidate is then declared the winner.

We emphasize that when a set $\rhd[x, y]$ of candidates is eliminated, each~$\rhd[j]$ for $j > y$ is automatically shifted upward by $y - x + 1$ positions in~$\rhd$. As a result, the first candidate in the agenda is always denoted by~$\rhd[1]$.

Furthermore, if the election involves at most $m \leq h$ candidates, the {\gamend} winner is determined using the $(m-1)$-amendment procedure.
\end{itemize}

The amendment procedure is exactly the $1$-amendment procedure\footnote{The amendment procedure has also  been studied under the names Anglo-American procedure, elimination procedure, and ordinary committee procedure (see, e.g., the works of Black~\cite{Black1958}, Farquharson~\cite{Farquharson1969}).}. 
The full-amendment procedure is the $(m-1)$-amendment procedure. Notice that as each round in the amendment procedure eliminates exactly one candidate, the amendment procedure takes $m-1$ rounds. In contrast, the full-amendment procedure may take anywhere from~$1$ to $m-1$ rounds depending on the situation.

From the definition, it is easy to see that the $h$-amendment procedure is Condorcet-consistent for all possible values of~$h$. On the other hand, the successive procedure is not Condorcet-consistent (see the work of Miller~\cite{Miller1977} for detailed discussion). To verify this, consider an election with three votes: \underline{$a\Succ b\Succ c\Succ d$}, \underline{$d\Succ a\Succ b\Succ c$}, and \underline{$b\Succ c\Succ a\Succ d$}. Obviously,~$a$ is the Condorcet winner, but with respect to the agenda $(a, b, c, d)$, the successive winner is~$b$ ($a$ is not the successive winner because only one out of three votes ranks~$a$ before all her successors). Nevertheless, it is easy to see that if an election admits a Condorcet winner, then, as long as the Condorcet winner is the last one in the agenda, this Condorcet winner is also the successive winner of the election (see the work of Rasch~\cite{Rasch2000} for detailed discussion). 
In addition, it holds that the $h$-amendment/successive winner~$c$ of each election is the Condorcet winner of the election restricted to~$c$ and all successors of~$c$.
As pointed out earlier, the full-amendment winner of an election with respect to an agenda beats both the amendment winner and the successive winner of the same election with respect to the same agenda, if they do not coincide. One can also observe that, in the same election, the full-amendment winner can be neither a successor of the amendment winner nor a successor of the successive winner.

Given an oriented graph~$G$ with the vertex set~$C$ and an agenda~$\rhd$ on~$C$, the $h$-amendment/successive winner of~$G$ with respect to~$\rhd$ is the $h$-amendment/successive winner of an election whose majority graph is~$G$ with respect to~$\rhd$.
See Example~\ref{ex-a} for an illustration.

\begin{example}[An illustration of different voting procedures.] 
\label{ex-a}
Consider an election $(C, V)$, where $C=\{a, b, c, d\}$, and~$V$ consists of three votes. The three votes in~$V$ (left side) and the majority graph of~$(C, V)$ together with three different agendas are given below.

\smallskip
\begin{minipage}{0.2\textwidth}
%\begin{center}
vote 1: $b\Succ d\Succ c\Succ a$

vote 2: $c\Succ a\Succ b\Succ d$

vote 3: $a\Succ d\Succ b\Succ c$
%\end{center}
\end{minipage}\begin{minipage}{0.8\textwidth}
\begin{center}
\includegraphics[width=0.95\textwidth]{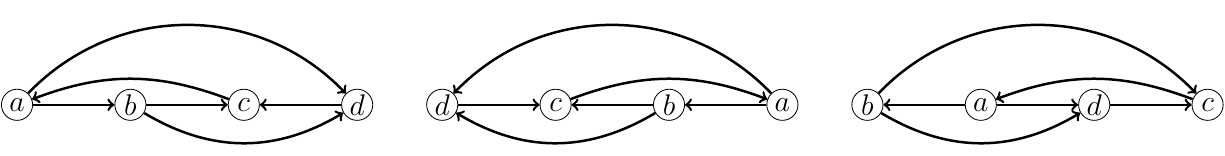}
\end{center}
\end{minipage}
\smallskip

With respect to the agenda $(a, b, c, d)$,~$b$ is the $h$-amendment winner for $h\in \{2, 3\}$, and~$d$ is both the amendment winner and the successive winner.
With respect to the agenda $(d, c, b, a)$,~$a$ is both the successive winner and the $h$-amendment winner for all $h\in [3]$.
With respect to the agenda $(b, a, d, c)$,~$c$ is the amendment winner, and $d$ is the successive winner and the $h$-amendment winner for $h\in \{2, 3\}$.
\end{example}

We remark that, at first glance, it may seem that the $(h+1)$-amendment winner of an election with respect to an agenda cannot be a successor of the election's $h$-amendment winner under the same agenda. However, this is not the case, as the $h$-amendment winner may be eliminated by some of its predecessors (which may also be eliminated later) when the $(h+1)$-amendment procedure is applied. To illustrate this, consider the election derived from the one in Example~\ref{ex-a} by introducing an additional candidate,~$e$, who is ranked first in the second vote and immediately after~$d$ in the other two votes, such that~$e$ is beaten by~$d$ but beats both~$a$ and~$b$. With respect to the agenda $(e, a, d, b, c)$, $b$ is the amendment winner, while~$c$ is the $2$-amendment winner.

\subsection{Election Control Problems}
Let~$\tau$ be a voting procedure.
We study eight standard control problems which are special cases of the problems defined below.

\EP
{Constructive Multimode Control {\textnormal{for}} $\tau$ (CMC-$\tau$)}
{A set~$C$ of registered candidates, a distinguished candidate $p\in C$, a set~$D$ of unregistered candidates, a multiset~$V$ of registered votes over $C\cup D$, a multiset~$W$ of unregistered votes over $C\cup D$, an agenda~$\rhd$ on~$C\cup D$, and four nonnegative integers~$\kav$,~$\kdv$,~$\kac$, and~$\kdc$.}
{Are there $C'\subseteq C\setminus \{p\}$, $D'\subseteq D$, $V'\subseteq V$, and $W'\subseteq W$ such that\\
(1) $\abs{C'}\leq \kdc$, $\abs{D'}\leq \kac$, $\abs{V'}\leq \kdv$, $\abs{W'}\leq \kav$, and\\
(2) $p$ is the~$\tau$ winner of $((C\setminus C')\cup D', (V\setminus V')\muplus W')$ with respect to the agenda~${\rhd}$ restricted to $(C\setminus C')\cup D'$?
}

To put it in plain words, in the {\prob{CMC-${\tau}$}} problem, an election controller attempts to make a given distinguished candidate the winner by deleting a limited number of registered candidates and votes and adding a limited number of unregistered candidates and votes.

\prob{Destructive Multimode Control {\textnormal{for}}~$\tau$} ({\prob{DMC-$\tau$}}) is defined similar to {\prob{CMC-$\tau$}} with only the difference that the goal of the controller is to make the given distinguished candidate not the winner (i.e., we require in requirement~(2) in the above definition that~$p$ is not the~$\tau$ winner of $((C\setminus C')\cup D', (V\setminus V')\muplus W')$).\footnote{Many traditional voting procedures, such as Borda, Copeland, Maximin, etc., do not need an agenda to determine the winners. In these cases,~$\rhd$ is dropped in the definitions of {\prob{CMC-$\tau$}} and {\prob{DMC-$\tau$}}.}

The eight standard control problems are special cases of {\prob{CMC-$\tau$}} and {\prob{DMC-$\tau$}}. The abbreviations of the problem names and their specifications are provided in Table~\ref{tab-problem-specifications}. For simplicity, when we write an instance of a standard control problem, we omit its components of constant values~$0$ and~$\emptyset$, and use~$k$ to denote the one in $\{k_{\text{AV}}, k_{\text{DV}}, k_{\text{AC}}, k_{\text{DC}}\}$ not requested to be zero. For example, an instance of {\prob{CCAV}}-$\tau$ is written as $(C, p, V, W, \rhd, k)$, where~$k$ represents~$k_{\text{AV}}$.  

\begin{table}
\caption{Specifications of some election control problems. In the abbreviations, the letter~X is either CC standing for ``constructive control'' or DC standing for ``destructive control''. For $\text{X}=\text{CC}$ the problems are special cases of {\prob{CMC-$\tau$}}, and for $\text{X}=\text{DC}$ the problems are special cases of {\prob{DMC-$\tau$}}. The second letters~A and~D in the abbreviations respectively stand for ``adding'' and ``deleting'', and the third letters~V and~C  respectively stand for ``voters'' and ``candidates''. 
}
\centering
\begin{tabular}{ll}\toprule
abbreviations & restrictions\\ \midrule

\prob{XAV-$\tau$}  & $\kdv=\kac=\kdc=0$, and $D=\emptyset$\\

\prob{XDV-$\tau$}   &  $\kav=\kac=\kdc=0$, and $W=D=\emptyset$\\

\prob{XAC-$\tau$}   & $\kav=\kdv=\kdc=0$, and $W=\emptyset$ \\

\prob{XDC-$\tau$}  & $\kav=\kdv=\kac=0$, and $W=D=\emptyset$ \\ \bottomrule
\end{tabular}
\label{tab-problem-specifications}
\end{table}

To improve readability and avoid repetitive phrasing, we adopt the following simplifications in discussions concerning an election $(C, V)$ restricted to a subset~$C'$ of candidates. (Such discussions are particularly relevant to election control problems involving the addition or deletion of candidates.)  Instead of explicitly stating that a candidate~$c$ is the $\tau$ winner of $(C, V)$ with respect to an agenda~$\rhd$ restricted to~$C'$, we simply say that~$c$ is the~$\tau$ winner of $(C', V)$ with respect to $\rhd$. When the context makes the agenda clear, we may further simplify by stating that~$c$ is the $\tau$ winner of $(C', V)$.  

In the paper, we establish numerous hardness results for restricted cases of the control problems where the distinguished candidate~$p$ occupies a specific position (e.g., the first or the last) in the agenda. Additionally, when we assert that a hardness result holds as long as~$p$ is not the first candidate in the agenda, we mean that the result holds when restricted to instances where~$p$ is in the~$i^{\text{th}}$ position in the given agenda, for all positive integers~$i\geq 2$. 

A voting procedure is immune to a constructive control type if it is impossible to turn a nonwinning candidate into a winner by performing the corresponding controlling operation (adding/deleting voters/candidates). Analogously, a voting procedure is immune to a destructive control type if it is impossible to turn a winner into a nonwinning candidate by performing the corresponding control operation. 
%We say that a voting procedure is susceptible to a control type if it is not immune to it. 

\subsection{Parameterized Complexity}
A parameterized problem is a subset of $\Sigma^* \times \mathbb{N}$, where $\Sigma$ is a finite alphabet. A complexity hierarchy has been developed to classify parameterized problems into numerous classes:
\[\fpt\subseteq \wa\subseteq \wb\subseteq \cdots \subseteq \xp.\] 
For a detailed discussion of these classes, we refer the reader to the monograph by Downey and Fellows~\cite{DBLP:series/txcs/DowneyF13}. 

For each $X\in \Sigma^*$, let~$\abs{X}$ denote the size of~$X$. 
A parameterized problem is fixed-parameter tractable ({\fpt}) if there is an algorithm which solves each instance $(X, \kappa)$ of the problem in time $f(\kappa)\cdot \abs{X}^{\bigo{1}}$, where~$f$ can be any computable function in the parameter~$\kappa$. A parameterized problem is in the class {\xp} if there is an algorithm which solves each instance $(X, \kappa)$ of this problem in time $\abs{X}^{f(\kappa)}$, where~$f$ can be any computable function in the parameter~$\kappa$.

\begin{definition}[Parameterized Reduction]
A parameterized problem~$P$ is parameterized reducible to a parameterized problem~$Q$ if there is an algorithm which takes as input an instance $(X, \kappa)$ of~$P$ and outputs an instance $(X', \kappa')$ of~$Q$ such that
\begin{itemize}
\item the algorithm runs in time $f(\kappa)\cdot \abs{X}^{\bigo{1}}$;
\item $(X, \kappa)$ is a {\yesins} of~$P$ if and only if $(X', \kappa')$ is a {\yesins} of~$Q$; and
\item $\kappa'\leq g(\kappa)$ for some computable function~$g$ in~$\kappa$.
\end{itemize}
\end{definition}

A parameterized problem is {\wah} (respectively, {\wbh}) if all problems in {\wa} (respectively, {\wb}) are parameterized reducible to this problem. A {\wah} (respectively, {\wbh}) problem that is in {\wa} (respectively, {\wb}) is {\wac} (respectively, {\wbc}). {\wah} and {\wbh} problems do not admit any {\fpt}-algorithms unless the above hierarchy collapses to some level. 
To prove that a new problem is {\wah} (respectively, {\wbh}), one needs to select a known {\wah} (respectively, {\wbh}) problem and establish a parameterized reduction from the known problem to the new one.

Beyond the class of {\xp}, there exists another significant class of parameterized problems known as {\paranph}. A parameterized problem is classified as {\paranph} if it remains {\nph} for some fixed value of the parameter.  
For a more comprehensive exploration of parameterized complexity theory, we recommend the textbooks by Cygan~et~al.~\cite{Cygan2015}, Downey and Fellows~\cite{DBLP:series/txcs/DowneyF13}, and Niedermeier~\cite{rolf06}.

\subsection{Parameters We Consider} 
\label{subsec-parameters}

Our results for the eight election control problems encompass the following natural parameters.

\begin{description}
    \item[Solution size~$k$] The parameter $k$, representing the number of voters or candidates allowed to be added or deleted, is especially relevant in scenarios where the election controller faces resource constraints or seeks to minimize the exertion of influence. For instance, consider a university senate deciding on a critical policy change, such as revising tuition fees. 
In this scenario:
\begin{itemize}
    \item Adding voters corresponds to persuading external experts, alumni, or student representatives who were not initially planning to participate, to join the voting process. However, this might entail significant communication efforts, logistical arrangements, or financial costs (e.g., compensating for their time or travel expenses).
    \item Deleting voters may reflect attempts to convince certain stakeholders (e.g., administrators or faculty members) to abstain from voting, requiring delicate negotiations and justification to avoid reputational damage or conflict.
\end{itemize}

Suppose the election controller supports the policy change but operates under a limited budget, allowing them to influence the participation of at most $k$ voters. Here, minimizing $k$ aligns with practical constraints such as limited time to communicate with potential participants, a fixed budget for compensating or incentivizing voters, and the controller's preference to avoid significant disruptions or visible interventions.

    In another example, consider a national election where $k$ represents the number of candidates to be added or removed from the ballot to favor a preferred candidate. Adding candidates may involve recruiting new participants, which can be administratively challenging and resource-intensive. Similarly, removing candidates might require legal actions, political negotiations, or significant financial resources, all of which are constrained by $k$.
    
We show that most hard control problems exhibit fixed-parameter intractability concerning this parameter. Notably, for the problems of election control by adding voters/candidates, the intractability persists even when combined parameterized by the solution size and the number of registered voters/candidates, or when there is only one or two registered candidates.
    
    \item[Number of voters/candidates not deleted] This parameter becomes particularly relevant in scenarios where the election controller is tasked with selecting a small subset of voters/candidates to participate in the decision-making procedure, with the unselected voters/candidates effectively being considered as deleted. For example, consider the context of corporate board elections or committee decisions in organizations. In such settings, the controller might wish to identify a small representative group of voters—such as key stakeholders, senior employees, or domain experts—who will cast votes on behalf of a larger group, thereby streamlining the decision-making process. Here, the parameter captures the practical constraint of selecting a manageable number of voters while ensuring that the election outcome remains aligned with the controller’s objectives. 

    For problems involving voter or candidate deletions, our results demonstrate that these hardness control problems are fixed-parameter intractable with respect to this parameter. This intractability underscores the challenge of achieving the controller’s goals when the number of retained voters or candidates is tightly constrained, reinforcing the significance of understanding such scenarios both theoretically and in practice.
    
    \item[Number of predecessors/successors of the distinguished candidate] Previous works indicate that candidates may benefit from their positions on the agenda. For example, in a scenario with three candidates and a decision determined by the amendment procedure, it becomes evident that the candidate listed last on the agenda is more likely to be selected.\footnote{If a Condorcet winner exists, the Condorcet winner will be selected regardless of her position. In all other cases, the procedure consistently designates the last candidate on the agenda as the winner.} Hence, it is intriguing to explore if the control problem faced by the election controller becomes more manageable when the distinguished candidate has only a few predecessors or successors. Our findings demonstrate that concerning this parameter, certain problems are fixed-parameter intractable, while others are {\fpt}. 
\end{description}

For the $h$-amendment and the $(m-h)$-amendment procedures, we also analyze the parameter~$h$. We prove that many control problems remain computationally hard even for constant values of~$h$, thereby identifying their classification within the hierarchy of {\paranph} problems with respect to~$h$.

Two additional natural parameters are the number~$m$ of candidates and the number~$n$ of voters. It is straightforward to verify that the candidate control problems (CCAC, CCDC, DCAC, DCDC) are {\fpt} with respect to~$m$. For the voter control problems parameterized by~$m$, one of our results (Theorem~\ref{thm-mgcav-fpt}) establishes their {\fpt} status. Detailed discussions on this matter will follow the presentation of Theorem~\ref{thm-mgcav-fpt}. Consequently, our investigation implicitly provides a comprehensive understanding of the parameterized complexity of the eight standard control problems for the parameter~$m$. 

Regarding the parameter~$n$, it is easy to see that the voter control problems (CCAV, CCDV, DCAV, DCDV) are {\fpt}. Furthermore, several polynomial-time solvability results for these problems are summarized in Table~\ref{tab-resulst-summary}. However, beyond these tractability results, the {\fpt} status of certain candidate control problems parameterized by~$n$ remains an open question. Specifically, this includes {\prob{CCAC}}-{\Famend}, {\prob{CCAC}}-Successive, {\prob{CCDC}}-Successive, {\prob{DCAC}}-Successive, and {\prob{DCDC}}-Successive. 
We note that a systematic study of the candidate control problems with respect to~$n$ for many traditional voting procedures has been conducted by Chen~et~al.~\cite{DBLP:journals/jair/ChenFNT17}.

\subsection{Auxiliary Problems}
\label{subsec-supporting-problems}
We elaborate on problems used for establishing our hardness results.
Let $G=(\vset, \eset)$ be a simple undirected graph. For a vertex $\vere\in \vset$, $N_G(\vere)=\{\vere'\in \vset \setmid \edge{\vere}{\vere'}\in \eset\}$ is the set of neighbors of~$\vere$ in~$G$, and $N_G[\vere]=N_G(\vere)\cup \{\vere\}$ is the set of closed neighbors of~$\vere$ in~$G$. For $\vset'\subseteq \vset$, let $N_G(\vset')=\bigcup_{\vere\in \vset'} N_G(\vere)\setminus \vset'$, and let $N_G[\vset']=N_G(\vset')\cup  \vset'$. In addition, $G-\vset'$ is the graph obtained from~$G$ by deleting all vertices in~$\vset'$. Two vertices are adjacent if there exists an edge between them. An independent set of~$G$ is a subset of vertices that are pairwise nonadjacent. A clique of $G$ is a subset of pairwise adjacent vertices. A vertex~$\vere$ dominates a vertex~$\vere'$ if either they are the same vertex or they are adjacent. 
For two subsets $S, S'\subseteq \vset$, we say that~$S$ dominates~$S'$ if every vertex in~$S'$ is dominated by at least one vertex in~$S$.  A dominating set of~$G$ is a subset of vertices which dominate all vertices of~$G$.
 A perfect code of~$G$ is a dominating set~$P$ such that every vertex in~$G$ is dominated by exactly one vertex in~$P$, i.e.,~$P$ is a dominating set of~$G$ such that for every two distinct $\vere, \vere'\in P$ it holds that $N_G[\vere]\cap N_G[\vere']=\emptyset$. Note that every perfect code is also an independent set.

For two positive integers~$i$ and~$j$,~$K_{i, j}$ denotes the complete bipartite graph with~$i$ vertices in one component and~$j$ vertices in the other component of the bipartition of~$K_{i, j}$.

\EP
{Perfect Code}
{A graph~$G$ and an integer~$\kappa$.}
{Does~$G$ have a perfect code of size~$\kappa$?}

\EP
{Clique}
{A graph~$G$ and an integer~$\kappa$.}
{Does~$G$ have a clique of size~$\kappa$?}

\EP
{Biclique}
{A bipartite graph~$G$ and an integer~$\kappa$.}
{Does~$G$ contain a complete bipartite graph $K_{\kappa,\kappa}$ as a subgraph?}

\EP
{Red-Blue Dominating Set (RBDS)}
{A bipartite graph~$G$ with the vertex bipartition $(R, B)$ and an integer~$\kappa$.}
{Is there a subset $B'\subseteq B$ such that $\abs{B'}= \kappa$ and~$B'$ dominates~$R$ in~$G$?}

In the above definition, we call vertices in~$R$ red vertices and call those in~$B$ blue vertices.

It is known that {\prob{Perfect Code}}, {\prob{Clique}}, {\prob{Biclique}}, and {\prob{RBDS}} are {\npc} and, moreover, with respect to~$\kappa$, {\prob{Perfect Code}}, \prob{Clique}, and  {\prob{Biclique}}\footnote{{\prob{Biclique}} is originally defined over general graphs, but we use its restricted version to bipartite graphs to establish several of our {\wahns} results. Whether {\prob{Biclique}} is {\wah} had remained as a significant long-standing open question until Lin~\cite{DBLP:conf/soda/Lin15} resolved it in the affirmative in 2015. More importantly, Lin's reduction~\cite[Theorem~1.1, Corollary~1.3]{DBLP:conf/soda/Lin15} applies to the special case where the input graph is a bipartite graph.} are
 {\wac} (see, e.g., the works of Cesati~\cite{DBLP:journals/ipl/Cesati02}, Downey and Fellows~\cite{DBLP:journals/tcs/DowneyF95}, Lin~\cite{DBLP:conf/soda/Lin15}), and {\prob{RBDS}} is {\wbc} (see, e.g., the works of Downey, Fellows, and Stege~\cite{fellows2001}, and Garey and Johnson~\cite{garey}).

We note that in several of our reductions where the {\prob{RBDS}} problem is used, we make the following assumptions:
\begin{enumerate}
    \item[(1)] $\abs{B} > \kappa > 1$, 
    \item[(2)] the graph $G$ does not contain any isolated vertices (i.e., vertices without any neighbors),
    \item[(3)] all vertices in~$R$ have the same degree, which we denote~$\ell$, and
    \item[(4)] $\ell + \kappa \leq \abs{B}$.
\end{enumerate}
These assumptions do not affect the {\wbhns} of the problem with respect to $\kappa$. The first assumption, $\abs{B} > \kappa > 1$, is straightforward, as otherwise the problem is polynomial-time solvable. The second assumption is also easy to verify. If an instance does not satisfy the third assumption, we can obtain an equivalent instance by the following operation: Let $\ell$ be the maximum degree of vertices in~$R$. For each red vertex $r \in R$ whose degree is strictly smaller than~$\ell$, we add new degree-$1$ vertices adjacent only to~$r$ until $r$ reaches degree~$\ell$. An important observation is that there is an optimal solution (a subset $B' \subseteq B$ dominating~$R$ with the minimum cardinality) of the new instance that does not contain any of the newly introduced degree-$1$ vertices. This ensures the equivalency of the two instances. Lastly, note that if $\ell + \kappa > \abs{B}$, the given instance must be a {\yesins} (any subset of~$\kappa$ blue vertices dominates all red vertices).

\subsection{Some General Notations Used in Our Reductions}
In our hardness reductions, we will consistently use the following notations. For a set~$X$, we use $\overrightarrow{X}$ to denote an arbitrary but fixed linear order on~$X$, unless stated otherwise. For a subset $X'\subseteq X$, $\overrightarrow{X}[X']$ denotes the order~$\overrightarrow{X}$ restricted to~$X'$. In addition, $\overrightarrow{X}\setminus X'$ is $\overrightarrow{X}$ without the elements of~$X'$. For two linear orders $\overrightarrow{X}$ and $\overrightarrow{Y}$, the notation $(\overrightarrow{X}, \overrightarrow{Y})$ represents their concatenation, where $\overrightarrow{X}$ precedes $\overrightarrow{Y}$. For example, if $\overrightarrow{X} = (a, b, c)$ and $\overrightarrow{Y} = (1, 2, 3)$, then $(\overrightarrow{X}, \overrightarrow{Y}) = (a, b, c, 1, 2, 3)$.

\subsection{Some Remarks}
All the parameterized reductions derived in the paper run in polynomial time. Since all eight standard control problems under the $h$-amendment and the successive procedures are in~{\np}, in the following sections, if we prove that a problem is {\wah} or {\wbh}, it implies that the problem is also {\npc}. Particularly, for all our new {\wahns} or {\wbhns} results shown in Table~\ref{tab-resulst-summary}, the corresponding problems are {\npc}. To avoid repetitiveness, we will not assert these implied results in the corresponding theorems. 

\section{Our Results}
\label{sec-main-results}
In this section, we study the complexity of the eight standard election control problems. 
%In particular, we explore how the position of the distinguished candidate in the given agenda shapes the parameterized complexity landscape of the problems.
Our exploration starts with the following lemma.

\begin{lemma}
\label{lem-a}
Let $E = (C, V)$ be an election, and let~$\rhd$ be an agenda on~$C$.  
For all positive integers~$h$, $\rhd[1]$ is the {\gamend} winner of $E$ with respect to~$\rhd$ if and only if~$\rhd[1]$ is the Condorcet winner of $E$.  
\end{lemma}

\begin{proof}
    Let $E = (C, V)$ and $\rhd$ be as stipulated in the lemma.  
    Suppose $\rhd[1]$ is the Condorcet winner of $E$, i.e., $\rhd[1]$ beats every other candidate with respect to~$V$. Consequently, in each round of the {\gamend} procedure, all candidates ranked after $\rhd[1]$ (up to $h$ in number) are eliminated. This ensures that $\rhd[1]$ remains the {\gamend} winner of $E$ with respect to~$\rhd$.  

    Conversely, assume that $\rhd[1]$ is the {\gamend} winner of $E$. By the definition of the {\gamend} procedure, $\rhd[1]$ is not eliminated in any round. This implies that $\rhd[1]$ beats every other candidate with respect to~$V$, establishing that $\rhd[1]$ is the Condorcet winner of $E$.  
\end{proof}

It is known that {\prob{CCAV}}-Condorcet and {\prob{CCDV}}-Condorcet are {\npc} (see the work of Bartholdi, Tovey, and Trick~\cite{Bartholdi92howhard}). Moreover, {\prob{CCAV}}-Condorcet is {\wah} with respect to the number of registered votes plus the number of added votes, {\prob{CCDV}}-Condorcet is {\wah} with respect to the number of deleted votes (see the work of Liu~et~al.~\cite{DBLP:journals/tcs/LiuFZL09}), {\prob{CCDV}}-Condorcet is {\wbh} with respect to the number of votes not deleted (see the work of Liu and Zhu~\cite{DBLP:journals/tcs/LiuZ13})\footnote{Although not explicitly stated in the paper of Liu and Zhu~\cite{DBLP:journals/tcs/LiuZ13} that {\prob{CCDV}}-Condorcet is {\wbh} with respect to the number of votes not deleted, the reduction in the proof of Theorem~1 in~\cite{DBLP:journals/tcs/LiuZ13} for a {\wbhns} of {\prob{CCDV}}-Maximin can be directly used to show the result. Particularly, in the proof of Theorem~1 in~\cite{DBLP:journals/tcs/LiuZ13}, the authors established an instance of {\prob{CCDV}}-Maximin from the {\wbh} problem {\prob{Dominating Set}}. They showed that if the instance of {\prob{Dominating Set}} is a {\yesins}, there are $2\kappa+1$ votes in the constructed election with respect to which the distinguished candidate beats all the other candidates, where~$\kappa$ is the parameter of the {\prob{Dominating Set}} instance. For the other direction, they showed that if there are $2\kappa+1$ votes with respect to which the distinguished candidate is the Maximin winner, the Maximin winner is the Condorcet winner.}. These results and Lemma~\ref{lem-a} offer us the following corollary.

\begin{corollary}[\cite{DBLP:journals/tcs/LiuFZL09,DBLP:journals/tcs/LiuZ13}]
\label{cor-many-np}
For any {\gamend} procedure $\tau$, the following hold:
\begin{itemize}
\item {\prob{CCAV}}-$\tau$ is {\memph\wah} with respect to~$\abs{V} + k$, the number of registered votes plus the maximum number of unregistered votes allowed to be added.
\item {\prob{CCDV}}-$\tau$ is {\memph\wah} with respect to~$k$, the maximum number of votes that can be deleted.
\item {\prob{CCDV}}-$\tau$ is {\memph\wbh} with respect to~$\abs{V} - k$, the minimum number of votes that cannot be deleted.
\end{itemize}
Moreover, the above results hold even when the distinguished candidate is ranked first on the agenda.
\end{corollary}

In addition, as {\prob{DCAV}-Condorcet and {\prob{DCDV}}-Condorcet are polynomial-time solvable (see the work of Bartholdi, Tovey, and Trick~\cite{Bartholdi92howhard}), we can deduce the following corollary concerning destructive control by adding or deleting voters.

\begin{corollary}[\cite{Bartholdi92howhard}]
\label{cor-dcav-dcdv-amd-suc-p}
For $\tau$ being any {\gamend} procedure, {\prob{DCAV}}-$\tau$ and {\prob{DCDV}}-$\tau$ are polynomial-time solvable when the distinguished candidate is the first one in the agenda.
\end{corollary}

Regarding candidate control, it is known that Condorcet is immune to {\prob{CCAC}} and {\prob{DCDC}}, and {\prob{CCDC}}-Condorcet and {\prob{DCAC}}-Condorcet are polynomial-time solvable (see the work of Bartholdi, Tovey, and Trick~\cite{Bartholdi92howhard}).
The following results follow.

\begin{corollary}[\cite{Bartholdi92howhard}]
\label{cor-ccdc-dcac-amd-suc-immune}
Let~$\tau$ be an $h$-amendment procedure. If the distinguished candidate is the first one in the agenda, then~{$\tau$} is immune to {\prob{CCAC}} and {\prob{DCDC}}, and {\prob{CCDC}}-$\tau$ and {\prob{DCAC}}-$\tau$ are polynomial-time solvable.
\end{corollary}

In the following subsections, we show  that some of the {\poly}-results in Corollary~\ref{cor-ccdc-dcac-amd-suc-immune} hold irrespective of the position of the distinguished candidate, while some others, although generally intractable, are {\fpt} when parameterized by the number of predecessors of the distinguished candidate.

We note that throughout the paper, for all results regarding {\prob{CCAC}}, the position of the distinguished candidate is counted with respect to candidates in $C\cup D$.

The remainder of this section focuses on cases not covered by the above corollaries. We begin by examining the $h$-amendment procedures and the $(m-h)$-amendment procedures (Section~\ref{sec-amd-h} and Section~\ref{sec-amd-m-h}). Following this, we explore the successive procedure (Section~\ref{sec-suc}).  
Within each section, we first consider control by adding or deleting voters and then study control by adding or deleting candidates. We note that $h$ is assumed to be positive and integral. 
For readability, this assumption will not be repeated throughout the discussion.  

\subsection{The h-Amendment Procedure for Constant~h}
\label{sec-amd-h}
By Corollary~\ref{cor-many-np}, {\prob{CCAV}}-$h$-Amendment is intractable if the distinguished candidate is the first one in the agenda. 
In the following, we show a similar result for the case where the distinguished candidate is the last one in the agenda.

\begin{theorem}
\label{thm-ccav-amd-np}
For every constant~$h$, {\prob{CCAV}}-$h$-Amendment is {\memph\wbh} with respect to the number of added votes plus the number of registered votes. This holds even if the distinguished candidate is the last candidate in the agenda.
\end{theorem}

\begin{proof}
We first prove the theorem for the case where $h = 1$ (i.e., the amendment procedure) using a reduction from {\prob{RBDS}}. Subsequently, we demonstrate how to modify the reduction for other constant values of~$h$.  

Let $(G, \kappa)$ be an instance of {\prob{RBDS}}, where $G$ is a bipartite graph with bipartition $(R, B)$.  
We construct an instance of {\prob{CCAV}}-Amendment as follows.   
For each vertex in~$R$, we create a corresponding candidate, denoted by the same symbol for simplicity. Additionally, we introduce two candidates,~$q$ and~$p$, designating~$p$ as the distinguished candidate. Let $C = R \cup \{p, q\}$, and let $m = \abs{C}$ denote the total number of candidates.  
We create $\kappa+3$ registered votes as follows:
\begin{itemize}
\item $\kappa+1$ votes with the preference $q \Succ \overrightarrow{R} \Succ p$;
\item one vote with the preference $p\Succ \overrightarrow{R} \Succ q$; and
\item one vote with the preference $\overrightarrow{R} \Succ p\Succ q$.
\end{itemize}
Let~$V$ be the multiset of the  registered votes created above. 
Let $\rhd = (q, \overrightarrow{R}, p)$ be an agenda, where $q$ and $p$ are the first and last candidates, respectively, and $\overrightarrow{R}$ represents the ordered set of candidates between $q$ and $p$. 
Figure~\ref{fig-amd-ccav-hard} depicts the weighted majority graph of $(C, V)$ and the agenda.

\begin{figure}
    \centering
    \includegraphics[width=0.4\textwidth]{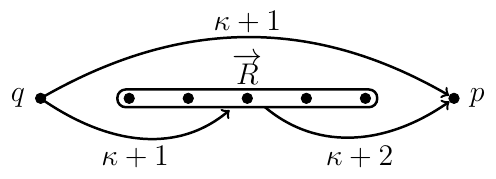}
    \caption{The weighted majority graph of $(C, V)$  and the agenda as used in the proof of Theorem~\ref{thm-ccav-amd-np}. All arcs among vertices in $R$ are forward arcs with a uniform weight of $\kappa+3$. The agenda is represented by the left-to-right ordering of the vertices.}
    \label{fig-amd-ccav-hard}
\end{figure}

The unregistered votes are created according to the blue vertices. Precisely, for each blue vertex $b\in B$, we create one vote~$\succ_b$ with the preference
\[\left(\overrightarrow{R} \setminus N_G(b)\right) \Succ p\Succ q \Succ \left(\overrightarrow{R}[N_G(b)]\right).\]
Let $W=\{\succ_b\, \setmid b\in B\}$. 
Let $k=\kappa$. The instance of \prob{CCAV}-Amendment is $(C, p, V,  W, \rhd, k)$.  
We prove the correctness of the reduction as follows.

$(\Rightarrow)$ Assume that there exists a subset $B'\subseteq B$ such that $\abs{B'}=\kappa$ and~$B'$ dominates~$R$ in~$G$. Let $W'=\{\succ_b\, \setmid b\in B'\}$ be the multiset of the~$\kappa$ unregistered votes corresponding to~$B'$. Let $\elec=(C, V\muplus W')$. We show that~$p$ is the amendment winner of~$\elec$ with respect to~$\rhd$. First, as~$B'$ dominates~$R$, from the above construction, for every candidate $r\in R$, the vote~$\succ_b$, where $b\in B'$ and~$b$ dominates~$r$, ranks~$q$ before~$r$. Therefore, in~$\elec$, at least $\kappa+2$ votes rank~$q$ before~$r$. This means that~$q$ beats all candidates from~$R$ in~$\elec$. As~$q$ is the first candidate of the agenda,~$q$ is the winner of the $(m-1)^{\text{th}}$ round. Second, as all unregistered votes rank~$p$ before~$q$,~$p$ beats~$q$ in~$\elec$. Consequently,~$p$ is the winner of the last round and hence is the amendment winner of~$\elec$ with respect to~$\rhd$.

$(\Leftarrow)$ Assume that there is a subset $W'\subseteq W$ of at most~$k$ unregistered votes such that~$p$ wins $\elec=(C, V\muplus W')$ with respect to~$\rhd$. Observe first that $\abs{W'}\geq k-1$, since otherwise~$q$ still beats all the other candidates and hence remains as the winner in~$\elec$, a contradiction. Second, observe that no matter which at most~$k$ votes are contained in~$W'$, every red vertex beats all her successors in~$\elec$. Let $B'=\{b\in B\setmid\ \succ_b\in W'\}$ be the set of blue vertices corresponding to~$W'$. We claim that~$B'$ dominates~$R$ in~$G$. Assume, for the sake of contradiction, that there exists a red vertex~$r\in R$ not dominated by any vertex from~$B'$. From the construction of the unregistered votes, it follows that all votes in~$W'$ rank~$r$ before~$q$. Consequently, there are exactly $2 + \abs{W'}$ votes in $V \cup W'$ that rank~$r$ before~$q$, which implies that~$r$  beats or ties with~$q$ in~$\elec$. Let~$r'$ denote the leftmost red vertex in the agenda that beats or ties with~$q$ in~$\elec$. Since~$r'$ beats all of its successors,~$r'$ wins~$\elec$. However, this contradicts the assumption that~$p$ is the winner of~$\elec$ with respect to~$\rhd$. Therefore,~$B'$ dominates~$R$ in~$G$.

Now we show how to modify the above reduction for the case where $h>1$. For each candidate corresponding to a red vertex in~$R$, if we make $h-1$ copies of this candidate and replace the candidate with a fixed order of these newly created candidates in the above votes, the reduction works for other values of~$h$. Precisely, for each red vertex $r\in R$, we create an ordered block of~$h$ candidates, and we replace the candidate created for~$r$ with this ordered block in all created votes. 
By using similar arguments, one can check that the {\prob{RBDS}} instance is a {\yesins} if and only if we can make~$p$ the $h$-amendment winner by adding at most~$k$ unregistered votes.
\end{proof}

For constructive control by deleting voters, we have a similar result.

\begin{theorem}
\label{thm-ccdv-amd-np}
For every constant~$h$, {\prob{CCDV}}-$h$-Amendment is {\emph\wbh} with respect to the number of deleted votes. This holds even if the distinguished candidate is the last candidate in the agenda.
\end{theorem}

\begin{proof}
Similar to the proof of Theorem~\ref{thm-ccav-amd-np}, we provide a detailed proof for the case where $h = 1$, based on a reduction from {\prob{RBDS}}. We then describe how to adapt this reduction for other values of~$h$.

Let $(G, \kappa)$ be an instance of {\prob{RBDS}}, where $G$ is a bipartite graph with the bipartition $(R, B)$. We make the following assumptions which do not change the {\wbhns} of {\prob{RBDS}}, as discussed in Subsection~\ref{subsec-supporting-problems}. First, we assume that $\abs{B}> \kappa>1$. Second, we assume that~$G$ does not contain any isolated vertices. Third, we assume that all vertices of~$R$ have the same degree~$\ell$ for some positive integer $\ell\leq \abs{B}$. Lastly, we assume that $\ell+\kappa\leq \abs{B}$.   
We construct an instance $(C, p, V, \rhd, k)$ of {\prob{CCDV}}-Amendment as follows.

The candidate set and the agenda are exactly the same as in the proof of Theorem~\ref{thm-ccav-amd-np}, i.e., $C=R\cup \{p, q\}$ and~$\rhd=(q, \overrightarrow{R}, p)$. We create $2\abs{B}-\kappa+1$ votes in~$V$:
\begin{itemize}
\item $\abs{B}-\kappa+1-\ell$ votes with the preference $\overrightarrow{R} \Succ p\Succ q$;
\item $\ell$ votes with the preference $p\Succ q\Succ \overrightarrow{R}$; and
\item for each $b\in B$, one vote~$\succ_b$ with the preference
\[\left(\overrightarrow{R}[N_G(b)]\right) \Succ q\Succ \left(\overrightarrow{R} \setminus N_{G}(b)\right) \Succ p.\]
\end{itemize} 
Let $V_B=\{\succ_b\, \setmid b\in B\}$. 
Figure~\ref{fig-amd-ccdv-wbh-k} depicts the weighted majority graph of $(C, V)$ and the agenda. 
\begin{figure}
    \centering
    \includegraphics[width=0.4\textwidth]{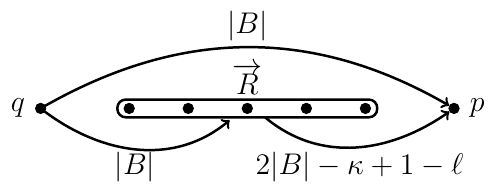}
    \caption{The weighted majority graph of $(C, V)$ and the agenda as used in the proof of Theorem~\ref{thm-ccdv-amd-np}. All arcs among the vertices in $R$ are forward arcs with a weight of at least $\abs{B} - \kappa + 1$. The agenda is represented by the left-to-right ordering of the vertices.}
    \label{fig-amd-ccdv-wbh-k}
\end{figure}
Let $k=\kappa$. One can check that under the assumptions $\kappa>1$ and $\ell+\kappa\leq \abs{B}$, all candidates beat~$p$, implying that~$p$ cannot be the amendment winner of~$(C,V)$. 
We show the correctness of the reduction as follows.

$(\Rightarrow)$ Suppose that there is a subset $B'\subseteq B$ of cardinality~$\kappa$ which dominates~$R$ in~$G$. Let $V'=\{\succ_b\, \setmid b\in B'\}$, and let $\elec=(C, V\setminus V')$. We show below that in~$\elec$,~$q$ beats all candidates from~$R$, but is beaten by~$p$, and hence~$p$ becomes the winner in~$\elec$. First, let~$r$ be any arbitrary candidate of~$R$. In~$V$ there are exactly $\ell+(\abs{B}-\ell)=\abs{B}$ votes ranking~$q$ before~$r$. As~$B'$ dominates~$R$, there exists at least one $b\in B'$ such that~$b$ dominates~$r$. According to the construction of the votes, the $\succ_b$, which is from $V'$, ranks~$r$ before~$q$, implying that there are at most $\kappa-1$ votes in~$V'$ ranking~$q$ before~$r$. As a result, in $V\setminus V'$ there are at least $\abs{B}-\kappa+1$ votes ranking~$q$ before~$r$. As $\abs{V\setminus V'}=2\abs{B}-2\kappa+1$, it holds that~$q$ beats~$r$ in~$\elec$. As all votes in~$V'$ rank~$q$ before~$p$, there are exactly $\abs{B}-\kappa+1$ votes in $V\setminus V'$ ranking~$p$ before~$q$, implying that~$p$ beats~$q$ in $\elec$.

$(\Leftarrow)$ Assume that there is a $V'\subseteq V$ such that $\abs{V'}\leq \kappa$ and~$p$ wins $\elec=(C, V\setminus V')$ with respect to~$\rhd$. Observe first that as $\ell+\kappa\leq \abs{B}$, no matter which at most~$\kappa$ votes are in~$V'$, every candidate from~$R$ beats all of its successors with respect to $V\setminus V'$. Therefore, it must be that~$q$ beats all candidates of~$R$, and~$p$ beats or ties with~$q$ in~$\elec$. The latter implies that $\abs{V'}\geq \kappa-1$ and $V'\subseteq V_B$. Let $B'=\{b\in B \setmid\: \succ_b\in V'\}$ be the set of blue vertices corresponding to~$V'$. We claim that~$B'$ dominates~$R$ in~$G$. Suppose, for the sake of contradiction, that there exists $r\in R$ not dominated by any vertex from~$B'$. From the construction of the votes,~$q$ is ranked before~$r$ in all votes in~$V'$. So, there are $(\abs{B}-\kappa+1-\ell)+\ell=\abs{B}-\kappa+1$ votes ranking~$r$ before~$q$ in~$\elec$, contradicting that~$q$ beats~$r$ in~$\elec$. So, we know that~$B'$ dominates~$R$.

To prove the {\wbhns} for other values of~$h$, we modify the above reduction in the same way as we did in the proof of Theorem~\ref{thm-ccav-amd-np}. That is, we replace in each of the above created vote the candidate created for~$r$ with an ordered block of~$h$ candidates.
\end{proof}

Now we show that the {\wbhns} of {\prob{CCDV}}-$h$-Amendment remains with respect to the dual parameter of the solution size. 

\begin{theorem}
\label{thm-ccdv-amd-wbh-remaining-votes}
For every constant~$h$, {\prob{CCDV}}-$h$-Amendment is {\emph\wbh} with respect to the number of votes not deleted. This holds even if the distinguished candidate is the last one in the agenda.
\end{theorem}

\begin{proof}
We prove the theorem via reductions from {\prob{RBDS}}. We consider first the case where $h=1$. Let $(G, \kappa)$ be an instance of {\prob{RBDS}}, where~$G$ is a bipartite graph with the bipartition $(R, B)$. We construct an instance of {\prob{CCDV}}-Amendment as follows. Let $C=R\cup \{q, q', p\}$, and let $\rhd=(q, q', \overrightarrow{R}, p)$. We create the following $\abs{B}+\kappa+1$ votes:
\begin{itemize}
\item a multiset~$V_1$ of~$\kappa$ votes, each with the preference $p\Succ q\Succ \overrightarrow{R}\Succ q'$;
\item a singleton~$V_2$ of one vote with the preference $\overrightarrow{R}\Succ p\Succ q\Succ q'$; and
\item for each blue vertex $b\in B$ one vote~$\succ_b$ with the preference
\[q'\Succ \left(\overrightarrow{R}\setminus N_G(b)\right) \Succ q\Succ \left(\overrightarrow{R}[N_G(b)]\right)\Succ p.\]
\end{itemize}
For a given $B'\subseteq B$, let $V_{B'}=\{\succ_b\, \setmid b\in B'\}$ be the multiset of votes created for vertices in~$B'$. Let $V=V_1\muplus V_2\muplus V_B$. In total, we have $\abs{B}+\kappa+1$ votes. Let $k=\abs{B}-\kappa$. The instance of {\prob{CCDV}}-Amendment is $(C, p, V, \rhd, k)$. Observe that after deleting at most~$k$ votes from the election~$(C, V)$, at least $2\kappa + 1$ votes remain. In the following, we establish the correctness of the reduction.

$(\Rightarrow)$ Assume that there is a $B'\subseteq B$ of cardinality~$\kappa$ such that~$B'$ dominates~$R$ in~$G$. Let $V'=V_1\muplus V_2\muplus V_{B'}$, and let $\elec=(C, V')$. Recall that $\abs{V_1}+\abs{V_2}=\kappa+1$. We claim that~$p$ is the amendment winner of~$\elec$ with respect to~$\rhd$. Clearly,~$q$ beats $q'$ in~$\elec$. For each $r\in R$, let $b\in B'$ be a vertex dominating~$r$ in~$G$. From the construction of the votes,~$q$ is ranked before~$r$ in~$\succ_b$. Therefore, there are at least $\abs{V_1}+1$ votes in~$V'$ ranking~$q$ before~$r$, implying that~$q$ beats~$r$ in $\elec$. This ensures that all candidates in $R\cup \{q'\}$ are eliminated before the last round. Furthermore, since all $\kappa+1$ votes in $V_1\cup V_2$ rank~$p$ before~$q$,~$q$ is eliminated in the last round, leaving~$p$ as the amendment winner of~$\elec$. Hence, the constructed instance of {\prob{CCDV}}-Amendment is a {\yesins}.

$(\Leftarrow)$ Suppose that there is a $V'\subseteq V$ of at least $2\kappa+1$ votes such that~$p$ is the amendment winner of~$(C, V')$ with respect to~$\rhd$. Let $\elec=(C, V')$. Observe that~$V'$ cannot contain more than $\kappa+1$ votes from~$V_B$, since otherwise~$q'$ beats all the other candidates and hence it is impossible that~$p$ is the amendment winner of~$\elec$. If~$V'$ contains exactly $\kappa+1$ votes from~$V_B$, then either~$q'$ is the amendment winner of~$\elec$ (when~$V'$ contains exactly~$\kappa$ votes from $V_1\cup V_2$), or someone in~$R$ is the amendment winner (when $(V_1\muplus V_2)\subseteq V'$), both of which contradict the winning of~$p$ in~$\elec$. Then, as $\abs{V'}\geq 2\kappa+1$, and $\abs{V_1}+\abs{V_2}=\kappa+1$, it must hold that~$V'$ contains exactly~$\kappa$ votes from~$V_B$. As a consequence,~$V'$ contains all the $\kappa+1$ votes from $V_1\muplus V_2$. Let~$V_{B'}$, $B'\subseteq B$, be the intersection of~$V'$ and~$V_{B}$. We claim that~$B'$ dominates~$R$ in~$G$. For the sake of contradiction, assume that there is a red vertex $r\in R$ not dominated by any vertex from~$B'$. Then, by the construction of the votes, all votes in~$V_{B'}$ rank~$r$ before~$q$, implying that~$r$ beats~$q$ in~$\elec$. Moreover, it is easy to see that~$q$ beats~$q'$. By the definition of the agenda~$\rhd$, this implies that the winner of the second-to-last round is from~$R$. However, as everyone in~$R$ beats~$p$ in~$\elec$, this contradicts the winning of~$p$ in~$\elec$. Since~$B'$ dominates~$R$ in~$G$, and $\abs{B'}=\abs{V_{B'}}=\kappa$, the {\prob{RBDS}} instance is a {\yesins}.

To prove the {\wbhns} for other values of~$h$, we modify the above reduction by replacing each candidate in~$R\cup \{q'\}$ with an ordered block of~$h$ candidates.
\end{proof}

Now we study destructive control by adding/deleting voters. 
Particularly, we extend the {\poly}-results for the $h$-amendment procedures in Corollary~\ref{cor-dcav-dcdv-amd-suc-p} to a unified {\fpt}-algorithm, considering the number of predecessors of the distinguished candidate in the agenda.
To achieve this, we examine an intermediate problem and formulate it using integer linear programming (ILP).

\EP
{Majority Graph Control by Editing Voters (MGCEV)}
{A set~$C$ of candidates,  two multisets~$V$ and~$W$ of votes over~$C$, an oriented graph~$G$ with the vertex set~$C$, and two integers~$k$ and~$k'$.}
{Are there $V'\subseteq V$ and $W'\subseteq W$ such that $\abs{V'}=k$, $\abs{W'}=k'$, and the majority graph of $(C, V'\muplus W')$ is exactly~$G$?}

The following lemma is well-known.

\begin{lemma}[\cite{{Frank1987,kannan87a,DBLP:journals/mor/Lenstra83}}]
\label{lem-ilp-fpt}
{\memph{ILP}} is {\memph\fpt} with respect to the number of variables.
\end{lemma}

Based on Lemma~\ref{lem-ilp-fpt}, we can show the following result.

\begin{theorem}
\label{thm-mgcav-fpt}
{\prob{MGCEV}} is  {\memph\fpt} with respect to the number of candidates.
\end{theorem}

\begin{proof}
Let $I=(C, V, W, G, k, k')$ be an instance of {\prob{MGCEV}}. Let $m=\abs{C}$ denote the number of candidates. Let~$\mathcal{L}$ be the set of all linear orders on~$C$. It holds that $\abs{\mathcal{L}}=m!$. For each $\succ\in \mathcal{L}$, we use $n_V^{\succ}$ (respectively, $n_W^{\succ}$) to denote the number of votes in~$V$ (respectively,~$W$) with the preference~$\succ$. For each $\succ \in \mathcal{L}$, we introduce two nonnegative integer variables, $x_{\succ}$ and $y_{\succ}$, which represent the number of votes in $V$ and $W$, respectively, with the preference $\succ$ that are included in a desired feasible solution. The constraints are as follows.
\begin{itemize}
\item As we aim to identify two multisets $V'\subseteq V$ and $W'\subseteq W$, with cardinalities~$k$ and~$k'$, respectively, we have the following constraints: $\sum_{\succ\in \mathcal{L}}x_{\succ}=k$ and $\sum_{\succ\in \mathcal{L}} y_{\succ}=k'$.
\item To ensure that the majority graph of $(C, V'\muplus W')$ is exactly~$G$, we create the following two classes of constraints:
\begin{itemize}
\item for every arc from a candidate~$c$ to another candidate~$c'$ in the graph~$G$, we have that
\[
\sum_{\succ \in \mathcal{L} : c \succ c'} x_{\succ} + \sum_{\succ \in \mathcal{L} : c \succ c'} y_{\succ} > \frac{k + k'}{2},
\]
\item for every two candidates~$c$ and~$c'$ in~$G$ without an arc between them, we require that
\[\sum_{\succ\in \mathcal{L} : c\succ c'}x_{\succ}+\sum_{\succ\in\mathcal{L} : c\succ c'}y_{\succ}=\frac{k+k'}{2}.\]
\end{itemize}
\item Lastly, for each variable~$x_{\succ}$ and~$y_{\succ}$, we have that $n_V^{\succ}\geq x_{\succ}\geq 0$ and $n_W^{\succ}\geq y_{\succ}\geq 0$.
\end{itemize}
This completes the construction of the ILP instance. Clearly, the given instance~$I$ is a {\yesins} if and only if the above ILP has a feasible solution. The theorem follows from Lemma~\ref{lem-ilp-fpt} and the fact that the number of variables of the above ILP is bounded from above by a function in~$m$.
\end{proof}

Before presenting our parameterized results concerning the number of predecessors of the distinguished candidate, we note that Theorem~\ref{thm-mgcav-fpt} can be straightforwardly leveraged to derive {\fpt}-algorithms for the four standard voter control problems under the $h$-amendment procedures and the successive procedure, with respect to the number~$m$ of candidates. To verify this, first observe that the majority graph of an election and an agenda suffices to determine the winner. Based on this observation, solving a voter control problem can be reduced to enumerating all oriented graphs over the candidate set as vertices, such that the distinguished candidate $p$ is either the winner (in the case of constructive control) or not the winner (in the case of destructive control) for the enumerated graph and the given agenda.  
For each enumerated graph, we verify whether it is possible to add or delete a given number of voters so that the majority graph of the resulting election matches the enumerated graph. This verification is equivalent to solving the {\prob{MGCEV}} problem. Theorem~\ref{thm-mgcav-fpt} guarantees that this can be performed in {\fpt}-time with respect to~$m$. The original instance is a {\yesins} if and only if at least one enumeration yields a {\yes}-instance of the {\prob{MGCEV}} problem. Since the number of oriented graphs with $m$ vertices is bounded above by a function of $m$, it follows that the voter control problems are {\fpt} when parameterized by~$m$.

Returning to our main focus, we now present stronger {\fpt}-results for the destructive voter control problems by leveraging Theorem~\ref{thm-mgcav-fpt}. Specifically, we establish the fixed-parameter tractability of a generalized problem that subsumes both the {\prob{DCAV}} problem and  the {\prob{DCDV} problem as special cases. This result is parameterized by the number of predecessors of the distinguished candidate, which is strictly smaller than~$m$.

\EP
{\prob{Exact Destructive Control by Editing Voters for~$\tau$} (\prob{E-DCEV-$\tau$})}
{A set~$C$ of candidates, a distinguished candidate $p\in C$, a multiset~$V$ of registered votes over~$C$, a multiset $W$ of unregistered votes over~$C$, an agenda~$\rhd$ on~$C$, and two nonnegative integers~$k$ and~$k'$.}
{Are there $V'\subseteq V$ and $W'\subseteq W$ such that $\abs{V'}=k$, $\abs{W'}=k'$, and~$p$ is not the~$\tau$ winner of $(C, V'\muplus W')$ with respect to~$\rhd$?}

\begin{theorem}
\label{thm-edcadv-amd-p}
{\prob{E-DCEV}}-$h$-Amendment is {\memph\fpt} with respect to~$h$ plus the number of predecessors of the distinguished candidate.
\end{theorem}

\begin{proof}
Let $I=(C, p, V, W, \rhd, k, k')$ be an instance of {\prob{E-DCEV}}-$h$-Amendment. 
Let~$\ell$ be the number of predecessors of~$p$ in the agenda~$\rhd$. Therefore, $p=\rhd[\ell+1]$. Let $\ell'=\min \{\ell+h+1, m\}$.
In the following, we derive an algorithm running in {\fpt}-time in~$\ell'$.

Observe that~$p$ is not the $h$-amendment winner of an election $(C, U)$  where $U\subseteq V\cup W$ with respect to~$\rhd$ if and only if one of the following two cases occurs:
\begin{enumerate}
\item[(1)] $p$ is not the {\gamend} winner of $(C, U)$ restricted to $\rhd[1,\ell']$ (with respect to the agenda~${\rhd}$ restricted to $\rhd[1, \ell']$).
\item[(2)] $p$ is the {\gamend} winner of $(C, U)$ restricted to $\rhd[1,\ell']$, but some candidate after~$\rhd[\ell']$ in the agenda beats or ties with~$p$ (in this case it must be that $\ell'<m$).
\end{enumerate}

For the first case, only the head-to-head comparisons among candidates in $\rhd[1, \ell']$ matter. For the second case, in addition to $\rhd[1, \ell']$, we need to guess one candidate $q\in \rhd[\ell'+1, m]$ who is supposed to beat or tie with~$p$ in the final election. Once such a candidate~$q$ is fixed, we need only to focus on the head-to-head comparisons among candidates in $\rhd[1, \ell']\cup \{q\}$.

Our algorithm proceeds as follows.
We define a configuration as an oriented graph whose vertices are a subset of candidates. 
We enumerate all possible configurations, with our focused candidates being the vertices, and check if at least one of the enumerated configurations leads to the occurrence of at least one of the above two cases. Precisely, the enumerated configurations consist of the following oriented graphs:
\begin{itemize}
\item all oriented graphs whose vertex set is $\rhd[1, \ell']$; and
\item for each $q\in C\setminus \rhd[1, \ell']$, all oriented graphs whose vertex set is $\rhd[1, \ell']\cup \{q\}$.
\end{itemize}
It is clear that there are {\fpt}-many configurations to enumerate with respect to~$\ell'$.
For each configuration~$H$, we check in polynomial time whether one of the two cases occurs, based on the definition of the $h$-amendment procedure. More precisely, we first compute the $h$-amendment winner of~$H$ with respect to the agenda~$\rhd$ restricted to only our focused candidates (corresponding to vertices of~$H$). If~$H$ has vertex set $\rhd[1, \ell']$ and the winner is not~$p$, this configuration results in Case~(1). If~$H$ has vertex set $\rhd[1, \ell']\cup \{q\}$, where~$q\in C\setminus \rhd[1, \ell']$, such that~$p$ is the $h$-amendment winner of~$H-\{q\}$ with respect to~$\rhd$, and there is no arc from~$p$ to~$q$ in~$H$,~then this configuration results in Case~(2). If a configuration results in neither Case~(1) nor Case~(2), we discard it. Otherwise, we utilize the {\fpt}-algorithm in Theorem~\ref{thm-mgcav-fpt} to check if there are desired $V'\subseteq V$ and $W'\subseteq W$ that achieve the configuration.
To be precise, let~$C'$ be the vertex set of~$H$. 
Then, we solve an instance $(C', V|_{C'}, W|_{C'}, H, k, k')$ of {\prob{MGCEV}} by Theorem~\ref{thm-mgcav-fpt}, and return the output for the configuration.

If we obtain a {\yes}-answer from at least one of the enumerated configurations, we conclude that~$I$ is a {\yesins}; otherwise we conclude that~$I$ is a {\noins}.
\end{proof}

The following result is a consequence of Theorem~\ref{thm-edcadv-amd-p}.

\begin{corollary}
\label{cor-dcav-dcdv-amd-fpt}
{\prob{DCAV}}-$h$-Amendment and {\prob{DCDV}}-$h$-Amendment are {\memph\fpt} with respect to the combined parameter of~$h$ and the number of predecessors of the distinguished candidate.
\end{corollary}

In contrast to these tractability results, we show that if the distinguished candidate moves to the last position of the agenda, both problems become intractable from the parameterized complexity point of view.

\begin{theorem}
\label{thm-dcav-amd-np}
For every constant~$h$, {\prob{DCAV}}-$h$-Amendment is {\memph\wah} with respect to the number of added votes plus the number of registered votes. This holds even if the distinguished candidate is the last one in the agenda.
\end{theorem}

\begin{proof}
We begin with a proof for the amendment procedure (i.e., when $h=1$), based on the {\prob{Perfect Code}} problem. Let $(G, \kappa)$ be an instance of the {\prob{Perfect Code}} problem, where~$G$ is a graph with vertex set $\vset$. We denote the vertex set as $\vset = \{\vere_1, \dots, \vere_n\}$. We create an instance of {\prob{DCAV}}-Amendment as follows.

For each $\vere_i \in \vset$, we create two candidates, denoted by $x_i$ and $y_i$, respectively. Let $X=\{x_i \setmid i\in [n]\}$, and let $Y=\{y_i \setmid i\in [n]\}$. In addition, we introduce two candidates,~$p$ and~$q$, where~$p$ is the distinguished candidate. Let $C=X\cup Y\cup \{p, q\}$. The agenda is~$\rhd=(q, \overrightarrow{X}, \overrightarrow{Y}, p)$. 
We create the following $3\kappa+3$ registered votes:
\begin{itemize}
\item one vote with the preference $p\Succ q\Succ \overrightarrow{X} \Succ \overrightarrow{Y}$;
\item $\kappa+2$  votes with the preference $q\Succ p\Succ \overrightarrow{X} \Succ \overrightarrow{Y}$;
\item $\kappa-2$ votes with the preference $p\Succ \overrightarrow{X}\Succ q\Succ \overrightarrow{Y}$; and
\item $\kappa+2$ votes with the preference $p\Succ \overrightarrow{X}\Succ \overrightarrow{Y} \Succ q$.
\end{itemize}
Let~$V$ denote the multiset of the above registered votes. Figure~\ref{fig-amd-dcav-hard} shows the weighted majority graph of $(C, V)$ and the agenda.
As~$p$ beats every other candidate,~$p$ is the amendment winner of~$(C, V)$.

\begin{figure}
    \centering
    \includegraphics[width=0.45\textwidth]{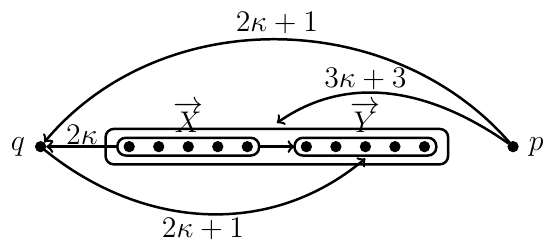}
    \caption{An illustration of the weighted majority graph of $(C, V)$ and the agenda as used in the proof of Theorem~\ref{thm-dcav-amd-np}. All arcs among the vertices in $X$ (respectively,~$Y$) are forward arcs with a uniform weight of $3\kappa + 3$. The agenda is represented by the left-to-right ordering of the vertices.}
    \label{fig-amd-dcav-hard}
\end{figure}

For each $\vere_i\in \vset$, let $N_X[\vere_i]$ (respectively,~$N_Y[\vere_i]$) be the set of candidates from~$X$ (respectively,~$Y$) corresponding to the closed neighbors of~$\vere_i$ in~$G$. Precisely, $N_X[\vere_i]=\{x_j\in X \setmid \vere_j\in N_G[\vere_i]\}$ and $N_Y[\vere_i]=\{y_j\in Y \setmid \vere_j\in N_G[\vere_i]\}$.
For each $\vere_i\in \vset$, we create one unregistered vote~$\succ_{i}$ with the preference
\[\overrightarrow{X}[N_X[\vere_i]] \Succ \left(\overrightarrow{Y}\setminus N_Y[\vere_i]\right) \Succ q\Succ p\Succ \left(\overrightarrow{X}\setminus N_X[\vere_i]\right)\Succ \overrightarrow{Y}[N_Y[\vere_i]].\]
Let $W=\{\succ_i\, \setmid \vere_i\in \vset\}$.

Let $k=\kappa$. The instance of {\prob{DCAV}}-Amendment is  $(C, p, V, W, \rhd, k)$. 
It remains to prove the correctness of the reduction.

$(\Rightarrow)$ Suppose that~$G$ has a perfect code~$S$ of size~$\kappa$. Let~$W'$ be the set of unregistered votes corresponding to~$S$. Let $\elec=(C, V\muplus W')$. We show below that~$q$ wins~$\elec$ by proving that~$q$ beats every other candidate with respect to $V\muplus W'$. First, as all unregistered votes rank~$q$ before~$p$, there are exactly $2\kappa+2$ votes in~$\elec$ ranking~$q$ before~$p$, implying that~$q$ beats~$p$ in~$\elec$. Let $x_i$ be a candidate from~$X$. Because~$S$ is a perfect code, there is exactly one $u_j\in S$ such that $x_i\in N_X[u_j]$. From the construction of the unregistered votes, we know that there are exactly $\kappa-1$ votes in~$W'$ ranking~$q$ before~$x_i$, resulting in $\kappa+3+\kappa-1=2\kappa+2$ votes ranking~$q$ before~$x_i$ in the election~$\elec$. Similarly, we can show that for every $y_i\in Y$, exactly $1+(\kappa+2)+(\kappa-2)+1=2\kappa+2$ votes ranking~$q$ before~$y_i$ in~$\elec$. from $\abs{V\muplus W'}=4\kappa+3$, it follows that~$q$ beats all candidates from $X\cup  Y$ in~$\elec$.

$(\Leftarrow)$ Suppose that $W'\subseteq W$ is a subset of~$k$ unregistered votes such that~$p$ is not the amendment winner of $\elec=(C, V\muplus W')$ with respect to~$\rhd$. Observe that no matter which at most~$\kappa$ votes are contained in~$W'$,~$p$ beats all candidates in~$X\cup Y$. This implies that~$q$ is the only candidate which is capable of preventing~$p$ from winning. In other words,~$q$ is the amendment winner of~$\elec$ with respect to~$\rhd$. It follows that~$q$ beats all the other candidates in~$\elec$. Then, observe that~$\abs{W'}=k$, since otherwise~$p$ is not beaten by~$q$ in~$\elec$. Let~$S$ be the set of vertices corresponding to~$W'$. Similar to the analysis for the other direction, as for every $y_i\in Y$ we have $2\kappa+1$ votes in~$V$ ranking~$q$ before~$y_i$, there is at least one vote in~$W'$ ranking~$q$ before~$y_i$. By the construction of the unregistered votes, we know that there is at least one $u_j\in S$ such that $y_i\in N_Y[u_j]$, and hence $u_i\in N_G[u_j]$. As~$y_i$ can be any candidate in~$Y$, this implies that~$S$ is a dominating set of~$G$. Now, we show that for every two vertices $\vere_i, \vere_{i'}\in S$, it holds that $N_G[\vere_i]\cap N_G[\vere_{i'}]=\emptyset$. Assume, for the sake of contradiction, that there are $\vere_i, \vere_{i'}\in S$ and a vertex~$\vere_{j'}$ such that $\vere_{j'}\in N_G[\vere_i]\cap N_G[\vere_{i'}]$. Then, as the two votes~$\succ_{i}$ and~$\succ_{i'}$ corresponding to~$\vere_i$ and~$\vere_{i'}$ both rank~$x_{j'}$ before~$q$, there can be at most $\kappa-2$ votes in~$W'$ ranking~$q$ before~$x_{j'}$, resulting in at most $\kappa+3+(\kappa-2)=2\kappa+1$ votes in $V\muplus W'$ ranking~$q$ before~$x_{j'}$, contradicting that~$q$ beats~$x_{j'}$ in~$\elec$. We can conclude now that~$S$ is a perfect code of~$G$.

To prove for the case where $h>1$, we replace each~$x_i$ and each~$y_i$ with an ordered block of size~$h$ in the above reduction.
\end{proof}

For control by deleting voters, we have similar results as encapsulated in the following two theorems. 

\begin{theorem}
\label{thm-dcdv-amd-np}
For every constant~$h$, {\prob{DCDV}}-$h$-Amendment is {\memph\wbh} with respect to the number of deleted votes, even if the distinguished candidate is the last one in the agenda.
\end{theorem}

\begin{proof}
We consider first the amendment procedure.
We give a reduction from {\prob{RBDS}} to {\prob{DCDV}}-Amendment. Let $(G, \kappa)$ be an instance of {\prob{RBDS}}, where $G$  is a bipartite graph with the bipartition $(R, B)$. Similar to the proof of Theorem~\ref{thm-ccdv-amd-np}, we assume that each red vertex in~$G$ has the same degree~$\ell>0$, $\abs{B}>\kappa>1$, and $\ell+\kappa\leq \abs{B}$. We create an instance  of {\prob{DCDV}}-Amendment as follows.
The candidates and the agenda are the same as in the proof of Theorem~\ref{thm-ccdv-amd-np}, i.e., $C=R\cup \{p, q\}$, where~$p$ serves as the distinguished candidate, and $\rhd=(q, \overrightarrow{R}, p)$.
We create the following $2\abs{B}+1-\kappa$ votes in~$V$:
\begin{itemize}
\item $\abs{B}+1-\ell-\kappa$ votes with the same preference $\overrightarrow{R} \Succ q\Succ p$;
\item $\ell$ votes with the same preference $q\Succ p\Succ \overrightarrow{R}$; and
\item for each blue vertex $b\in B$, one vote $\succ_b$ with the preference $p\Succ (\overrightarrow{R}[N_G(b)]) \Succ q \Succ (\overrightarrow{R}\setminus N_G(b))$.
\end{itemize}
Figure~\ref{fig-amd-dcdv-hard} illustrates the weighted majority graph of $(C, V)$ and the agenda.
\begin{figure}
    \centering
    \includegraphics[width=0.35\textwidth]{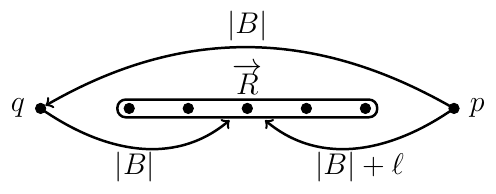}
    \caption{An illustration of the weighted majority graph of $(C, V)$ and the agenda as used in the proof of Theorem~\ref{thm-dcdv-amd-np}. All arcs among the vertices in $R$ are forward with a weight of at least $2\abs{B} - \kappa - \ell + 1$. The agenda is represented by the left-to-right ordering of the vertices.}
    \label{fig-amd-dcdv-hard}
\end{figure}
To complete the reduction, we let $k=\kappa$. The instance of {\prob{DCDV}}-Amendment is $(C, p, V, \rhd, k)$. It is easy to check that~$p$ beats everyone else, and hence is the amendment winner of $(C, V)$ with respect to~$\rhd$. We show the correctness of the reduction as follows.

$(\Rightarrow)$ Suppose that there is a subset $B'\subseteq B$ of cardinality~$\kappa$ which dominates~$R$ in~$G$. Let $V'=\{\succ_b\, \setmid b\in B'\}$ be the set of votes corresponding to~$B'$. Let $\elec=(C, V\setminus V')$. Clearly,~$\elec$ contains exactly $2\abs{B}+1-2\kappa$ votes. We show below that~$q$ is the amendment winner of~$\elec$ with respect to~$\rhd$  by proving that~$q$ beats all the other candidates. First, as all votes in~$V'$ rank~$p$ before~$q$, the number of votes in $V\setminus V'$ ranking~$q$ before~$p$ remains as $\abs{B}+1-\kappa$, implying that~$q$ beats~$p$ in~$\elec$. We consider candidates in~$R$ now. Let~$r$ be any arbitrary candidate from~$R$. As~$B'$ dominates~$R$, there is at least one vertex $b\in B'$ which is adjacent to~$r$ in~$G$. Then, the vote~$\succ_b$ corresponding to~$b$ ranks~$r$ before~$q$, meaning that there can be at most $\kappa-1$ votes in~$V'$ ranking~$q$ before~$r$. This further implies that in $V\setminus V'$, at least $\ell+(\abs{B}-\ell-(\kappa-1))=\abs{B}+1-\kappa$ votes rank~$q$ before~$r$, and hence~$q$ also beats~$r$ in~$\elec$. As this holds for all $r\in R$, we know that~$q$ beats all candidates from~$R$ in~$\elec$.

$(\Leftarrow)$ Assume that there exists a subset $V' \subseteq V$ such that $\abs{V'} \leq k$ and $p$ is not the amendment winner of the election $\elec = (C, V \setminus V')$ with respect to~$\rhd$. First, observe that regardless of which at most $k$ votes are included in~$V'$, the candidate~$p$ beats all candidates from~$R$ in~$\elec$. This implies that the amendment winner of~$\elec$ with respect to~$\rhd$ must be~$q$. In other words,~$q$ beats all other candidates in~$\elec$.  
Furthermore, note that $V'$ must consist of exactly~$k$ votes, all of which correspond to vertices in~$B$. Otherwise,~$p$ remains the winner. Thus, we have $\abs{V \setminus V'} = 2\abs{B} + 1 - 2\kappa$. Now, for every $r \in R$, since there are precisely $\ell + (\abs{B} - \ell) = \abs{B}$ votes in~$V$ ranking~$q$ before~$r$, at most $\kappa - 1$ votes from~$V'$ can rank~$q$ before~$r$. 
This ensures that for every $r \in R$, there exists at least one vote $\succ_b \in V'$ where~$r$ is ranked before~$q$. By the definition of~$\succ_b$, it follows that~$b$ dominates~$r$ in~$G$. Therefore, the subset of blue vertices corresponding to~$V'$ dominates~$R$ in~$G$.

To establish hardness for every constant $h > 1$, we replace each candidate corresponding to a red vertex with an ordered block of size~$h$ in the above-constructed election. The correctness analysis remains similar.
\end{proof}

\begin{theorem}
\label{thm-dcdv-amd-wbh-remaining}
For every constant~$h$, {\prob{DCDV}}-$h$-Amendment is {\memph\wbh} with respect to the number of remaining votes, even if the distinguished candidate is the last one in the agenda.
\end{theorem}

\begin{proof}
We prove the theorem by a reduction from {\prob{RBDS}}. Let $(G, \kappa)$ be an instance of {\prob{RBDS}}, where~$G$ is a bipartite graph with the bipartition $(R, B)$. Without loss of generality, we assume that $\kappa<\abs{B}$. Let $C=R\cup \{q, p\}$, with~$p$ as the distinguished candidate, and let $\rhd=(q, \overrightarrow{R}, p)$.
We create the following votes:
\begin{itemize}
\item a multiset~$V_1$ of $\kappa$ votes with the preference $q\Succ p\Succ \overrightarrow{R}$;
\item a singleton~$V_2$ of one vote with the preference $\overrightarrow{R}\Succ q\Succ p$; and
\item for each blue vertex $b\in B$, one vote $\succ_b$ with the preference
\[p\Succ \left(\overrightarrow{R}\setminus N_G(b)\right) \Succ q\Succ \left(\overrightarrow{R}[N_G(b)]\right).\]
\end{itemize}
For a given $B'\subseteq B$, we use $V_{B'}=\{\succ_b\, \setmid b\in B'\}$ to denote the multiset of votes created for the blue vertices in~$B'$. Let $V=V_1\muplus V_2\muplus V_B$. It holds that $\abs{V}=\abs{B}+\kappa+1$. Let $k=\abs{B}-\kappa$. The instance of {\prob{DCDV}}-$h$-Amendment is $(C, p, V, \rhd, k)$.
Obviously, under the assumption that $\kappa<\abs{B}$, the $h$-amendment winner of $(C, V)$ with respect to~$\rhd$ is~$p$.
It remains to show the correctness of the reduction.

$(\Rightarrow)$ Assume that there is a subset $B'\subseteq B$ such that~$\abs{B'}=\kappa$ and~$B'$ dominates~$R$ in~$G$. Let $V'=V_1\muplus  V_2\muplus V_{B'}$, and let $\elec=(C, V')$. Obviously, $\abs{V'}=2\kappa+1$. We claim that~$q$ is the $h$-amendment winner of~$\elec$ with respect to the agenda~$\rhd$. As~$q$ is the first candidate in~$\rhd$, we need to show that~$q$ beats all the other candidates in~$\elec$. It is clear that~$q$ beats~$p$ in~$\elec$. For each red vertex $r\in R$, there exists at least one $b\in B'$ dominating~$r$ in~$G$. From the construction of the votes, we know that~$q$ is ranked before~$r$ in the vote~$\succ_b$. Therefore, there are in total at least $\abs{V_1}+1=\kappa+1$ votes ranking~$q$ before~$r$ in~$V'$, implying that~$q$ beats~$r$. As this holds for all $r\in R$, the correctness for this direction follows.

$(\Leftarrow)$ Assume that there is a subset $V'\subseteq V$ of at least $2\kappa+1$ votes such that~$p$ is not the $h$-amendment winner of $(C, V')$ with respect to~$\rhd$. Observe that as $\abs{V_1}+\abs{V_2}=\kappa+1$ and all votes in~$V_B$ rank~$p$ in the first place, it must be that $(V_1\muplus V_2)\subseteq V'$, and~$V'$ contains exactly~$\kappa$ votes from~$V_B$ (since otherwise~$p$ remains as the $h$-amendment winner of $(C, V')$). Let $V_{B'}=V'\cap V_B$, where $B'\subseteq B$. As just discussed, it holds that $\abs{V_{B'}}=\kappa$. We also observe that~$p$ beats all candidates in~$R$ in $(C, V')$ no matter which~$\kappa$ votes from~$V_B$ are contained in~$V'$. Therefore, it must be that~$q$ is the $h$-amendment winner of $(C, V')$, i.e.,~$q$ beats all the other candidates. We claim that~$B'$ dominates~$R$. Suppose, for the sake of contradiction, that there exists at least one red vertex $r\in R$ not dominated by any vertex from~$B'$ in~$G$. Then, from the construction of the votes,~$r$ is ranked before~$q$ in all votes of~$V_{B'}$. Together with the vote in~$V_2$, there are $\kappa+1$ votes in~$V'$ ranking~$r$ before~$q$, implying that~$r$ beats~$q$ in~$(C, V')$. However, this contradicts that~$q$ is the $h$-amendment winner of $(C, V')$ with respect to~$\rhd$.
\end{proof}

We now investigate election control through the addition or deletion of candidates. Unlike voter control operations, we demonstrate that candidate control operations generally remain solvable in polynomial time, regardless of the position of the distinguished candidate in the agenda. These results extend the findings for the amendment procedure stated in Corollary~\ref{cor-ccdc-dcac-amd-suc-immune}.

We first consider {\prob{CCAC}-Amendment}.
We need the following notions for our study.

Let $\elec=(C, V)$ be an election, and let~$\rhd=(c_1, c_2, \dots, c_m)$ be an agenda on~$C$. For a subset $C'\subseteq C$ and two candidates $c, c'\in C'$ such that $c\rhd c'$, a $(c\leftarrow c')$-beating path with respect to $(C', V)$ and~$\rhd$ is a sequence of candidates $(c_{\pi(0)}, c_{\pi(1)}, \dots, c_{\pi(t)})$, where $\{\pi(0), \dots, \pi(t)\} \subseteq [m]$, satisfying the following conditions:
\begin{itemize}
\item $c_{\pi(i)}\in C'$ for all $i\in [t]\cup \{0\}$, i.e., all candidates in the $(c\leftarrow c')$-beating path belong to~$C'$;
\item $c_{\pi(0)}=c'$ and $c_{\pi(t)}=c$; 
\item $\pi(i-1)>\pi(i)$ for all $i\in [t]$, i.e., $c_{\pi(t)}\rhd c_{\pi(t-1)}\rhd \cdots \rhd c_{\pi(0)}$; 
\item for all $i\in [t]$, $c_{\pi(i)}$ is either beaten by or ties with~$c_{\pi(i-1)}$; and
\item for all $i\in [t]$, $c_{\pi(i)}$ beats all of its successors from $C'$ that lie between $c_{\pi(i)}$ and $c_{\pi(i-1)}$ in the agenda $\rhd$, i.e., $c_{\pi(i)}$ beats~$\rhd[j]$ for all $\rhd[j]\in C'$ such that $\pi(i)<j<\pi(i-1)$.
\end{itemize}
In particular, for any candidate $c \in C$, we define the trivial $ (c \leftarrow c)$-beating path as the sequence consisting solely of $(c)$. 
We refer to~Figure~\ref{fig-beating-path} for an illustration. 

\begin{figure}
    \centering
    \includegraphics[width=0.8\textwidth]{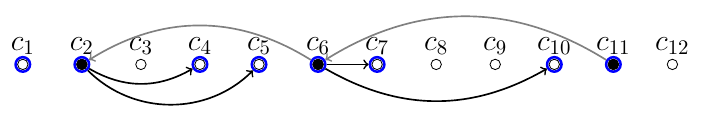}
    \caption{The sequence $(c_{11}, c_6, c_2)$ (highlighted as dark vertices) represents a $(c_{2} \leftarrow c_{11})$-beating path with respect to the subset $C' = \{c_1, c_2, c_4, c_5, c_6, c_7, c_{10}, c_{11}\}$ (vertices enclosed in blue circles), the election $(C', V)$, and the agenda $\rhd = (c_1, c_2, \dots, c_{12})$. A dark arc $\arc{a}{b}$ signifies that $a$ beats $b$, while a gray arc $\arc{a}{b}$ indicates that $a$ either beats or ties with $b$.}
    \label{fig-beating-path}
\end{figure}

Checking whether there is a $(c\leftarrow c')$-beating path with respect to $(C', V)$ and~$\rhd$ can be clearly done in polynomial time. 
Moreover, a candidate~$c$ is the amendment winner of $(C,V)$ with respect to~$\rhd$ if and only if~$c$ beats all her successors, and there is a~$(\rhd[1] \leftarrow c)$-beating path with respect to $(C, V)$ and~$\rhd$. 

We are now prepared to present our algorithm for {\prob{CCAC}-Amendment}.

\begin{theorem}
\label{thm-ccac-amd-p}
{\prob{CCAC}-Amendment} is polynomial-time solvable.
\end{theorem}

\begin{proof}
We derive a dynamic programming algorithm for the {\prob{CCAC}-Amendment} problem as follows. Let $I=(C,  p, D, V, \rhd, k)$ be an instance of {\prob{CCAC}-Amendment}, where~$C$ and~$D$ are respectively the set of registered candidates and the set of unregistered candidates, $p\in C$ is the distinguished candidate,~$V$ is a multiset of votes,~$\rhd$ is an agenda, and~$k$ is the number of unregistered candidates that we can add. Let~$c^{\star}$ be the leftmost candidate from~$C$ in the agenda~$\rhd$. That is, $c^{\star}\rhd c$ holds for all $c\in C\setminus \{c^{\star}\}$. If~$p$ does not beat all her successors in~$C$, we directly conclude that~$I$ is a {\noins}. So, let us assume that~$p$ beats all her successors from~$C$ in~$\rhd$. If~$p$ is already the amendment winner,~$I$ is a {\yesins}; we are done. Otherwise, 
%a $(c^{\star} \leftarrow p)$-beating path with respect to $(C, V)$ and~$\rhd$ does not exist. The question now is
we determine whether we can add at most~$k$ candidates from~$D$ into~$C$ to make~$p$ the amendment winner. Let $D_1\subseteq D$ be the set of predecessors of~$p$ contained in~$D$. Observe that it does not make any sense to add any candidates from~$D\setminus D_1$.

We maintain a table~$H(d)$ indexed by candidates~$d\in D_1$. We define~$H(d)$ as the minimum number of successors of~$d$ from~$D_1$ that need to be added, under the condition that~$d$ is added into~$C$, such that a $(d \leftarrow p)$-beating path exists with respect to the resulting election and~$\rhd$. 
If such candidates do not exist, we define $H(d)=+\infty$. We process the candidates in~$D_1$ sequentially, following the decreasing order of their positions in the agenda~$\rhd$ (i.e., from the rightmost to the leftmost). In particular, the table is computed as follows.

First, let~$d\in D_1$ be the rightmost predecessor of~$p$ from~$D_1$. We let $H(d)=1$ if there is a $(d\leftarrow p)$-beating path with respect to $(C\cup \{d\}, V)$ and~$\rhd$. Otherwise, we let $H(d)=+\infty$. 

Next, we recursively compute other entries. To this end, let~$d\in D_1$ denote the currently considered candidate. Assuming that~$H(d')$ has been computed for all $d'\in D_1$ such that $d \rhd d'$, we compute~$H(d)$ as follows.
\begin{itemize}
\item If there is a $(d\leftarrow p)$-beating path with respect to $(C\cup \{d\}, V)$ and~$\rhd$, we let $H(d)=1$.
\item Otherwise, we define~$H(d) = +\infty$ if no $(d \leftarrow d')$-beating path exists with respect to 
$(C \cup \{d, d'\}, V)$ and~$\rhd$ for any successor~$d' \in D_1$ of~$d$.
\item Otherwise, we let
\[H(d)= 1+\min_{\substack{d'\in D_1,\, d\rhd d'\\ \exists (d\leftarrow d')\text{-beating path w.r.t.\ }(C\cup \{d, d'\}, V)~\text{and}~\rhd}}H(d').\]
\end{itemize}
After computing all entries of the table, we conclude that the given instance~$I$ is a {\yesins} if and only if one of the following two cases occurs:
\begin{itemize}
\item there is a candidate $d\in D_1$ after~$c^{\star}$ in the agenda~$\rhd$ such that $H(d)\leq k$, and there is a $(c^{\star}\leftarrow d)$-beating path with respect to $(C\cup \{d\},V)$ and~$\rhd$.
\item there is a candidate $d\in D_1$ before~$c^{\star}$ in the agenda~$\rhd$ such that $H(d)\leq k$.
\end{itemize}

The algorithm runs in polynomial time since the table has at most~$\abs{D_1}$ entries and computing each entry can be done in polynomial time as shown above.
\end{proof}

For constructive control by deleting candidates, we have the same result.

\begin{theorem}
\label{thm-ccdc-amd-p}
{\prob{CCDC}}-Amendment is polynomial-time solvable.
\end{theorem}

\begin{proof}

Let $I=(C, p, V, \rhd, k)$ be an instance of the {\prob{CCDC}}-Amendment problem, where~$C$ is a set of candidates, $p \in C$ is a distinguished candidate,~$V$ is a multiset of votes on~$C$,~$\rhd$ is an agenda on~$C$, and~$k$ is an integer.

A necessary condition for~$p$ to be the amendment winner is that~$p$ beats all her successors. Therefore, the first step of our algorithm is to remove all the successors of~$p$ not beaten by~$p$, and decrease~$k$ accordingly.

If $k < 0$ after doing so, we directly conclude that~$I$ is a {\noins}. Otherwise, if $p$ has at most $k$ predecessors, we immediately conclude that $I$ is a {\yesins}, because in this case, we can make $p$ the winner by deleting all its predecessors without exceeding the budget~$k$.

We assume now that $p$ has at least $k+1$ predecessors. Let $(c_1, c_2, \dots, c_t)$ be the agenda~$\rhd$ restricted to the predecessors of~$p$. It holds that $t \geq k+1$. For notational brevity, let $c_{t+1} = p$. For $i, j \in [t]$ such that $i < j$, let $C[i, j]$ denote the union of $\{c_i, c_j\}$ and the set of candidates between~$c_i$ and~$c_j$ in~$\rhd$. Additionally, for $i = j$, we let $C[i, j] = \{c_i\}$, and for $i > j$, we define $C[i, j] = \emptyset$. We maintain a table~$H(c_i)$ for~$i \in [t+1]$. Specifically,~$H(c_i)$ is defined as the minimum number of candidates that need to be removed from $C[i+1, t]$ such that there is a $(c_i \leftarrow p)$-beating path in the remaining election with respect to~$\rhd$. 

We compute the entries in the order of~$H(c_{t+1})$,~$H(c_t)$,~$\dots$,~$H(c_1)$. First, we set $H(c_{t+1}) = 0$. For each~$c_i$ where $i \in [t]$, if~$c_i$ is not beaten by any candidate from $C[i+1, t+1]$, we let $H(c_i) = +\infty$; otherwise, we compute
\[
H(c_i) = \min_{\substack{i < j \leq t+1,\\ c_j~\text{beats or ties with}~c_i}} \left(H(c_j) + j - i - 1\right).
\]
Here, $j - i - 1$ accounts for that all candidates in $C[i+1, j-1]$ are considered to be deleted.

The given instance~$I$ is a {\yesins} if and only if there exists a candidate~$c_i$ with~$i \in [t]$ such that~$H(c_i) \leq k - i + 1$. 
The existence of such a candidate~$c_i$ implies that we can delete all the~$i-1$ predecessors of~$c_i$, along with at most~$k - i + 1$ candidates between~$c_i$ and~$p$ in~$\rhd$, to ensure the presence of a~$(c_i \leftarrow p)$-beating path in the remaining election. Since~$p$ beats all her successors, this guarantees that~$p$ will become the winner after the deletion of the aforementioned candidates. Moreover, note that $p$ has at least $k+1$ predecessors, which ensures that at least one candidate $c_j$ with~$j \in [t]$ is not deleted.

Since the table contains polynomially many entries, and each entry can be computed in polynomial time, the overall algorithm runs in polynomial time.
\end{proof}

As {\prob{DCAC}}-Amendment and {\prob{DCDC}}-Amendment are respectively polynomial-time Turing reducible to {\prob{CCAC}-Amendment} and {\prob{CCDC}}-Amendment, we obtain the following corollary as a consequence of Theorem~\ref{thm-ccac-amd-p} and Theorem~\ref{thm-ccdc-amd-p}.

\begin{corollary}
\label{cor-dcac-dcdc-amd-p}
{\prob{DCAC}}-Amendment and {\prob{DCDC}}-Amendment are polynomial-time solvable.
\end{corollary}

\subsection{The \texorpdfstring{$(m-h)$-Amendment}{Lg} Procedure for Constant~\texorpdfstring{$h$}{Lg}}
\label{sec-amd-m-h}
This section focuses on election control for the $(m-h)$-amendment procedures, where~$m$ denotes the number of candidates and~$h$ is a positive integer constant. We stipulate that for a fixed constant~$h$, when the $(m-h)$-amendment procedure is applied to an election with $m \leq h$ candidates and an agenda, the winner is the first candidate in the agenda. 

Our first result is stated as follows.

\begin{theorem}
\label{thm-ccav-suc-np}
For every constant~$h$, {\prob{CCAV}}-$(m-h)$-Amendment is {\memph\wah} with respect to the combined parameter of the number of added votes and the number of registered votes. This holds even when the distinguished candidate is the last candidate in the agenda.
\end{theorem}

\begin{proof}
We prove the theorem through a reduction from the {\prob{Perfect Code}} problem. Consider an instance $(G, \kappa)$ of {\prob{Perfect Code}}, where $G$ is a graph, and $\kappa$ is an integer. Let~$\vset$ denote the vertex set of~$G$. We first give a reduction for the case where $h=1$.
Particularly, we construct an instance  of {\prob{CCAV}}-{\Famend} as follows.

Let $(\vere_1, \vere_2, \dots, \vere_n)$ be any arbitrary linear order of~$\vset$. 
For each $\vere_i$, we create two candidates denoted by~$x_i$ and~$y_i$. In addition, we create one candidate~$p$. Let $X=\{x_i \setmid i\in [n]\}$, let $Y=\{y_i \setmid i\in [n]\}$, and let $C=X\cup Y\cup \{p\}$. Furthermore, let $\overrightarrow{X}=(x_1, x_2, \dots, x_n)$, and let $\overrightarrow{Y}=(y_1, y_2, \dots, y_n)$. The agenda is $\rhd=(\overrightarrow{X}, \overrightarrow{Y}, p)$. 
For each $\vere_i$ where $i\in [n]$, let
\[N_X[\vere_i]=\{x_j\in X \setmid j\in [n], \vere_j\in N_G[\vere_i]\}\]
be the set of candidates in~$X$ corresponding to the closed neighborhood of~$\vere_i$ in~$G$.
Analogously, let
\[N_Y[\vere_i]=\{y_j\in Y \setmid j\in [n], \vere_j\in N_G[\vere_{i}]\}.\]
 We create the following registered votes:
\begin{itemize}
\item $\kappa+2$ votes with the preference $p\Succ \overrightarrow{X} \Succ \overrightarrow{Y}$;  %%% for the second variant of amendment create one more vote with this preference
\item $\kappa-2$ votes with the preference  $\overrightarrow{X}\Succ p \Succ \overrightarrow{Y}$; and
\item $\kappa+2$ votes with the preference  $\overrightarrow{X}\Succ \overrightarrow{Y} \Succ p $.
\end{itemize}
Let~$V$ be the multiset of the above~$3\kappa+2$ registered votes. Figure~\ref{fig-ccav-suc-hard} shows the majority graph of $(C, V)$.
\begin{figure}
    \centering
    \includegraphics[width=0.45\textwidth]{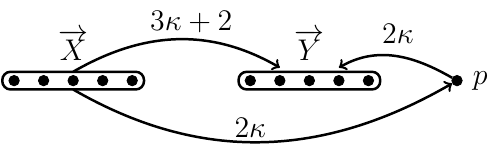}
    \caption{An illustration of the weighted majority graph of $(C, V)$ and the agenda as used in the proof of Theorem~\ref{thm-ccav-suc-np}. Arcs among the vertices in $X$ (respectively,~$Y$) are forward arcs with a uniform weight of $3\kappa + 2$. The agenda is represented by the left-to-right ordering of the vertices.}
    \label{fig-ccav-suc-hard}
\end{figure}
Now we construct the unregistered votes.
For each $\vere_i$, $i\in [n]$, we create one unregistered vote~$\succ_i$ with the preference
\[\left(\overrightarrow{X}[N_X[\vere_i]]\right) \Succ \left(\overrightarrow{Y}\setminus N_Y[\vere_i]\right) \Succ p \Succ \left(\overrightarrow{X}\setminus N_X[\vere_i]\right) \Succ \left(\overrightarrow{Y}[N_Y[\vere_i]]\right).\]
Let~$W$ denote the multiset of the above unregistered votes. 
Let $k=\kappa$, i.e., we are allowed to add at most~$\kappa$ unregistered votes. The {\prob{CCAV}}-{\Famend} instance is $(C, p, V, W, \rhd, k)$.

In the following, we show that the given {\prob{Perfect Code}} instance is a {\yesins} if and only if the above constructed {\prob{CCAV}}-{\Famend} instance is a {\yesins}. 
Observe that as for every $c\in C\setminus \{p\}$, all the $3\kappa+2$ registered votes rank~$c$ before all her successors except~$p$, no matter which at most~$\kappa$ unregistered votes are added,~$c$ still beats all her successors in~$C\setminus \{p\}$.

$(\Rightarrow)$ Assume that the graph~$G$ admits a perfect code $\vset'\subseteq \vset$ of size~$\kappa$. Let $W'=\{\succ_i\, \setmid \vere_i\in \vset'\}$ be the set of unregistered votes corresponding to~$\vset'$. Let $\elec=(C, V\muplus W')$. As~$p$ is the last candidate in the agenda, and every candidate from $C\setminus \{p\}$ beats all her successors from~$C\setminus \{p\}$,~$p$ wins~$\elec$ if and only if~$p$ is not beaten by anyone else in~$\elec$. Consider two candidates $x_i\in X$ and $y_i\in Y$ where $i\in [n]$. As~$\vset'$ is a perfect code of~$G$, there is exactly one $\vere_j\in \vset'$ such that $\vere_i\in N_{G}[\vere_j]$. Due to the definition of unregistered votes, this implies that in~$W'$, only the vote~$\succ_j$ created for~$\vere_j$ ranks~$x_i$ before~$p$ and, moreover, only the same vote~$\succ_j$ ranks~$p$ before~$y_i$. It follows that~$n_{V\muplus W'}(p, x_i)=(\kappa+2)+(\kappa-1)=2\kappa+1$ and $n_{V\muplus W'}(p, y_i)=(\kappa+2)+(\kappa-2)+1=2\kappa+1$. As there are exactly $3\kappa+2+\kappa=4\kappa+2$ votes in~$\elec$, we know that~$p$ ties with both~$x_i$ and~$y_i$. As~$i$ can be any integer in~$[n]$, we know that~$p$ ties with all candidates from $X\cup Y$ in~$\elec$.

$(\Leftarrow)$ Assume that there is a subset $W'\subseteq W$ of at most~$\kappa$ votes such that~$p$ is the {\famend} winner of $\elec=(C, V\muplus W')$ with respect to the agenda~$\rhd$. First, observe that every candidate from $X$ beats each of its successors in $\elec$, except~$p$, regardless of which at most $\kappa$ votes are contained in $W'$. This is because all the $3\kappa + 2$ registered votes rank every candidate from $X$ before all of its successors except~$p$. 

Moreover, if $W'$ contains at most $\kappa - 1$ votes, then there exists at least one candidate from~$X$ who also beats~$p$ with respect to $V \cup W'$. This follows because there are already~$2\kappa$ registered votes that rank every candidate from $X$ before~$p$. If no votes are added, all candidates from $X$ beat~$p$. Alternatively, if any unregistered vote, say $\succ_i$, corresponding to a vertex~$u_i$ in $G$ is added, any candidate from $X$, say $x_i$, that is ranked before~$p$ beats~$p$ in $\elec$. 

In other words, if $\abs{W'} \leq \kappa - 1$, there would exist at least one candidate in~$X$ who beats all of her successors. In this case, the leftmost such candidate in the agenda would be the winner of~$\elec$.  
However, this contradicts the assumption that~$p$ wins~$\elec$.

From the above analysis, we conclude that $\abs{W'} = \kappa$. It follows that $\abs{V \muplus W'} = 4\kappa + 2$. Let $\vset' = \{\vere_i \in \vset \setmid\; \succ_i \in W'\}$ be the set of vertices corresponding to~$W'$. We claim that~$\vset'$ is a perfect code of~$G$.
 Suppose, for the sake of contradiction, that this is not the case. Then, one of the following two cases occurs.
\begin{description}
\item[Case~1:] $\exists \vere_i\in \vset$ such that $\vere_i\not\in N_G[\vset']$. \hfill

In this case, all unregistered votes in~$W'$ rank~$y_i$ before~$p$. So, in~$\elec$ there are in total $\kappa+2+\kappa=2\kappa+2$ votes ranking~$y_i$ before~$p$, implying that~$p$ is beaten by~$y_i$ in~$\elec$, 
contradicting that~$p$ wins~$\elec$.

\item[Case~2:] $\exists \vere_i\in \vset$ and $\vere_j, \vere_{j'}\in \vset'$ such that $\vere_i\in N_G[\vere_j]\cap N_G[\vere_{j'}]$. \hfill

Due to the construction of the unregistered votes, the two votes~$\succ_j$ and~$\succ_{j'}$ in~$W'$ both rank~$x_i$ before~$p$. In this case, there are at least $2\kappa+2$ votes in $V\muplus W'$ ranking~$x_i$ before~$p$, implying that~$x_i$ beats~$p$ in~$\elec$. However, this contradicts that~$p$ wins~$\elec$.
\end{description}

Since both of the above cases lead to contradictions,~$\vset'$ is a perfect code of~$G$.

To show the {\wahns} of {\prob{CCAV}}-$(m-h)$-Amendment we add a set of $h-1$ additional candidates $p_1$, $p_2$, $\dots$, $p_{h-1}$ in the following way:
\begin{enumerate}
    \item[(1)] we add them before~$p$ with their relative order being $p_{h-1}$, $p_{h-2}$, $\dots$, $p_1$ in the agenda; and
    \item[(2)] in each vote, we rank them after~$p$ with their relative order being~$p_1$,~$p_2$,~$\dots$,~$p_{h-1}$.
\end{enumerate}
So, in the new election,~$p$ beats each~$p_i$ where $i\in [h-1]$, and the head-to-head comparison between every candidate~$x_i$ (respectively,~$y_i$) and each~$p_j$ where $j\in [h-1]$ is exactly the same as that between~$x_i$ (respectively,~$y_i$) and~$p$.
\end{proof}

For control by deleting voters, we also obtain a parameterized intractability result.

\begin{theorem}
\label{thm-ccdv-suc-np}
For every constant~$h$, {\prob{CCDV}}-$(m-h)$-Amendment is {\memph\wbh} with respect to the number of deleted votes. This holds even when the distinguished candidate is the last candidate in the agenda.
\end{theorem}

\begin{proof}
We prove the theorem by reductions from the {\prob{RBDS}} problem. Let $(G, \kappa)$ be an instance of {\prob{RBDS}}, where $G$  is a bipartite graph with the bipartition $(R, B)$. We first provide a reduction for the case where $h=1$. Similar to the proof of Theorem~\ref{thm-ccdv-amd-np}, we assume that every red vertex has the same degree, which we denote by $\ell$. We also assume that $\kappa\geq 3$. The candidate set and the agenda are exactly the same as in the proof of Theorem~\ref{thm-ccdv-amd-np}. Precisely, we have that $C=R\cup \{p, q\}$ where~$p$ is the distinguished candidate, and $\rhd=(q, \overrightarrow{R}, p)$.
We create the following $2\abs{B}+2\ell-\kappa$ votes in~$V$:
\begin{itemize}
\item $\abs{B}-\kappa+1$ votes with the preference $\overrightarrow{R}\Succ p\Succ q$;
\item $\ell$ votes with the preference $q\Succ p\Succ \overrightarrow{R}$;
\item $\ell-1$ votes with the preference $p\Succ q\Succ \overrightarrow{R}$; and 
\item for each blue vertex $b\in B$, one vote~$\succ_b$ with the preference
\[q\Succ \left(\overrightarrow{R}[N_G(b)]\right) \Succ p \Succ \left(\overrightarrow{R}\setminus N_G(b)\right).\]
\end{itemize}
Let $k=\kappa$. The constructed instance of {\prob{CCDV-\Famend}} is $(C, p, V, \rhd, k)$.
The weighted majority graph of $(C, V)$ and the agenda~$\rhd$ are illustrated in Figure~\ref{fig-ccdv-suc-np}.

\begin{figure}
    \centering
    \includegraphics[width=0.45\textwidth]{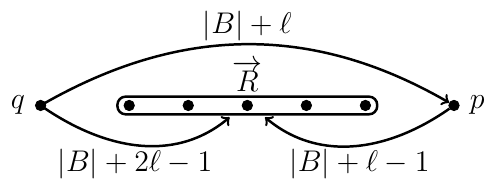}
    \caption{An illustration of the weighted majority graph of $(C, V)$ and the agenda as used in the proof of Theorem~\ref{thm-ccdv-suc-np}. All arcs between the vertices in $R$ are forward arcs with a weight of at least $2\abs{B} - \kappa + \ell$. The agenda is depicted by the left-to-right ordering of the vertices.}
    \label{fig-ccdv-suc-np}
\end{figure}

It is easy to see that~$q$ beats everyone else and is the {\famend} winner of~$(C, V)$ with respect to~$\rhd$.
The construction clearly can be created in polynomial time.
In the following, we prove the correctness of the reduction.

$(\Rightarrow)$ Assume that there is a subset $B'\subseteq B$ of cardinality~$\kappa$ such that~$B'$ dominates~$R$ in~$G$. Let $V'=\{\succ_b\, \setmid b\in B'\}$ be the set of votes corresponding to~$B'$. Let $\elec=(C, V\setminus V')$. Clearly, $\abs{V\setminus V'}=2\abs{B}-2\kappa+2\ell$. We shall show that~$p$ is not beaten by anyone else in~$\elec$ and hence is the {\famend} winner. As all votes in~$V'$ rank~$q$ before~$p$, it holds that $n_{V\setminus V'}(p, q)=(\abs{B}-\kappa+1)+(\ell-1)=\abs{B}-\kappa+\ell$, meaning that~$p$ ties with~$q$ in~$\elec$. Moreover, as~$B'$ dominates~$R$, for every $r\in R$, there is at least one blue vertex $b\in B'$ dominating~$r$. By the construction of the votes, we know that~$r$ is ranked before~$p$ in the vote $\succ_b\in V'$ corresponding to~$b$. It follows that at most $\kappa-1$ votes in~$V'$ rank~$p$ before~$r$. By the construction of the votes, we know that there are at least $\ell+(\ell-1)+(\abs{B}-\ell)-(\kappa-1)=\abs{B}+\ell-\kappa$ votes ranking~$p$ before~$r$ in~$V\setminus V'$, implying that~$p$ ties with~$r$ in~$\elec$. As this holds for all $r\in R$, we know that~$p$ ties with all candidates in~$R$. It follows that~$p$  is the {\famend} winner of~$\elec$.

$(\Leftarrow)$ Assume that we can make~$p$ the {\famend} winner by deleting at most $k=\kappa$ votes from~$V$, i.e., $\exists V'\subseteq V$ such that $\abs{V'}\leq k$ and~$p$ wins $\elec=(C, V\setminus V')$ with respect to~$\rhd$. It is easy to see that no matter which at most~$k$ votes are contained in~$V'$,~$q$ beats all candidates in~$R$, and every $r\in R$ beats all her successors in~$R$ in the election~$\elec$.  Therefore, it must be that~$p$ beats or ties with all the other candidates in~$\elec$. This implies that all votes in~$V'$ must rank~$q$ before~$p$ and, moreover, it must be that $\abs{V'}=k=\kappa$, since otherwise~$p$ is beaten by~$q$ in~$\elec$. There are two groups of votes ranking~$q$ before~$p$: those corresponding to the blue vertices, and those with the preference $q \Succ p \Succ \overrightarrow{R}$. We may assume that all votes in~$V'$ belong to the first group. This assumption holds because if~$V'$ contained any vote with the preference $q \Succ p \Succ \overrightarrow{R}$, we could construct another solution~$V''$ by replacing this vote with any vote from the first group that is not currently in~$V'$. 
 Let~$r$ be a candidate from~$R$. As $n_V(r, p)=\abs{B}-\kappa+\ell+1$ and $\abs{V\setminus V'}=2\abs{B}-2\kappa+2\ell$, we know that there is at least one vote $\succ_b\in V'$ which ranks~$r$ before~$p$. By the reduction, we know that the vertex~$b$ corresponding to~$\succ_b$ dominates~$r$ in~$G$. It is clear now that the subset $B'=\{b\in B \setmid\  \succ_b\in V'\}$ of blue vertices corresponding to~$V'$ dominates~$R$ in~$G$.

To prove the {\wbhns} of {\prob{CCDV}}-$(m-h)$-Amendment for other constant values of~$h$, we modify the reduction above in the same manner as in the proof of Theorem~\ref{thm-ccav-suc-np}. Specifically, we add $h-1$ additional candidates, denoted as~$p_1$,~$p_2$, $\dots$,~$p_{h-1}$, in the same way as in the proof of Theorem~\ref{thm-ccav-suc-np}.
\end{proof}

\begin{theorem}
\label{thm-ccdv-famend-wbh-remaining}
For every constant~$h$, {\prob{CCDV}}-$(m-h)$-Amendment is {\memph\wbh} with respect to the number of remaining votes. This result holds even when the distinguished candidate occupies the last position in the agenda. 
\end{theorem}

\begin{proof}
We prove the theorem via reductions from the {\prob{RBDS}} problem. First, we consider the case where $h=1$. Let $(G, \kappa)$ be an instance of {\prob{RBDS}}, where~$G$ is a bipartite graph with the bipartition $(R, B)$. Without loss of generality, assume that $\kappa < \abs{B}$. We construct an instance of {\prob{CCDV}}-Full-Amendment as follows.

For each red vertex, we create one candidate denoted still by the same symbol for notational simplicity. In addition, we create two candidates~$q$ and~$p$, where~$p$ is the distinguished candidate. Let $C=R\cup \{q, p\}$. The agenda is $\rhd=(q, \overrightarrow{R}, p)$.
We create the following votes:
\begin{itemize}
\item a multiset~$V_1$ of $\kappa-1$ votes with the preference $p\Succ q\Succ \overrightarrow{R}$;
\item a singleton~$V_2$ of one vote with the preference $\overrightarrow{R}\Succ p\Succ q$; and
\item for each blue vertex $b\in B$, one vote $\succ_b$ with the preference
\[q\Succ \left(\overrightarrow{R}\setminus N_G(b)\right) \Succ p\Succ \left(\overrightarrow{R}[N_G(b)]\right).\]
\end{itemize}
For a given $B'\subseteq B$, let $V_{B'}=\{\succ_b \setmid b\in B'\}$ be the submultiset of votes created for the blue vertices in~$B'$. Let $V=V_1\cup V_2\cup V_B$. It holds that $\abs{V}=\abs{B}+\kappa$. Let $k=\abs{B}-\kappa$. Obviously, under the assumption that $\kappa<\abs{B}$, the {\famend} winner of $(C, V)$ with respect to~$\rhd$ is~$q$.
The instance of {\prob{CCDV}}-{\Famend} is $(C, p, V, \rhd, k)$, which asks if there is a submultiset $V'\subseteq V$ of cardinality at least $\abs{V}-k=2\kappa$ such that~$p$ is the {\famend} winner of $(C, V')$ with respect to~$\rhd$. It remains to show the correctness of the reduction.

$(\Rightarrow)$ Assume that there is a subset $B'\subseteq B$ such that~$\abs{B'}=\kappa$ and~$B'$ dominates~$R$ in~$G$. Let $V'=V_1\cup  V_2\cup V_{B'}$, and let $\elec=(C, V')$. We claim that~$p$ is the {\famend} winner of~$\elec$ with respect to~$\rhd$. Since~$p$ is the last candidate in the agenda, it suffices to show that~$p$ is not beaten by any other candidate in~$\elec$. It is clear that~$p$ ties with~$q$ in~$\elec$. Let $r\in R$ be a red vertex. As~$B'$ dominates~$R$, there is at least one vote~$\succ_b$ in $V_{B'}$ such that~$b$ dominates~$r$. From the construction of the votes, we know that~$p$ is ranked before~$r$ in the vote~$\succ_b$. Therefore, there are in total at least $\abs{V_1}+1=\kappa$ votes ranking~$p$ before~$r$ in~$V'$, implying that~$p$ is not beaten by~$r$. As this holds for all $r\in R$, the correctness for this direction follows.

$(\Leftarrow)$ Assume that there exists a submultiset $V' \subseteq V$ of at least $2\kappa$ votes such that~$p$ is the {\famend} winner of $(C, V')$ with respect to~$\rhd$. Observe that, since $\abs{V_1} + \abs{V_2} = \kappa$ and all votes in~$V_B$ rank~$q$ in first place, it must hold that $(V_1 \cup V_2) \subseteq V'$, and~$V'$ contains exactly~$\kappa$ votes from~$V_B$. Otherwise,~$q$ would be the {\famend} winner of $(C, V')$, contradicting that~$p$ wins $(C, V')$. Define $V_{B'} = V' \cap V_B$, where $B' \subseteq B$. As noted, it follows that $\abs{V_{B'}} = \kappa$. 

We claim that~$B'$ dominates~$R$ in~$G$. Suppose, for the sake of contradiction, that there exists a red vertex $r \in R$ not dominated by any vertex from~$B'$. Then, by the construction of the votes,~$r$ is ranked before all its successors in all votes of~$V_{B'}$. Combined with the vote in~$V_2$, there are $\kappa + 1$ votes in~$V'$ ranking~$r$ before all its successors. In this case, the leftmost $r \in R$ not dominated by~$B'$ becomes the {\famend} winner of $(C, V')$, contradicting the winning of~$p$ in $(C, V')$. 

Therefore, the claim holds. Since $\abs{B'} = \kappa$, the {\prob{RBDS}} instance is a {\yesins}.

To prove the {\wbhns} of {\prob{CCDV}}-$(m-h)$-Amendment for other constant values of~$h$, we modify the above reduction in the same way as done  in the proof of Theorem~\ref{thm-ccav-suc-np}.
\end{proof}

Now we study destructive control by adding/deleting voters. Compared with previous results, the complexity transition of these problems for the $(m-h)$-amendment procedures is more sharp. Particularly, by Corollary~\ref{cor-dcav-dcdv-amd-suc-p}, these problems are polynomial-time solvable if the distinguished candidate is the first one in the agenda. However, when the distinguished candidate is the second, the third, or the last candidate in the agenda, we show that the problems become already intractable from the parameterized complexity point of view. This stands in contrast to the fixed-parameter tractability of the same problems for the $h$-amendment procedures with respect to the combined parameter of~$h$ and the number of predecessors of the distinguished candidate (Corollary~\ref{cor-dcav-dcdv-amd-fpt}).

\begin{theorem}
\label{thm-dcav-suc-np}
For every constant~$h$,
{\prob{DCAV}}-$(m-h)$-Amendment is {\memph\wah} with respect to the number of added votes plus the number of registered votes as long as the distinguished candidate is not the first one in the agenda.
\end{theorem}

\begin{proof}
We first show that the reduction  in the proof of Theorem~\ref{thm-dcav-amd-np} directly applies here (with only the agenda being different). Recall that the reduction was from the {\prob{Perfect Code}} problem. We had the candidate set $C=X\cup Y\cup \{p, q\}\cup \{p_1, p_2, \dots, p_{h-1}\}$, where $X=\{x_1, \dots, x_m\}$ and $Y=\{y_1, \dots, y_m\}$ with~$x_i$ and~$y_i$ each being the two candidates created for a vertex~$\vere_i$ in the given graph~$G$ of the {\prob{Perfect Code}} instance, and $(p, p_1, p_2, \dots, p_{h-1})$ is an ordered block. We let~$\rhd$ be an agenda such that~$q$ is the first candidate, and, excluding~$p$, the last $h-1$ candidates in the agenda are~$p_1$,~$p_2$, $\dots$,~$p_{h-1}$. (Other candidates except~$p$ can be ordered arbitrarily between~$q$ and~$p_1$, and~$p$ can be in any position after~$q$ in the agenda.)  We created a multiset~$V$ of votes such that~$p$ beats all the other candidates and, moreover, except for~$q$, this holds even after we add at most~$k$ unregistered votes. This means that~$q$ is the only candidate which is able to replace~$p$ as the new winner. By the definition of the $(m-h)$-amendment procedure,~$q$ is the $(m-h)$-amendment winner if and only if~$q$ beats all the other candidates.
In the proof of Theorem~\ref{thm-dcav-amd-np}, we constructed a multiset of unregistered votes such that there exists a perfect code of $G$ if and only if at most $k$ unregistered votes can be added to ensure that $q$ beats all other candidates. The theorem follows.
\end{proof}

One might wonder why, in Theorem~\ref{thm-dcav-amd-np}, we require the distinguished candidate~$p$ to be the last candidate in the agenda, whereas in Theorem~\ref{thm-dcav-suc-np}, we only stipulate that~$p$ must not be the first candidate in the agenda. 
The distinction arises because, if~$p$ is not the last candidate in the agenda, preventing~$p$ from being an $h$-amendment winner does not necessarily require~$q$ to beat all other candidates. For example, suppose~$p$ immediately follows~$q$ in the agenda. In this case, to prevent~$p$ from being an amendment winner, it suffices for~$q$ to beat~$p$, even if~$q$ is beaten by a candidate appearing after~$p$ in the agenda. 
Our {\fpt}-result for {\prob{DCAV}}-Full-Amendment, stated in Corollary~\ref{cor-dcav-dcdv-amd-fpt}, also reflects this distinction from a different perspective.

For destructive control by deleting voters, we have a similar result.

\begin{theorem}
\label{thm-dcdv-suc-np}
For every constant~$h$, {\prob{DCDV}}-$(m-h)$-Amendment is {\memph\wbh} with respect to the number of deleted votes, as long as the distinguished candidate is not the first one of the agenda.
\end{theorem}

\begin{proof}
We adopt the reduction in the proof of Theorem~\ref{thm-dcdv-amd-np} to prove this theorem. Recall that the reduction was from {\prob{RBDS}}. Specifically, given an instance $I=(G, \kappa)$ of {\prob{RBDS}}, where $G$ is a bipartite graph with the bipartition $(R, B)$, we created a set $C=R\cup \{q\}\cup \{p, p_1, p_2, \dots, p_{h-1}\}$ of candidates, where~$p$ is the distinguished candidate. Moreover, we created a multiset of votes such that
\begin{itemize}
\item $(p, p_1, \dots, p_{h-1})$ is an ordered block;
\item $p$ beats all the other candidates and, moreover, except for~$q$, this holds even if we delete at most~$k=\kappa$ votes; and
\item $I$ is a {\yesins} if and only if we can delete at most~$k$ votes such that~$q$ beats all other candidates.
\end{itemize}
It is fairly easy to verify that the reduction applies to {\prob{DCDV}}-$(m-h)$-Amendment for all agendas where~$q$ is the first candidate, and~$p_1$,~$p_2$, $\dots$, $p_h$ are the last~$h$ candidates in the agenda without~$p$.
\end{proof}

\begin{theorem}
\label{thm-dcdv-amend-wbh-remaining}
For every constant~$h$, {\prob{DCDV}}-$(m-h)$-Amendment is {\memph\wbh} with respect to the number of remaining votes, as long as the distinguished candidate is not the first one of the agenda.
\end{theorem}

\begin{proof}
We prove the theorem via a reduction from the {\prob{RBDS}} problem. Let $(G, \kappa)$ be an instance of {\prob{RBDS}}, where $G$ is a bipartite graph with the bipartition $(R, B)$. Without loss of generality, we assume $\abs{B}\geq \kappa+2$. The candidate set is defined as $C = R \cup \{q, p\}$, where~$p$ is the distinguished candidate. The agenda $\rhd$ can be any ordering in which $q$ occupies the first position. We construct the following votes:  
\begin{itemize}
    \item a multiset $V_1$ of $\kappa$ votes, each with the preference $q \Succ p \Succ \overrightarrow{R}$;  
    \item a singleton $V_2$ containing one vote with the preference $\overrightarrow{R} \Succ q \Succ p$; and  
    \item for each blue vertex $b \in B$, one vote $\succ_b$ with the preference  
    \[
    p \Succ \left(\overrightarrow{R} \setminus N_G(b)\right) \Succ q \Succ \left(\overrightarrow{R}[N_G(b)]\right).
    \]
\end{itemize}

For each subset $B' \subseteq B$, let $V_{B'} = \{\succ_b\, \setmid b \in B'\}$ denote the multiset of votes corresponding to the blue vertices in $B'$. Define $V = V_1 \cup V_2 \cup V_B$. Clearly, $\abs{V} = \abs{B} + \kappa + 1$. Let $k = \abs{B} - \kappa$. Under the assumption $\abs{B}\geq \kappa+2$, it is evident that the $(m-h)$-amendment winner of $(C, V)$ with respect to the agenda~$\rhd$ is $p$.  

It remains to show the correctness of the reduction.  

$(\Rightarrow)$ Assume there exists a subset $B' \subseteq B$ such that $\abs{B'} = \kappa$ and $B'$ dominates $R$ in~$G$. Let $V' = V_1 \cup V_2 \cup V_{B'}$, and let $\elec = (C, V')$. We claim that $q$ is the $(m-h)$-amendment winner of $\elec$ with respect to the agenda $\rhd$. Since $q$ is in the first position of $\rhd$, it suffices to show that~$q$ beats every other candidate in $\elec$ (Lemma~\ref{lem-a}).  

Observe that $\abs{V'} = 2\kappa + 1$. As all the $\kappa + 1$ votes from $V_1 \cup V_2$ prefer $q$ to $p$, it follows that~$q$ beats $p$ in $\elec$. Let $r \in R$ be a red vertex. Since~$B'$ dominates $R$ in~$G$, there exists at least one vote $\succ_b$ in $V_{B'}$ such that~$b$ dominates~$r$, i.e., $r\in N_G(b)$. By the construction of~$\succ_b$, we know that~$q$ is ranked before~$r$ in this vote. Thus, there are at least $\abs{V_1} + 1 = \kappa + 1$ votes in~$V'$ ranking~$q$ before~$r$, implying that~$q$ beats $r$. Since this holds for all $r \in R$, the correctness of this direction is established.  

$(\Leftarrow)$ Assume there exists a submultiset $V' \subseteq V$ of at least $2\kappa + 1$ votes such that $p$ is not the $(m-h)$-amendment winner of $(C, V')$ with respect to~$\rhd$. As $\abs{V_1} + \abs{V_2} = \kappa + 1$ and all votes in $V_B$ rank $p$ in the first place, we know that $(V_1 \cup V_2) \subseteq V'$, and $V'$ must contain exactly $\kappa$ votes from~$V_B$, because otherwise $p$ remains the $(m-h)$-amendment winner of $(C, V')$, a contradiction. Let $V_{B'} = V' \cap V_B$, where $B' \subseteq B$. It follows that $\abs{V_{B'}} = \kappa$.  

Observe that $p$ beats all candidates from $R$ in $(C, V')$, regardless of which $\kappa$ votes from~$V_B$ are included in $V'$. Therefore, $q$ must be the $(m-h)$-amendment winner of $(C, V')$, meaning~$q$ beats all other candidates. We claim that $B'$ dominates $R$.  

Suppose, for the sake of contradiction, that there exists a red vertex $r \in R$ not dominated by any vertex in $B'$. Then, by the construction of the votes, $r$ is ranked before $q$ in all votes from $V_{B'}$. Together with the vote in $V_2$, there are $\kappa + 1$ votes in~$V'$ ranking $r$ before $q$, implying that $r$ beats $q$ in $(C, V')$. This contradicts the assumption that $q$ beats all candidates in $(C, V')$.  
Thus, $B'$ must dominate $R$ in~$G$. 

As $\abs{B'}=\kappa$, the given instance of the {\prob{RBDS}} problem is a {\yesins}.
\end{proof}

We now turn to candidate control operations. Unlike the $h$-amendment procedures, we first show that constructive control by adding candidates under the $(m-h)$-amendment procedures is hard to solve, for each positive integer constant $h$. Notably, this is our only hardness result for candidate control under the various variants of amendment procedures.

\begin{theorem}
\label{thm-ccac-suc-np}
For every constant~$h$, {\prob{CCAC}}-$(m-h)$-Amendment is {\memph\wbh} with respect to the number of added candidates. This holds even when the distinguished candidate is the last one in the agenda.
\end{theorem}

\begin{proof}
We prove the theorem via a reduction from {\prob{RBDS}}. Let $(G, \kappa)$ be an instance of {\prob{RBDS}}, where $G$ is a bipartite graph with the bipartition $(R, B)$. We first consider the {\famend} procedure, i.e., the case where $h=1$. We construct an instance of {\prob{CCAC}}-Full-Amendment as follows. 

For each vertex in $G$, we create one candidate, denoted by the same symbol for notational simplicity. Additionally, we create a distinguished candidate $p$. Let $C = R \cup \{p\}$, let $D = B$, let $k = \kappa$, and define the agenda as $\rhd = (\overrightarrow{R}, \overrightarrow{B}, p)$.
 We create a multiset~$V$ of votes such that
\begin{itemize}
\item every candidate from~$R$ beats~$p$;
\item $p$ beats every candidate from~$B$;
\item for each $r\in R$ and $b\in B$, if~$b$ dominates~$r$ in~$G$, then~$b$ beats~$r$; otherwise,~$r$ beats~$b$; and
\item every candidate $r\in R$ is beaten by all of its predecessors in the agenda~$\rhd$.
\end{itemize}
We refer to Figure~\ref{fig-ccac-amd-hard} for an illustration of the majority graph of $(C\cup D, V)$ and the agenda~$\rhd$. 
\begin{figure}
    \centering
    \includegraphics[width=0.45\textwidth]{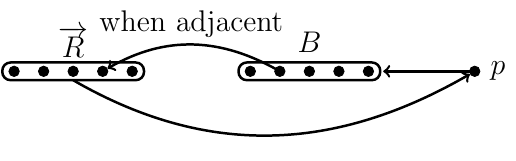}
    \caption{An illustration of the majority graph of $(C \cup D, V)$ and the agenda used in the proof of Theorem~\ref{thm-ccac-suc-np}. All arcs between the vertices in $R$ are forward. Arcs between the vertices in~$B$ are irrelevant as they have no impact on the correctness of the reduction. The agenda is depicted by the left-to-right ordering of the vertices.}
    \label{fig-ccac-amd-hard}
\end{figure}
By McGarvey's theorem~\cite{McGarvey1953}, such votes can be constructed in polynomial time. The corresponding {\prob{CCAC}}-Full-Amendment instance is given by $(C, p, D, V, \rhd, k)$.

The correctness of the reduction is easy to see.
In particular, if there is a subset $B'\subseteq B$ of~$\kappa$ vertices which dominate~$R$ in $G$, then after adding the candidates corresponding to~$B'$, for every~$r\in R$, there is at least one candidate from~$B'$ beating~$r$, excluding the winning of~$r$. Candidates from~$B'$ cannot win as they are beaten by~$p$. Therefore, after adding these candidates,~$p$ becomes the winner. If, however, the {\prob{RBDS}} instance is a {\noins}, then no matter which at most~$\kappa$ candidates from~$B$ are added, there is at least one candidate in~$R$ who beats all her successors in the agenda. As all candidates from~$R$ are before~$p$ in the agenda~$\rhd$, in this case, we cannot add at most~$\kappa$ candidates to make~$p$ the winner.

The above reduction can be modified in the following way to prove the hardness for every $h>1$: replace the distinguished candidate by an ordered block of~$h$ candidates with the first one in the block being the distinguished candidate.
\end{proof}

However, when the distinguished candidate has only a few predecessors, the problem becomes tractable from the perspective of parameterized complexity.

\begin{theorem}
\label{thm-ccac-suc-fpt}
For every constant~$h$, {\prob{CCAC}}-$(m-h)$-Amendment is {\memph\fpt} with respect to the number of predecessors of the distinguished candidate.
\end{theorem}

\begin{proof}
Let $I=(C, p, D, V, \rhd, k)$ be an instance of {\prob{CCAC}}-$(m-h)$-Amendment, where $p\in C$ is the distinguished candidate. Let~$C_1$ (respectively,~$D_1$) be the set of predecessors of~$p$ contained in~$C$ (respectively,~$D$), and let $C_2=C\setminus (C_1\cup \{p\})$ (respectively, $D_2=D\setminus D_1$). Additionally, let $D_2'\subseteq D_2$ be the set of candidates from~$D_2$ beaten by~$p$ with respect to~$V$.

We first give an {\fpt}-algorithm for the {\famend} procedure, and then we show how to extend the algorithm for the $(m-h)$-amendment procedures. Let $\ell=\abs{C_1}+\abs{D_1}$ be the number of predecessors of~$p$ in $C\cup D$.

If at least one candidate from~$C_2$ is not beaten by~$p$ with respect to~$V$, we immediately conclude that~$I$ is a {\noins}.
So, let us assume that~$p$ beats all candidates from~$C_2$. 

We split~$I$ into at most~$2^{\ell}$ subinstances by enumerating all subsets~$S\subseteq D_1$. In particular, each subinstance takes~$I$ and a subset $S\subseteq D_1$ as input, and determines if there is a subset $D'\subseteq D_2$ such that $\abs{D'}\leq k-\abs{S}$ and~$p$ is the {\famend} winner of $(C\cup S\cup D', V)$ with respect to~${\rhd}$. Clearly,~$I$ is a {\yesins} if and only if at least one of the subinstances is a {\yesins}. To solve a subinstance associated with an enumerated subset~$S$, we create an instance of {\prob{RBDS}} as follows.
\begin{itemize}
\item Red vertices are those in~$C_1\cup S$ who beat their respective successors from~$C\cup S$.
\item Blue vertices are those in~$D_2'$.
\item We create an edge between a red vertex~$c$ and a blue vertex~$c'$ if and only if~$c'$ is not beaten by~$c$ with respect to~$V$.
\item Let~$\kappa=\min \{k-\abs{S}, \abs{D_2'}\}$.
\end{itemize}
It is straightforward to verify that the {\prob{RBDS}} instance is a {\yesins} if and only if the subinstance is a {\yesins}. Regarding the running time, we reiterate that {\prob{RBDS}} is {\fpt} with respect to the number of red vertices (see, e.g., the works of Dom, Lokshtanov, and Saurabh~\cite{DBLP:journals/talg/DomLS14}, Fomin, Kratsch, and Woeginger~\cite{DBLP:conf/wg/FominKW04}). Then, as each {\prob{RBDS}} instance constructed above has at most~$\ell$ red vertices, and we have at most~$2^{\ell}$ subinstances to consider, the whole algorithm runs in {\fpt}-time with respect to~$\ell$.

We now consider the $(m-h)$-amendment procedures, where $h$ is a positive constant integer. We enumerate all tuples $(S, S')$ of two subsets such that 
\begin{itemize}
    \item $S\subseteq D_1$,
    \item $S'\subseteq D_2'$, 
    \item $\abs{S'}\leq \min\{h-1, \abs{D_2'}\}$, and 
    \item $\abs{S}+\abs{S'}\leq k$.
\end{itemize} 
Candidates in~$S$ are assumed to be exactly those from~$D_1$ that are included in a desired feasible solution, while candidates in~$S'$ are also considered part of the solution. Moreover, candidates in~$S'$ are assumed to be the last candidates from~$D_2'$ included in the solution. Let $B$ denote the set of candidates from~$D_2'$ that appear before any candidate from~$S'$ in~$\rhd$. Formally, $B \subseteq D_2'$, and for all $b \in B$ and all $c \in S'$, it holds that $b \rhd c$. Let $\kappa = \min\{k - \abs{S \cup S'}, \abs{B}\}$. For each enumerated $\{S, S'\}$, we determine whether there exists $B' \subseteq B$ such that
\begin{itemize}
    \item $\abs{B'} \leq \kappa$, and
    \item $p$ is the $(m-h)$-amendment winner of $(C \cup S \cup S' \cup B', V)$ with respect to~${\rhd}$.
\end{itemize}
The original instance $I$ is a {\yesins} if and only if there is a {\yes}-answer to the above question for at least one enumerated $(S, S')$. 

In the following, we describe how to solve a subproblem associated with an enumerated tuple $(S, S')$. 
Let $\rhd'$ be the agenda $\rhd$ restricted to $C\cup S\cup S'$. Let $A$ be the set of the last $h-1$ candidates in~$\rhd'$. 
%, and let $(a_1, a_2, \dots, a_{h-1})$ be $\rhd$ restricted to $A$. 
We have two cases to consider. 
\begin{description}
    \item[Case:] $p\in A$. \hfill

    In this case, $B=\emptyset$ holds. Consequently, we conclude that $I$ is a {\yesins} if $p$ is the $(m-h)$-amendment winner of $(C\cup S\cup S', V)$ with respect to $\rhd'$, and discard $(S, S')$ otherwise.
    
    \item[Case:]  $p\not\in A$. \hfill 

   In this case, we construct an instance of the {\prob{RBDS}} problem. Let $p = \rhd'[\ell']$, where $\ell' \leq \ell$. The red vertices in the {\prob{RBDS}} instance are from $C_1 \cup S$. Specifically, a candidate~$\rhd'[i]$, where $i \in [\ell' - 1]$, is a red vertex if and only if at least one of the following conditions holds:  
   \begin{itemize}
    \item $i \leq h - 1$, and $\rhd'[i]$ beats all candidates from $\rhd'[i+1, i+m'-h]$ with respect to~$V$, where $m' = \abs{C \cup S \cup S'}$;
    \item $i \geq h$, and $\rhd'[i]$ beats all of its successors in $C \cup S \cup S'$.
   \end{itemize}
   Let $R$ denote the set of the red vertices. 
Blue vertices are those in~$B$. 
Edges are constructed as follows: an edge exists between a vertex $\rhd'[i] \in R$ and a vertex $b \in B$ if and only if~$b$ is not beaten by~$\rhd'[i]$ with respect to~$V$. The instance asks if there exists a subset of $\kappa$ blue vertices that dominate all red vertices. 
If the resulting {\prob{RBDS}} instance is a {\yesins}, we conclude that~$I$ is a {\yesins}; otherwise, $(S, S')$ is discarded.
\end{description} 

If all enumerated tuples are discarded, we conclude that~$I$ is a {\noins}. 

The algorithm runs in {\fpt}-time with respect to~$\ell$ for the following reasons: 
we enumerate at most $2^{\ell} \cdot \abs{D}^{h-1}$ tuples,~$h$ is a constant, 
the {\prob{RBDS}} problem is {\fpt} with respect to the number of red vertices, 
and each {\prob{RBDS}} instance generated contains at most $\ell$ red vertices.
\end{proof}

For the operation of deleting candidates, we present a natural polynomial-time algorithm for the full-amendment procedure, based on two reduction rules outlined below. 

Let $I = (C, p, V, \rhd, k)$ be an instance of the {\prob{CCDC}}-{\Famend} problem.

\begin{reductionrule} 
\label{rule-1}
If $p$ has a successor $q$ in $\rhd$ that is not beaten by $p$ with respect to $V$, remove $q$ from $C$ and $\rhd$, and decrease $k$ by one.
\end{reductionrule}

\begin{reductionrule}
\label{rule-2}
If the {\famend} winner $q$ of $(C, V)$ with respect to $\rhd$ precedes~$p$ in~$\rhd$, remove~$q$ from $C$ and $\rhd$, and decrease $k$ by one.
\end{reductionrule}

Since each application of these reduction rules removes one candidate, the rules can be applied at most $m-1$ times, where $m$ denotes the number of candidates. As the condition in each rule can be verified in polynomial time, exhaustively applying these reduction rules runs in polynomial time. 

A reduction rule is said to be sound for a problem if applying it to any instance of the problem results in an equivalent instance of the same problem.  
The soundness of Reduction Rule~\ref{rule-1} for the {\prob{CCDC}}-{\Famend} problem is evident: if a candidate $p$ can become a winner by deleting other candidates, and there exists a successor~$q$ of $p$ in $\rhd$ not beaten by $p$, then $q$ must be deleted. 
We now prove the soundness of Reduction Rule~\ref{rule-2}.

\begin{lemma}
    Reduction Rule~\ref{rule-2} for the {\prob{CCDC}}-{\Famend} problem is sound.
\end{lemma}

\begin{proof}
Let $I' = (C', p, V', \rhd', k-1)$ be the instance obtained after applying Reduction Rule~\ref{rule-2} to $I = (C, p, V, \rhd, k)$ and $q$, where
\begin{enumerate}
    \item[(1)] $q$ is the {\famend} winner of $(C, V)$ with respect to $\rhd$,
    \item[(2)] $q \rhd p$,
    \item[(3)] $C' = C \setminus \{q\}$,
    \item[(4)] $V'$ is obtained from $V$ by removing $q$ from every vote in $V$,
    \item[(5)] $\rhd'$ is the agenda $\rhd$ restricted to $C'$.
\end{enumerate}

We prove below that $I'$ is a {\yesins} of the {\prob{CCDC}}-Full-Amendment problem if and only if $I$ is a {\yesins} of the same problem.

\textbf{($\Rightarrow$)} Assume there exists $D \subseteq C'$ such that $\abs{D} \leq k-1$ and $p$ is the {\famend} winner of $(C' \setminus D, V)$ with respect to~$\rhd'$. It follows that $D \cup \{q\}$ is a witness for $I$. Thus, $I$ is a {\yesins}.

\textbf{($\Leftarrow$)} Now assume that $I$ is a {\yesins}, i.e., there exists $D \subseteq C$ such that $|D| \leq k$ and~$p$ is the {\famend} winner of $(C \setminus D, V)$ with respect to ${\rhd}$.

If $q \in D$, then $D \setminus \{q\}$ is witness for $I'$. 

Now consider the case where $q \notin D$. By Condition~(1) above, $q$ beats all its successors in~$\rhd$. According to the definition of the {\famend} procedure, if $q$ is not eliminated during the winner determination process, $q$ will remain the winner of $(C \setminus D, V)$, which contradicts $p$ being the winner of $(C \setminus D, V)$. Thus, $q$ must be eliminated when some candidate $a \in C \setminus D$ with $a \rhd q$ is considered during the winner determination procedure.  
By Condition~(2) above, it follows that $a \rhd p$. Note that at each round of the {\famend} procedure, either the considered candidate is eliminated, or all its successors are eliminated (in which case the considered candidate is declared the winner). Therefore, if $q$ is eliminated when $a$ is considered, $p$ must also be eliminated, which contradicts $p$'s status as the winner.

We can now conclude that the case where $q \notin D$ is impossible. 
\end{proof}

With the two reduction rules at hand, we are ready to present the following result.

\begin{theorem}
\label{thm-ccdc-suc-p}
{\prob{CCDC}}-{\Famend} is polynomial-time solvable.
\end{theorem}

\begin{proof}
Let $I = (C, p, V, \rhd, k)$ be an instance of the {\prob{CCDC}}-{\Famend} problem. 
We exhaustively apply Reduction Rule~\ref{rule-1} and Reduction Rule~\ref{rule-2}. The order of application does not matter, as the soundness of each rule is independent of the other. As previously discussed, the exhaustive application of these rules can be performed in polynomial time.

After the reduction rules have been applied exhaustively, candidate $p$ becomes the {\famend} winner. 
We conclude that $I$ is a {\yesins} if and only if $k \geq 0$.
\end{proof}

Since {\prob{DCDC}}-{\Famend} is polynomial-time Turing reducible to {\prob{CCDC}}-{\Famend}, we obtain the following corollary from Theorem~\ref{thm-ccdc-suc-p}.

\begin{corollary}
\label{cor-dcdc-suc-p}
{\prob{DCDC}}-{\Famend} is polynomial-time solvable.
\end{corollary}

It is important to note that the soundness of Reduction Rule~\ref{rule-2} fails for the {\prob{CCDC}}-$(m-h)$-Amendment problem, even when $h = 2$, as demonstrated in Example~\ref{ex-2}.

\begin{example}[Reduction Rule~\ref{rule-2} is not sound for {\prob{CCDC}}-$(m-2)$-Amendment]
\label{ex-2}
Consider an instance $(C, p, V, \rhd, k)$, where the election $(C, V)$  and the majority graph of $(C, V)$ along with the agenda $\rhd$ are depicted below. The election is illustrated on the left, while the majority graph appears on the right, with the agenda represented by a left-to-right ordering of the vertices. Let $k = 1$. The $(m-2)$-amendment winner of $(C, V)$ is $q$. Consequently, if Reduction Rule~\ref{rule-2} is applied, $q$ will be removed, and $k$ will be reduced to zero. However, it is straightforward to verify that $p$ is not the $(m-2)$-amendment winner of the remaining election. (Note that the number of candidates in the remaining election is three, not four. In this case, the $(m-2)$-amendment procedure becomes equivalent to the amendment procedure.) Nevertheless, the instance is a {\yesins}; specifically, the unique feasible solution is to delete candidate $b$.

\begin{minipage}{0.5\textwidth}
\begin{center}
vote 1: $a \Succ q \Succ b \Succ p$

vote 2: $q \Succ b \Succ p \Succ a$

vote 3: $p \Succ b \Succ a \Succ q$
\end{center}
\end{minipage}\begin{minipage}{0.5\textwidth}
    \begin{center}
    \includegraphics[width=0.65\linewidth]{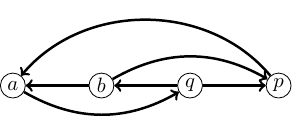}      
    \end{center}
\end{minipage}
\end{example}

Let us now consider the problem of destructive control by adding candidates. We first show that, for the {\famend} procedure, it suffices to consider adding at most one candidate. This result is formalized in the following lemma.

\begin{lemma}
    \label{lem-dcac-full-amendment-adding-one-candidate-enough}
Let $C$ and $D$ be two disjoint sets of candidates, $V$ a multiset of votes over $C \cup D$,~$\rhd$ an agenda on $C \cup D$, and $p \in C$ a candidate. If there exists $D' \subseteq D$ such that $p$ is not the {\famend} winner of $(C \cup D', V)$ with respect to $\rhd$, then there exists a subset $D'' \subseteq D'$ such that $\abs{D''}\leq 1$ and $p$ is not the {\famend} winner of $(C \cup D'', V)$ with respect to $\rhd$.
\end{lemma}

\begin{proof}
   Assume that there exists a subset $D' \subseteq D$ such that $p$ is not the {\famend} winner of $(C \cup D', V)$ with respect to~$\rhd$. 
   
   If $p$ has a successor $d \in C\cup D'$ that is unbeaten by $p$ with respect to~$V$, we are done. In this case, defining $D'' = D'\cap \{d\}$, $p$ cannot be the {\famend} winner of $(C \cup D'', V)$ due to the presence of~$d$.  

Otherwise, $p$ is eliminated when some candidate $c\in C\cup D'$, preceding $p$ in the agenda~$\rhd$, is considered during the winner determination process applied to $(C \cup D', V)$. In this scenario,~$c$ is declared the winner of $(C \cup D', V)$. Let $D'' = D' \cap \{c\}$. Clearly, $\abs{D''} \leq 1$.  
Now, consider the election $E = (C \cup D'', V)$. If $c$ is eliminated during the winner determination process when one of $c$'s predecessors in $E$ is considered, then, by the definition of the {\famend} procedure,~$p$ is also eliminated along with $c$, thereby preventing $p$ from winning.  
Otherwise, as $c$ wins $(C \cup D', V)$ and $D'' \subseteq D'$, the candidate $c$ beats all of its successors in~$E$. In this case, $c$ also wins $(C \cup D'', V)$.
\end{proof}

Lemma~\ref{lem-dcac-full-amendment-adding-one-candidate-enough} implies that {\prob{DCAC}}-{\Famend} is solvable in polynomial time. We generalize this result to all $(m-h)$-amendment procedures, where~$h$ is a constant, as stated in the following theorem.

To facilitate a better understanding of the algorithm presented below, let us first establish some convenient terminology. Consider an election $E = (C, V)$ and an agenda $\rhd$. Recall that candidates are considered in the order specified by $\rhd$. Assume that some candidate $c \in C$ is currently being considered. If $c$ beats all of its next $m-h$ successors with respect to $V$, we say that $c$ eliminates each of these $m-h$ candidates (or equivalently, that each of these candidates is eliminated by $c$). Otherwise, there exists a candidate $c'$ among these $m-h$ successors who is unbeaten by $c$. In this case, we say that $c$ is eliminated by $c'$.  

Note that a candidate $c$ may be considered multiple times. For example, if $c$ eliminates its next $m-h$ successors, it will remain the next candidate to be considered. We say that $c$ is directly eliminated if it is eliminated by one of its successors during its first consideration. Nevertheless, when $m \geq 2h-1$, each candidate can be considered at most twice. Specifically, if $C'$ denotes the set of the first $h-1$ candidates in $\rhd$, then every candidate in $C'$ can be considered at most twice, while every other candidate can be considered at most once.  
However, this is not the case when $m < 2h-1$. For instance, consider $m=5$ and $h=4$. In this scenario, $m-h=1$, and the first candidate in the agenda can be considered up to four times.

\begin{theorem}
\label{thm-dcac-suc-p}
For every constant $h$, 
{\prob{DCAC}}-$(m-h)$-Amendment is polynomial-time solvable.
\end{theorem}

\begin{proof}
Let $I = (C, p, D, V, \rhd, k)$ be an instance of the {\prob{DCAC}}-$(m-h)$-Amendment problem. 
Our algorithm first addresses the case where $I$ admits a constant-sized feasible solution, followed by the case where no such feasible solution exists.
\medskip

\noindent{\textbf{Step~1: Handling instances with a constant-sized feasible solution.}}

In this step, we determine whether there exists a subset $D' \subseteq D$ with $|D'| \leq \min\{2h-1, k\}$ such that $p$ is not the $(m-h)$-amendment winner of $(C \cup D', V)$ with respect to~$\rhd$. Since $h$ is a constant, this can be achieved in polynomial time. 

If such a subset exists, we conclude that $I$ is a {\yesins}. Otherwise, we consider the following two cases:
\begin{itemize}
    \item If $|D| \leq 2h-1$ or $k \leq 2h-1$, we directly conclude that $I$ is a {\noins}.
    \item If $|D| \geq 2h$ and $k \geq 2h$, the algorithm proceeds to the next step.
\end{itemize}
\medskip

\noindent{\textbf{Step~2: Handling instances without a constant-sized feasible solution.}}

At this point, we know that $I$ does not admit a feasible solution $D' \subseteq D$ with at most~$2h-1$ candidates. 
Our algorithm proceeds by enumerating all tuples $(S^{\text{L}}, S^{\text{R}})$ of subsets of~$D$ such that $S^{\text{L}} \cap S^{\text{R}} = \emptyset$ and $\abs{S^{\text{L}}} =\abs{S^{\text{R}}} = h-1$. (For clarity and ease of reference, we note that the superscripts~L and~R denote ``left'' and ``right'', respectively.) The candidates in~$S^{\text{L}}$ (respectively,~$S^{\text{R}}$) are intended to represent the first (respectively, last) $h-1$ candidates among all added candidates in the agenda~$\rhd$. 

Let $k' = k-\abs{S^{\text{L}}}-\abs{S^{\text{R}}}= k - 2h + 2$. As $k\geq 2h$, $k'$ is positive. For each enumerated  $(S^{\text{L}}, S^{\text{R}})$, we define $S^{\text{L} \cup \text{R}} = S^{\text{L}} \cup S^{\text{R}}$, and construct a subinstance to determine whether there exists a subset $S \subseteq D \setminus S^{\text{L} \cup \text{R}}$ such that:
\begin{itemize}
    \item $\abs{S} \leq k'$,
    \item $p$ is not the $(m-h)$-amendment winner of $(C \cup S^{\text{L} \cup \text{R}} \cup S, V)$ with respect to $\rhd$,
    \item $S^{\text{L}} \rhd S \rhd S^{\text{R}}$, meaning no candidate from $S^{\text{L}\cup \text{R}}$ appears between candidates from $S$ in the agenda $\rhd$.% i.e., every candidate in $S$ is a successor of every candidate in $S^{\text{L}}$ and a predecessor of every candidate in $S^{\text{R}}$.
\end{itemize}
Since~$h$ is a constant, the number of subinstances is at most polynomial in the size of $I$. Furthermore, $I$ is a {\yesins} if and only if at least one of the subinstances is a {\yesins}.

Next, we present a polynomial-time algorithm for solving a subinstance specified by an enumerated tuple $(S^{\text{L}}, S^{\text{R}})$.

Let $d^{\text{LR}}$ be the rightmost candidate from $S^{\text{L}}$ in $\rhd$, and let $d^{\text{RL}}$ be the leftmost candidate from $S^{\text{R}}$ in $\rhd$. To aid our exposition, we define two disjoint subsets $X$ and $Y$ of $D \setminus S^{\text{L} \cup \text{R}}$:
\begin{itemize}
    \item $X = \{d \in D \setmid d^{\text{LR}} \rhd d \rhd d^{\text{RL}}, d \rhd p\}$, which consists of all predecessors of $p$ from $D$ that lie between $d^{\text{LR}}$ and $d^{\text{RL}}$ in the agenda~$\rhd$.
    \item $Y = \{d \in D \setmid d^{\text{LR}} \rhd d \rhd d^{\text{RL}}, p \rhd d\}$, which consists of all successors of $p$ from $D$ that lie between $d^{\text{LR}}$ and $d^{\text{RL}}$ in the agenda~$\rhd$.
\end{itemize}

%We observe that,  that $p$ is the $(m-h)$-amendment winner of $(C \cup S^{\text{L}{{\cup}} \text{R}}, V)$ with respect to~$\rhd$, since otherwise we have conclude that $I$ is a {\yesins} in Step~1. 
We observe that $p$ beats all its successors in $C \cup S^{\text{L} \cup \text{R}} \cup Y$ with respect to $V$. The reasoning is as follows: if $p$ has a successor $d \in C \cup S^{\text{L} \cup \text{R}} \cup Y$ that is not beaten by $p$, then $p$ is not the $(m-h)$-amendment winner of $(C \cup \{d\}, V)$. In this case, the instance would already have been resolved in Step~1.

Let $\ell$ denote the number of successors of $p$ from $C \cup S^{\text{L} \cup \text{R}}$ with respect to the agenda $\rhd$. We analyze the problem by considering the following two cases:
\begin{description}
    \item[Case 1:]  $\ell \geq h-1$. \hfill
    
    In this case, using reasoning similar to the proof of Lemma~\ref{lem-dcac-full-amendment-adding-one-candidate-enough}, we can show that if the subinstance is a {\yesins}, there must exist a candidate $d \in X$ such that $p$ does not win $(C \cup S^{\text{L} \cup \text{R}} \cup \{d\}, V)$. Since $\abs{S^{\text{L} \cup \text{R}} \cup \{d\}} = 2h-1$, this implies that the instance would already have been resolved in Step~1.
    
    \item[Case 2:]  $\ell < h-1$. \hfill
    
    Let $k'' = \min\{h - \ell, k'\}$. In this case, we determine whether there exists a subset $Z \subseteq (X \cup Y)$ such that $\abs{Z} \leq k''$ and $p$ is not the $(m-h)$-amendment winner of $(C \cup S^{\text{L} \cup \text{R}} \cup Z, V)$ with respect to $\rhd$. Since $h$ is a constant, this can be done in polynomial time. 
    If such a subset~$Z$ exists, we conclude that $I$ is a {\yesins}; otherwise, we discard the enumerated~$(S^{\text{L}}, S^{\text{R}})$.

Regarding the correctness of this step, the result is evident when $h - \ell > k'$.  

For the case where $h - \ell \leq k'$, we claim that if the subinstance is a {\yesins}, then adding a set $Z$ of at most $k''=h-\ell$ candidates from $X \cup Y$ to the election $(C \cup S^{\text{L} \cup \text{R}}, V)$ is sufficient to prevent $p$ from winning. We prove this claim as follows. Assume that there exists a subset $Z' \subseteq (X \cup Y)$ such that $p$ is not the $(m-h)$-amendment winner of $(C \cup S^{\text{L} \cup \text{R}} \cup Z', V)$. Consequently, $p$ is eliminated when one of its predecessors, say $c^{\star}$, is considered during the winner determination process applied to $(C \cup S^{\text{L} \cup \text{R}} \cup Z', V)$. Let $\rhd'$ denote the agenda $\rhd$ restricted to $C \cup S^{\text{L} \cup \text{R}} \cup Z'$.

We construct the desired set $Z$ as follows: For each candidate $c \in \rhd'[1, h-1-\ell]$---that is, each candidate within the first $h-1-\ell$ positions of $\rhd'$---which does not beat at least one of its successors in $\rhd'$, we identify the leftmost successor of $c$ in $\rhd'$ that is not beaten by $c$ with respect to $V$. Let $C'$ denote the set of these ``leftmost'' candidates. We then define  
\[
Z = \left(C'\cup \{c^\star\}\right) \cap Z'.
\]
Clearly, $\abs{Z} \leq \abs{C'}+\abs{\{c^{\star}\}}\leq  (h-1-\ell)+1=h-\ell$. Let $E = (C \cup S^{\text{L}\cup \text{R}} \cup Z, V)$, and let~$\rhd''$ be $\rhd$ restricted to $C \cup S^{\text{L}\cup \text{R}} \cup Z$. Note that due to the presence of candidates in~$S^{\text{L}}$, the first $h-1$ candidates in $\rhd'$ are identical to those in $\rhd''$. Similarly, the presence of candidates in $S^{\text{R}}$ ensures that the last $h-1$ candidates in $\rhd'$ are also the same as those in $\rhd''$. 
 Using similar reasoning as in the proof of Lemma~\ref{lem-dcac-full-amendment-adding-one-candidate-enough}, we prove below that $p$ does not win $E$. Our proof proceeds by considering the following three cases.

\begin{description}
    \item[Case 1:] $c^{\star}$ is within the first $h-1-\ell$ positions of $\rhd'$. \hfill 

     In this case, $c^{\star}$ is not eliminated in $(C\cup S^{\text{L}\cup \text{R}}\cup Z', V)$ by any of its predecessors from $\rhd'[1,h-1-\ell]$. Consider now the winner determination process applied to~$E$ and $\rhd''$. By the definition of $C'$, all predecessors of~$c^{\star}$ will be directly eliminated. Since $c^{\star}$ eliminates $p$ in $(C \cup S^{\text{L}\cup \text{R}} \cup Z', V)$ and $Z \subseteq Z'$, it follows that $c^{\star}$ will also eliminate~$p$.

     \item[Case 2:] $c^{\star}$ is within neither the first $h-1-\ell$ positions nor the last $h-1$ positions of~$\rhd'$. \hfill

As $c^{\star}$ is not eliminated in $(C\cup S^{\text{L}\cup \text{R}}\cup Z', V)$ by any of its predecessors from $\rhd'[1,h-1-\ell]$ in the election $(C\cup S^{\text{L}\cup \text{R}}\cup Z', V)$, by the definition of $C'$, $c^{\star}$ is not eliminated by any of its predecessors from $\rhd''[1,h-1-\ell]=\rhd'[1,h-1-\ell]$ in $E$. It follows that all the first $h-1-\ell$ candidates in~$\rhd''$ are directly eliminated in $E$. Consequently, if $c^{\star}$ is eliminated when one of its predecessors, say $a$, is considered in $E$, then $a$ must not belong to the first $h-1-\ell$ candidates in~$\rhd''$. In this case, all the next $(m-h)$ successors of~$a$, including both~$c^{\star}$ and $p$, are eliminated. Otherwise,~$p$ will be eliminated by $c^{\star}$ in $E$.

\item[Case 3:] $c^{\star}$ is within the last $(h-1)$ positions of $\rhd'$. \hfill

Consider the winner determination process applied to $E$ and $\rhd''$. If all the first $h-1-\ell$ candidates in $\rhd''$ are directly eliminated, the proof proceeds similarly to that of Case~2. Otherwise, assume that some candidate $c  \in \rhd''[h-1-\ell]$ is not directly eliminated. In this case, after $c$ is considered for the first time, at most $h$ candidates remain. It follows that either~$p$ is eliminated together with $c^{\star}$ by some predecessors of $c^{\star}$, or $p$ is eliminated by $c^{\star}$.
\end{description}
%It remains to consider the case where $c^{\star}$ is within the last $(h-1)$ positions of $\rhd'$. In this scenario, $c^{\star}$ and $p$ are either eliminated simultaneously by some predecessor of $c^{\star}$, or $p$ is eliminated by $c^{\star}$.
\end{description}

Finally, if all enumerated tuples are discarded, we conclude that~$I$ is a {\noins}.
\end{proof}

The proof of Theorem~\ref{thm-dcac-suc-p} in fact implies that if an instance of {\prob{DCAC}}-$(m-h)$-Amendment is a {\yesins}, it must have a feasible solution of size at most $\min\{k, 3h-2\}$. 
%We conjecture that, with a more in-depth analysis, this bound can be improved to $\min\{k, h\}$. We leave this as an open question.

%%%%%%%%%%%%%%%%%%%%%%%%%%%%%%%%%%
%%%%%%%%%%%%%%%%%%%%%%%%%%%%%%%%%%%
%%%%%%%%%%%%%%%%%%%%%%%%%%%%%%%%%%
%%%%%%%%%%%%%%%%%%%%%%%%%%%%%%%%%

\subsection{The Successive Procedure}
\label{sec-suc}
In this section, we study the successive procedure. We first consider constructive control by adding/deleting voters. A slight modification of the reduction presented in the proof of Theorem~\ref{thm-ccav-suc-np} leads to the following result.

\begin{theorem}
\label{thm-ccav-suc-wah-last}
{\prob{CCAV}}-{\memph{Successive}} is {\memph\wah} with respect to the number of added votes plus the number of registered votes. This holds even when the distinguished candidate is the last one in the agenda.
\end{theorem}

\begin{proof}
We show that the reduction established in the proof of Theorem~\ref{thm-ccav-suc-np} can be used to prove the theorem. 
Recall that our reduction is from the {\prob{Perfect Code}} problem, and we have the following $3\kappa+2$ registered votes in~$V$:
\begin{itemize}
\item a multiset $V_1$ of $\kappa+2$ votes with the preference $p\Succ \overrightarrow{X} \Succ \overrightarrow{Y}$;
\item a multiset $V_2$ of $\kappa-2$ votes with the preference $\overrightarrow{X} \Succ p\Succ \overrightarrow{Y}$; and
\item a multiset $V_3$ of $\kappa+2$ votes with the preference $\overrightarrow{X} \Succ \overrightarrow{Y} \Succ p$.
\end{itemize}
In addition, for each vertex $\vere_i$ in~$G$, we have an unregistered vote~$\succ_i$ with the preference
\[\left(\overrightarrow{X}[N_X[\vere_i]]\right) \Succ \left(\overrightarrow{Y}\setminus N_Y[\vere_i]\right) \Succ p \Succ \left(\overrightarrow{X}\setminus N_X[\vere_i]\right) \Succ \left(\overrightarrow{Y}[N_Y[\vere_i]]\right).\]

The agenda is $\rhd=(\overrightarrow{X}, \overrightarrow{Y}, p)$. The proof for the correctness is similar.

$(\Rightarrow)$ In the proof of Theorem~\ref{thm-ccav-suc-np}, we showed that if there is a perfect code, then after adding the unregistered votes corresponding to the perfect code,~$p$ ties with everyone else. This means that none of the predecessors of~$p$ is the winner after the deletion of the votes. As~$p$ is the last one in the agenda,~$p$ wins the resulting election.

$(\Leftarrow)$ Suppose that there is a $W'\subseteq W$ such that $\abs{W'}\leq \kappa$ and~$p$ is the successive winner of $\elec=(C, V\muplus W')$ with respect to~$\rhd$. Let $U'=\{u_i\in \vset \setmid\; \succ_i\in W'\}$ be the set of vertices corresponding to~$W'$. We claim that~$U'$ is a perfect code. Assume, for the sake of contradiction, that this is not the case. Then, one of the following cases occurs.
\begin{description}
\item[Case~1:] $\exists \vere_i\in \vset$ such that $\vere_i\not\in N_G[\vset']$. \hfill

In this case, all unregistered votes in~$W'$ rank~$y_i$ before all her successors. Thus, in~$\elec$, there are a total of  $\abs{V_3}+\kappa=\kappa+2+\kappa=2\kappa+2$ votes ranking~$y_i$ before the set of all her successors. However, in this case, either some predecessor of~$y_i$ or~$y_i$ herself wins~$\elec$, contradicting the assumption that~$p$ wins~$\elec$.

\item[Case~2:] $\exists \vere_i\in \vset$ and $\vere_j, \vere_{j'}\in \vset'$ such that $\vere_i\in N_G[\vere_j]\cap N_G[\vere_{j'}]$. \hfill

By the construction of the unregistered votes, the two votes~$\succ_j$ and~$\succ_{j'}$ in~$W'$ both rank~$x_i$ before all the successors of~$x_i$. So, in~$\elec$ there are at least $\abs{V_2\cup V_3}+2=2\kappa+2$ votes ranking~$x_i$ before all the successors of~$x_i$. However, in this case, either some predecessor of~$x_i$ or~$x_i$ herself wins~$\elec$, contradicting that~$p$ wins~$\elec$.
\end{description}

Since both cases lead to contradictions, we conclude that~$U'$ is a perfect code of~$G$.
\end{proof}

For constructive control by deleting voters, we have two {\wbhns} results.

\begin{theorem}
\label{thm-ccdv-suc-wah-last}
{\prob{CCDV}}-{\memph{Successive}} is {\memph\wbh} with respect to the number of deleted votes. This holds even when the distinguished candidate is the last one in the agenda.
\end{theorem}

\begin{proof}
We prove the theorem by a reduction from the {\prob{RBDS}} problem. Let $(G, \kappa)$ be an instance of {\prob{RBDS}}, where~$G$ is a bipartite graph with the bipartition $(R, B)$. Without loss of generality, we assume that $\abs{B}\geq \kappa$. Similar to the proof of Theorem~\ref{thm-ccdv-amd-np}, we assume that every red vertex has the same degree, which we denote by $\ell$. We create an instance of {\prob{CCDV}}-Successive as follows. For each red vertex we create one candidate denoted still by the same symbol for simplicity. In addition, we create three candidates~$q$,~$p$, and~$p'$, where~$p$ is the distinguished candidate. Let $C=R\cup \{q, p, p'\}$. The agenda is $\rhd=(q, \overrightarrow{R}, p', p)$.
We create the following $2\abs{B}+2\ell-\kappa$ votes:
\begin{itemize}
\item a multiset $V_1$ of $\abs{B}-\kappa$ votes with the preference $\overrightarrow{R}\Succ p\Succ p'\Succ q$;
\item a multiset $V_2$ of $\ell-1$ votes with the preference $q\Succ p\Succ p'\Succ \overrightarrow{R}$;
\item a singleton $V_3$ of one vote with the preference $q\Succ \overrightarrow{R}\Succ p\Succ p'$;
\item a multiset $V_4$ of $\ell$ votes with the preference $p'\Succ p\Succ q\Succ \overrightarrow{R}$; and
\item for each blue vertex $b\in B$, one vote~$\succ_b$ with the preference
\[q\Succ \left(\overrightarrow{R}[N_G(b)]\right) \Succ p'\Succ p \Succ \left(\overrightarrow{R}\setminus N_G(b)\right).\]
\end{itemize}
Let~$V_B$ be the multiset of votes corresponding to the blue vertices in~$B$. Let $V=V_1\cup V_2\cup_3\cup V_4\cup V_B$.
Let $k=\kappa$.  The {\prob{CCDV}}-Successive instance is $(C, p, V, \rhd, k)$.

The weighted majority graph of $(C, V)$ and the agenda~$\rhd$ are illustrated in Figure~\ref{fig-ccdv-suc-hard-a}. We prove the correctness of the reduction as follows.

\begin{figure}
\centering
{\includegraphics[width=0.6\textwidth]{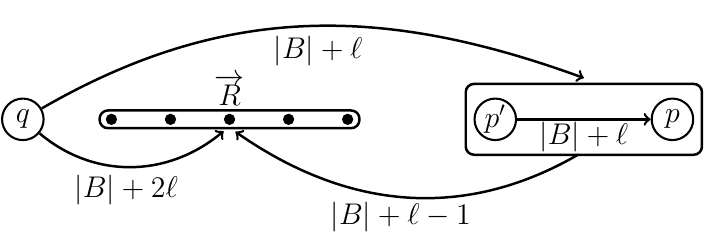}}
\caption{An illustration of the weighted majority graph of $(C, V)$ and the agenda as used in the proof of Theorem~\ref{thm-ccdv-suc-wah-last}. All arcs between the vertices in $R$ are forward arcs with a weight of at least $2\abs{B} - \kappa + \ell$. The agenda is depicted by the left-to-right ordering of the vertices.}
\label{fig-ccdv-suc-hard-a}
\end{figure}

$(\Rightarrow)$ Suppose that there is a $B'\subseteq B$ such that $\abs{B'}=\kappa$ and~$B'$ dominates~$R$ in~$G$. Let $V'=\{\succ_b\, \setmid b\in B'\}$ be the set of votes corresponding to~$B'$. Clearly, $\abs{V'}=\kappa=k$. Let $\elec=(C, V\setminus V')$.  We shall show that~$p$ is the successive winner of~$\elec$ with respect to~$\rhd$. First, as there are exactly~$(\ell-1)+1+(\abs{B}-\kappa)=\abs{B}+\ell-\kappa$ votes in $V\setminus V'$ ranking~$q$ in the top, $\abs{V\setminus V'}=2\abs{B}+2\ell-2\kappa$, and~$q$ is the first one in the agenda, we know that~$q$ is not the successive winner of~$\elec$. Let $r$ be any candidate from~$R$. Let~$S(r)$ be the set of~$r$'s successors in~$\rhd$. Clearly, all votes in $V_1\cup V_3$ rank~$r$ before~$S(r)$. Let $b\in B'$ be a blue vertex dominating~$r$. It holds that $r \succ_b p$. As exactly~$\ell$ votes in~$V_B$ rank~$r$ before~$p$, this implies that $V_B\setminus V'$ contains at most $\ell-1$ votes ranking~$r$ before~$p$. To summarize, there can be at most $\abs{V_1}+\abs{V_3}+\ell-1=\abs{B}-\kappa+\ell$ votes in $V\setminus V'$ ranking~$r$ before~$S(r)$, excluding the winning chance of~$r$. Observe that~$p$ ties with~$p'$ in~$\elec$, which excludes the winning chance of~$r$. We now know that none of the predecessors of~$p$ wins~$\elec$, and that~$p$ is the successive winner of~$\elec$ with respect to the agenda~$\rhd$, as it is the last candidate in~$\rhd$. 

$(\Leftarrow)$ Suppose that there exists $V'\subseteq V$ such that $\abs{V'}\leq k$ and~$p$ is the successive winner of $\elec=(C, V\setminus V')$ with respect to~$\rhd$. Observe that all votes in~$V'$ must rank~$q$ in the first place and $\abs{V'}= \kappa$, since otherwise~$q$ majority-dominates the set of all her successors in~$\elec$, which contradicts that~$p$ wins~$\elec$. There are three groups of votes ranking~$q$ in the first place: those in~$V_2$,~$V_3$, and~$V_B$. However, observe that it must hold that $V' \subseteq V_B$, since otherwise~$p'$ would beat~$p$ in~$\elec$. Given that~$p$ is the only successor of~$p'$, this contradicts the fact that~$p$ is the successive winner of~$\elec$. Let $B'=\{b\in B\, \setmid\; \succ_b\in V'\}$ be the set of blue vertices corresponding to~$V'$. We claim that~$B'$ dominates~$R$. Assume, for the sake of contradiction, that this is not the case. Let~$r$ be the leftmost red vertex in the agenda~$\rhd$ not dominated by any vertex from~$B'$. Then, all votes in~$V'$ rank~$p$ before~$r$. This implies that in $V\setminus V'$ there are at least $\abs{V_1}+\abs{V_3}+\ell=(\abs{B}-\kappa)+1+\ell$ votes ranking~$r$ before all successors of~$r$ (the last term~$\ell$ corresponds to the~$\ell$ votes in~$\{\succ_b\, \setmid b\in B, r\in N_G(b)\}$ which is a subset of $V\setminus V'$). However, as $\abs{V\setminus V'}= 2\abs{B}-2\kappa+2\ell$, this implies that~$r$ is the successive winner of~$\elec$, contradicting the winning of~$p$. So, we know that~$B'$ dominates~$R$. This implies that the {\prob{RBDS}} instance is a {\yesins}.
\end{proof}

We can show the {\wbhns} of {\prob{CCDV}}-Successive with respect to the dual parameter of the number of deleted votes.

\begin{theorem}
\label{thm-ccdv-suc-wah-remainning-last}
{\prob{CCDV}}-{\memph{Successive}} is {\memph\wbh} with respect to the number of remaining votes. Moreover, this holds even when the distinguished candidate is the last one in the agenda.
\end{theorem}

\begin{proof}
The reduction for the {\wbhns} of {\prob{CCDC}}-{\Famend} in the proof of Theorem~\ref{thm-ccdv-famend-wbh-remaining} directly applies here. The correctness argument is also the same (one needs only to replace ``{\famend}'' with ``successive'').
\end{proof}

However, if the distinguished candidate has only a few predecessors, the problem becomes tractable from the parameterized complexity point of view. As a matter of fact, we have a  result for a constructive control problem to which both {\prob{CCAV}}-$\tau$ and {\prob{CCDV}}-$\tau$ are polynomial-time Turing reducible.

\EP
{\prob{Exact Constructive Control by Editing Voters} \textnormal{for~$\tau$} (\prob{E-CCEV}-$\tau$)}
{A set~$C$ of candidates, a distinguished candidate $p\in C$, a multiset~$V$ of registered votes over~$C$, a multiset $W$ of unregistered votes over~$C$, an agenda~$\rhd$ on~$C$, and two nonnegative integers~$k$ and~$k'$.}
{Are there $V'\subseteq V$ and $W'\subseteq W$ such that $\abs{V'}=k$, $\abs{W'}=k'$, and~$p$ is the~$\tau$ winner of $(C, V'\muplus W')$ with respect to~$\rhd$?}

\begin{theorem}
\label{thm-e-ccev-suc-fpt}
{\prob{E-CCEV}}-{\memph{Successive}} is {\memph\fpt} with respect to the number of predecessors of the distinguished candidate in the agenda.
\end{theorem}

\begin{proof}
Let $(C, p, V, W, \rhd, k, k')$ be an instance of {\prob{E-CCEV}}-Successive. Let $m=\abs{C}$. Let $\rhd[\ell]=p$ for some $\ell\in [m]$. Therefore, the distinguished candidate~$p$ has $\ell-1$ predecessors. We present a natural ILP formulation for the given instance as follows.

We partition~$V$ (respectively,~$W$) into~$2^{\ell}$ submultisets, each identified by an $\ell$-length $1$-$0$ vector~${\bf{s}}$ and denoted by~$V_{\bf{s}}$ (respectively,~$W_{\bf{s}}$). Precisely, a vote $\succ$ from $V$ (respectively,~$W$) belongs to~$V_{\bf{s}}$ (respectively,~$W_{\bf{s}}$) if and only if  the following condition holds: 
\begin{itemize}
    \item For each $i\in [\ell]$, ${\bf{s}}[i]=1$ if and only if there exists some $j\in [m]$ with $j>i$ such that $\rhd[j] \succ \rhd[i]$, i.e., at least one of $\rhd[i]$'s successors is ranked before $\rhd[i]$ in the vote $\succ$.
\end{itemize}

For each submultiset~$V_{\bf{s}}$, we introduce a nonnegative integer variable $x_{\bf{s}}$, and for each submultiset~$W_{\bf{s}}$, we introduce a nonnegative integer variable $y_{\bf{s}}$. The variables $x_{\bf{s}}$ and $y_{\bf{s}}$ indicate, respectively, the number of votes from $V_{\bf{s}}$ and $W_{\bf{s}}$ that are contained in a certain feasible solution.  Let~${\bf{S}}$ be the set of all $\ell$-length $1$-$0$ vectors.
The constraints are as follows.
\begin{itemize}
    \item As we aim to find two submultisets $V'\subseteq V$ and $W'\subseteq W$ of cardinalities respectively~$k$ and~$k'$, we have that
    $\sum_{{\bf{s}}\in {\bf{S}}}x_{\bf{s}}=k$ and $\sum_{{\bf{s}}\in {\bf{S}}}y_{\bf{s}}=k'$.
    \item For each ${\bf{s}}\in {\bf{S}}$ we have that $0\leq x_{\bf{s}}\leq \abs{V_{\bf{s}}}$ and $0\leq y_{\bf{s}}\leq \abs{W_{\bf{s}}}$.
    \item For each $i\in [\ell-1]$, in order to ensure that $\rhd[i]$ is not the successive winner of $(C, V'\muplus W')$ with respect to~$\rhd$, we have that
    \[\sum_{\bf{s}\in {\bf{S}}} {\bf{s}}[i] \cdot (x_{\bf{s}}+y_{\bf{s}})\geq \frac{k+k'}{2}.\]
    \item To ensure that~$p$ is the successive winner of $(C, V'\muplus W')$ with respect to~$\rhd$, we introduce
    \[\sum_{\bf{s}\in {\bf{S}}} {\bf{s}}[\ell] \cdot (x_{\bf{s}}+y_{\bf{s}})< \frac{k+k'}{2}.\]
\end{itemize}

As we have~$2^{\ell+1}$ variables, by Lemma~\ref{lem-ilp-fpt}, the above ILP can be solved in {\fpt}-time in~$\ell$.
\end{proof}

From Theorem~\ref{thm-e-ccev-suc-fpt}, we have the following corollary.

\begin{corollary}
\label{cor-ccav-ccdv-suc-fpt-predecessors}
{\prob{CCAV}}-{\memph{Successive}} and {\prob{CCDV}}-{\memph{Successive}} are {\memph\fpt} with respect to the number of predecessors of the distinguished candidate in the agenda.
\end{corollary}

Now we move on to destructive control by adding/deleting voters. We show that both problems become polynomial-time solvable for  the successive procedure.

\begin{theorem}
\label{thm-e-dcev-suc-p}
{\prob{E-DCEV}}-{\memph{Successive}} is polynomial-time solvable.
\end{theorem}

\begin{proof}
Let $I=(C, p, V, W, \rhd, k, k')$ be an instance of {\prob{E-DCEV}}-Successive. Let $m=\abs{C}$ be the number of candidates. Let $p=\rhd[\ell]$ for some $\ell\in [m]$. Moreover, let~$S$ be the set of the successors of~$p$ in the agenda~$\rhd$.
We solve~$I$ as follows.

We first check whether there exist $V' \subseteq V$ and $W' \subseteq W$ of cardinalities~$k$ and~$k'$, respectively, such that~$p$ does not majority-dominate $S$ with respect to $V' \cup W'$. This can be achieved using the following procedure. 
Let $V_p$ (respectively,~$W_p$) denote the submultiset of~$V$ (respectively,~$W$) that ranks~$p$ before all candidates in $S$, i.e., $V_p = \{\succ \in V \setmid \{p\} \succ S\}$ (respectively, $W_p = \{\succ \in W \setmid \{p\} \succ S\}$). Furthermore, let $V_{\overline{p}} = V \setminus V_p$ and let $W_{\overline{p}} = W \setminus W_p$. 
If $\abs{V_{\overline{p}}} \geq k$, we let $V'$ be any arbitrary submultiset of $V_{\overline{p}}$ of cardinality~$k$. Otherwise, we let~$V'$ be the union of $V_{\overline{p}}$ and any submultiset of exactly $k - \abs{V_{\overline{p}}}$ votes from~$V_p$. The submultiset~$W'$ is defined analogously. 

Now, if~$p$ does not majority-dominate $S$ with respect to $V' \cup W'$, we conclude that the given instance~$I$ is a {\yesins}.

Otherwise, to prevent~$p$ from winning, some predecessor of~$p$ in the agenda~$\rhd$ must win. 
In this case, we determine, for each predecessor~$c = \rhd[i]$ of~$p$ where $i < \ell$, whether there exist $V' \subseteq V$ and $W' \subseteq W$ with $\abs{V'}=k$ and $\abs{W'}=k'$ such that~$c$ majority-dominates $\rhd[i+1, m]$ with respect to $V' \cup W'$. 
This can be accomplished in polynomial time using a greedy algorithm. Specifically, let~$V_c$ denote the submultiset of votes in~$V$ that rank~$c$ before $\rhd[i+1, m]$, and let~$W_c$ denote the submultiset of votes in~$W$ that rank~$c$ before $\rhd[i+1, m]$. Then, we construct~$V'$ as a submultiset of~$V$ with cardinality~$k$, including as many votes from~$V_c$ as possible. Similarly, we construct~$W'$ as a submultiset of~$W$ with cardinality~$k'$, including as many votes from~$W_c$ as possible. 
If~$c$ majority-dominates $\rhd[i+1, m]$ in $(C, V' \cup W')$, we conclude that~$I$ is a {\yesins}.

If none of the predecessors of~$p$ provides a {\yes}-answer, we conclude that~$I$ is a {\noins}.
\end{proof}

From Theorem~\ref{thm-e-dcev-suc-p}, we obtain the following corollary.

\begin{corollary}
\label{cor-dcav-dcdv-suc-p}
{\prob{DCAV}}-{\memph{Successive}} and {\prob{DCDV}}-{\memph{Successive}} are polynomial-time solvable.
\end{corollary}

Now we study control by adding/deleting candidates for the successive procedure.
The following result is easy to see.

\begin{corollary}
\label{cor-immune-suc-first}
{\memph{Successive}} is immune to~{\prob{CCAC}} if the distinguished candidate is the first one in the agenda.
\end{corollary}

However, as long as the distinguished candidate is not in the first place of the agenda, the problem becomes intractable from the parameterized complexity perspective, and this holds even when there are only two registered candidates.

\begin{theorem}
\label{thm-ccac-suc-wbh-solution-size}
{\prob{CCAC}}-{\memph{Successive}} is {\memph\wbh} with respect to the number of added candidates, as long as the distinguished candidate is not in the first place of the agenda. Moreover, this holds even when there are only two registered candidates.
\end{theorem}

\begin{proof}
We prove the theorem by a reduction from {\prob{RBDS}}. Let $(G, \kappa)$ be an instance of {\prob{RBDS}}, where $G$ is a bipartite graph with the bipartition $(R, B)$. We construct an instance $(C, p, D, V, \rhd, k)$ of {\prob{CCAC}}-Successive as follows. We create only two registered candidates, denoted by~$p$ and~$q$, where~$p$ is the distinguished candidate. Let $C=\{p, q\}$. Then, for each blue vertex $b\in B$, we create one unregistered candidate denoted still by~$b$ for simplicity. Let $D=B$. Let~$\rhd$ be any agenda on $C\cup D$ such that~$q$ is in the first place (the relative order of other candidates  can be set arbitrarily).
We create the following votes:
\begin{itemize}
    \item one vote with the preference $p\Succ \overrightarrow{B} \Succ q$;
    \item $\abs{R}$ votes with the preference $q\Succ p\Succ \overrightarrow{B}$; and
    \item for each red vertex $r\in R$, one vote $\succ_r$ with the preference
    \[\left(\overrightarrow{B}[N_G(r)]\right) \Succ q\Succ p\Succ \left(\overrightarrow{B}\setminus N_G(r)\right).\]
\end{itemize}
Let $V_R=\{\succ_r\, \setmid r\in R\}$ denote the multiset of votes corresponding to the red vertices, and let~$V$ denote the multiset of all the $2\abs{R}+1$ votes created above. Lastly, let $k=\kappa$. 

The construction of $(C, p, D, V, \rhd, k)$ can clearly be completed in polynomial time. We now proceed to prove the correctness of the reduction. 

$(\Rightarrow)$ Suppose there exists a subset $B' \subseteq B$ such that $\abs{B'} = \kappa$ and $B'$ dominates $R$ in~$G$. Let $\elec = (B' \cup \{p, q\}, V)$. We will show that~$p$ is the successive winner of~$\elec$ with respect to~$\rhd$. 
First, observe that there are $\abs{R} + 1$ votes ranking~$p$ before all candidates in $B$. Hence, it suffices to prove that $q$ is not the successive winner of $\elec$. 
Since~$B'$ dominates $R$, for every vote $\succ_r$ corresponding to a red vertex $r \in R$, there exists at least one $b \in B'$ such that $b$ dominates $r$. Consequently, $b$ is ranked before $q$ in $\succ_r$. This implies that the $\abs{R}$ votes with the preference $q \Succ p \Succ \overrightarrow{B}$ are all votes ranking $q$ before $B'$.  
However, since there are a total of $2\abs{R} + 1$ votes, $q$ cannot be the successive winner of $\elec$. Thus, we are done.

$(\Leftarrow)$ Suppose there exists a subset $B' \subseteq B$ such that $\abs{B'} \leq k$ and $p$ is the successive winner of the election $\elec = (B' \cup \{p, q\}, V)$ with respect to the agenda $\rhd$. We show below that~$B'$ dominates $R$ in~$G$. Assume, for the sake of contradiction, that exists a red vertex $r \in R$ such that $r \not\in N_G(B')$. Then, the vote~$\succ_r$ corresponding to $r$ ranks $q$ before $B'$.  
Together with the $\abs{R}$ votes with the preference $q \Succ p \Succ \overrightarrow{B}$, there are at least $\abs{R} + 1$ votes ranking $q$ in the first place in the election $\elec$, implying that $q$ is the successive winner of $\elec$, a contradiction.
%
%The above correctness arguments hold regardless of the relative order of candidates from $C \cup D \setminus \{q\}$ in the agenda.
\end{proof}

For constructive control by deleting candidates, we first present a intractability result.

\begin{theorem}
\label{thm-ccdc-suc-wah-solution-size}
{\prob{CCDC}}-{\memph{Successive}} is {\memph\wah} with respect to the number of deleted candidates. This holds even when the distinguished candidate is the first one in the agenda.
\end{theorem}

\begin{proof}
We prove the theorem by a reduction from the {\prob{Clique}} problem. Let $(G, \kappa)$ be an instance of {\prob{Clique}}, where $G=(\vset, \eset)$ is a graph. Let $n=\abs{\vset}$ and $m=\abs{\eset}$ be the number of vertices and the number of edges of~$G$, respectively. Without loss of generality, we assume that $m\geq \kappa \cdot (\kappa-1)$. For each vertex in~$G$, we create one candidate denoted by the same symbol for simplicity. In addition, we create a distinguished candidate~$p$. Let $C=\vset\cup \{p\}$. Let~$\rhd$ be an agenda where~$p$ is the first candidate. We create the following votes.
\begin{itemize}
    \item First, we create a multiset of $m-\kappa\cdot (\kappa-1)+1$ votes, each with the preference $p\Succ \overrightarrow{\vset}$. 
    \item Then, for each edge $\ede\in \eset$ between two vertices~$\vere$ and~$\vere'$, we create one vote, denoted~$\succ_{\ede}$, with the preference  $\overrightarrow{\ede} \Succ p\Succ (\overrightarrow{\vset}\setminus \ede)$. Let~$V_2$ be the set of these~$m$ votes.
\end{itemize}
Let~$V=V_1\muplus V_2$ be the multiset of the above $2m-\kappa \cdot (\kappa-1)+1$ created votes.
Let $k=\kappa$, i.e., we are allowed to delete at most~$\kappa$ candidates.
We have completed the construction of an instance  $(C, p, V, \rhd, k)$ of {\prob{CCDC}}-Successive, which clearly can be done in polynomial time. In the following, we prove the correctness of the reduction.

$(\Rightarrow)$ Assume that there is a clique $K\subseteq\vset$ of size~$\kappa$ in the graph~$G$. We prove below that~$p$ is the successive winner of $(C\setminus K, V)$ with respect to~$\rhd$. Since~$p$ is the first candidate in the agenda, this amounts  to proving that there are at least $m-\frac{\kappa\cdot (\kappa-1)}{2}+1$ votes ranking~$p$ before $\vset\setminus K$. Let $\eset(K)$ be the set of edges whose endpoints are both contained in~$K$. As~$K$ is a clique of~$\kappa$ vertices,~$\eset(K)$ consists of exactly $\frac{\kappa\cdot (\kappa-1)}{2}$ edges. Let $\succ_{\ede}$ be a vote corresponding to an edge $\ede\in \eset(K)$. From the definition of $\succ_{\ede}$, only the two candidates corresponding to the two endpoints of~$\ede$ are ranked before~$p$ in the vote. This implies that all the $\frac{\kappa\cdot (\kappa-1)}{2}$ votes corresponding to $\eset(K)$ rank~$p$ before $\vset\setminus K$. Then, as all votes in~$V_1$ rank~$p$ in the top, in total there are at least $\abs{V_1}+\frac{\kappa\cdot  (\kappa-1)}{2}=m-\frac{\kappa\cdot (\kappa-1)}{2}+1$ votes in~$V$ ranking~$p$ before $\vset\setminus K$, we are done.

$(\Leftarrow)$ Let $K \subseteq \vset$ be an arbitrary subset of at most $k = \kappa$ vertices (candidates). If~$\abs{K} < \kappa$ or~$K$ is not a clique, then there can be at most $\frac{\kappa \cdot (\kappa - 1)}{2} - 1$ edges with both endpoints in~$K$. 
Following a similar reasoning as above, one can verify that at most $\abs{V_1} + \frac{\kappa \cdot (\kappa - 1)}{2} - 1 = m - \frac{\kappa \cdot (\kappa - 1)}{2}$ votes rank~$p$ before $\vset \setminus K$. This implies that~$p$ cannot be the successive winner of $(C \setminus K, V)$ with respect to the agenda~$\rhd$. 
Consequently, if~$G$ does not contain a clique of size~$\kappa$, the instance of {\prob{CCDC}}-Successive constructed above is a {\noins}.
\end{proof}

Additionally, we show that the problem remains {\wah} with respect to the number of candidates not deleted.

\begin{theorem}
\label{label-ccdc-suc-wah-dual-ssize}
{\prob{CCDC}}-{\memph{Successive}} is {\memph\wah} with respect to the number of remaining candidates. This holds even when the distinguished candidate is the first one in the agenda.
\end{theorem}

\begin{proof}
We prove the theorem via a reduction from the {\prob{Biclique}} problem. Let $(G, \kappa)$ be an instance of the {\prob{Biclique}} problem, where~$G$ is a bipartite graph with the bipartition $(X, Y)$. Let $m=\abs{X}$, and let $n=\abs{Y}$. We assume that $m>\kappa$ and $n>2\kappa$, which does not change the {\wahns} of the  {\prob{Biclique}}  problem.  (Otherwise, the problem can be solved in {\fpt}-time with respect to $\kappa$.)
We construct a {\prob{CCDC}}-Successive instance  as follows. For each vertex $x\in X$, we create one candidate denoted still by~$x$. In addition, we create a candidate~$p$ which is the distinguished candidate. Let $C=X\cup \{p\}$. Let~$\rhd$ be an agenda on~$C$ such that~$p$ is the first candidate. We create the following votes.
\begin{itemize}
\item For each vertex $y\in Y$, we create one vote~$\succ_y$ with the preference
\[\left(\overrightarrow{X}\setminus N_G(y)\right) \Succ p\Succ \left(\overrightarrow{X}[{N_G(y)}]\right).\] 
For a given $Y'\subseteq Y$, let~$V_{Y'}$ be the multiset of votes created for vertices in~$Y'$.
\item We create a multiset $V'$ of $n-2\kappa+1$ votes, each with the preference $p\Succ \overrightarrow{X}$.
\end{itemize}
Let~$V=V_Y\muplus V'$ denote the multiset of the $2n-2\kappa+1$ votes constructed above.
Let $k=m-\kappa$. The instance of {\prob{CCDC}}-Successive is $(C, p, V, \rhd, k)$.

%The above reduction clearly can be done in polynomial time. It remains to show its 
We prove the correctness of the reduction as follows.

$(\Rightarrow)$ Assume that $G$ contains a biclique $(X', Y')$ such that $\abs{X'}=\abs{Y'}=\kappa$. Let $\elec=(X'\cup \{p\}, V)$. We show that~$p$ is the successive winner of~$\elec$ with respect to~$\rhd$. First, as $(X', Y')$ is a biclique in~$G$, due to the above construction, every vote $\succ_y$ where $y\in Y'$ ranks~$p$ before all candidates in~$X'$. As all votes in~$V'$ rank~$p$ in the top, there are in total at least $\abs{Y'}+\abs{V'}=n-\kappa+1$ votes ranking~$p$ before~$X'$. As there are in total $2n-2\kappa+1$ votes,~$p$ is the successive winner of~$\elec$.

$(\Leftarrow)$ Assume that there is a subset $X'\subseteq X$ such that~$\abs{X'}\geq m-k=\kappa$ and~$p$ is the successive winner of $(X'\cup \{p\}, V)$ with respect to~$\rhd$. Observe that, since~$p$ is the first candidate in the agenda, for any subset $X''\subseteq X'$,~$p$ remains as the successive winner of $(X''\cup \{p\}, V)$. This observation allows us to assume that $\abs{X'}=\kappa$. Given this, we know that there are in total at least $n-\kappa+1$ votes in~$V$ ranking~$p$ before~$X'$. Furthermore, since $\abs{V'}=n-2\kappa+1$ and all votes in~$V'$ rank~$p$ in the top, it follows that there are at least~$\kappa$ votes in~$V_Y$ ranking~$p$ before~$X'$. Let $V_{Y'}\subseteq V_Y$, where $Y'\subseteq Y$ is a submultiset of~$\kappa$ votes, each of which ranks~$p$ before~$X'$. By the construction, every vertex in~$Y'$ is adjacent to all vertices in~$X'$. In other words, $(X', Y')$ forms a biclique of size~$\kappa$ in~$G$.
\end{proof}

However, if the distinguished candidate has only a constant number of successors, we can solve {\prob{CCDC}}-Successive in polynomial time. In particular, we show that the problem is {\fpt} with respect to the number of successors of the distinguished candidate.

\begin{theorem}
\label{thm-ccdc-suc-fpt-sucessors}
{\prob{CCDC}}-{\memph{Successive}} can be solved in time $\bigos{2^{\ell}}$, where~$\ell$ is the number of successors of the distinguished candidate in the agenda.
\end{theorem}

\begin{proof}
    Let $I = (C, p, V, \rhd, k)$ be an instance of \textsc{CCDC-Successive}. Let $C'$ denote the set of successors of $p$ in the agenda $\rhd$, and let $\ell = \abs{C'}$. To solve the problem, we enumerate all subsets $S \subseteq C'$ containing at most $k$ candidates. Each enumerated~$S$ represents a hypothesis that, in a desired feasible solution, precisely the candidates in $S$ are deleted from~$C'$. 

For a fixed subset $S \subseteq C'$, we determine whether it is possible to expand $S$ by including at most $k - \abs{S}$ candidates from $C \setminus (C' \cup \{p\})$ such that $p$ becomes the successive winner of $(C \setminus S, V)$ with respect to the restriction of $\rhd$ to $C \setminus S$. It follows that the given instance $I$ is a {\yesins} if and only if there exists a subset $S \subseteq C'$ for which the answer to the above question is {\yes}.

Our algorithm proceeds as follows. For each enumerated subset $S$, we execute the following steps:

\begin{enumerate}
    \item[(1)] If $p$ is the successive winner of $(C \setminus S, V)$ with respect to $\rhd$, we conclude that $I$ is a \textsc{Yes} instance.

\item[(2)] Otherwise, we exhaustively apply the following procedure:
\end{enumerate}

\noindent \textbf{Procedure:} If the successive winner of $(C \setminus S, V)$ with respect to $\rhd$ precedes $p$ in $\rhd$, add the successive winner to $S$.
\medskip

After the above procedure has been applied exhaustively, if $p$ is the successive winner of $(C \setminus S, V)$ and $\abs{S} \leq k$, we conclude that $I$ is a {\yesins}. Otherwise, we discard the current subset $S$.

If all subsets $S \subseteq C'$ are discarded, we conclude that the given instance $I$ is a {\noins}.

Since there are at most $2^{\ell}$ subsets $S$ to enumerate, and the above procedure for each subset~$S$ can be executed in polynomial time, the overall algorithm runs in time $\bigos{2^{\ell}}$. 
\end{proof}

Note that the {\wbhns} reduction of {\prob{CCAC}}-Successive in the proof of Theorem~\ref{thm-ccac-suc-wbh-solution-size} can be adapted to establish the {\wbhns} of {\prob{DCAC}}-Successive by designating~$q$ as the distinguished candidate. In the following, we present a simplified variant of this reduction to show a similar result for a more restrictive case.

\begin{theorem}
\label{thm-dcac-suc-wbh}
{\prob{DCAC}}-{\memph{Successive}} is {\memph\wbh} with respect to the number of added candidates. Moreover, this holds even when the distinguished candidate is the only registered candidate and is the first one in the agenda.
\end{theorem}

\begin{proof}
We prove the theorem by slightly modifying the reduction for {\prob{CCAC}}-Successive presented in the proof of Theorem~\ref{thm-ccac-suc-wbh-solution-size}. 
Specifically, given an instance $(G, \kappa)$ of {\prob{RBDS}}, where~$G$ is a bipartite graph with the bipartition $(R, B)$, we construct an instance of {\prob{DCAC}}-Successive as follows. We create only one registered candidate~$p$. That is, $C=\{p\}$. Then, for each blue vertex $b\in B$, we create one unregistered candidate denoted still by the same symbol for simplicity. Let $D=B$ be the set of the unregistered candidates. The agenda is $\rhd=(p, \overrightarrow{B})$. We create the following votes:
\begin{itemize}
\item $\abs{R}-1$ votes with the preference $p\Succ \overrightarrow{B}$; and
\item for each $r\in R$,  one vote~$\succ_r$ with the preference
\[\left(\overrightarrow{B}[N_G(r)]\right) \Succ p \Succ \left(\overrightarrow{B}\setminus N_G(r)\right).\]
\end{itemize}
Let $V_R=\{\succ_r\, \setmid r\in R\}$ denote the multiset of votes corresponding to the red vertices, and let~$V$ denote the multiset of the above created $2\abs{R}-1$ votes. We set $k=\kappa$. The {\prob{DCAC}}-Successive instance is  $(C, p, D, V, \rhd, k)$.

The construction can clearly be completed in polynomial time. We prove its correctness as follows.

$(\Rightarrow)$ Suppose that there is a $B'\subseteq B$ such that $\abs{B'}=\kappa$ and~$B'$ dominates~$R$ in~$G$. We show that~$p$ is not the successive winner of $(B'\cup \{p\}, V)$ with respect to~$\rhd$. In fact, as~$B'$ dominates~$R$, for every vote $\succ_r$ corresponding to a red vertex $r\in R$, there exists $b\in B'$ such that~$b'$ dominates~$r$ in~$G$, and hence~$b$ is ranked before~$p$ in~$\succ_r$. This implies that the $\abs{R}-1$ votes in $V\setminus V_R$ are all those who rank~$p$ before~$B'$. As we have in total $2\abs{R}-1$ votes and $p$ is in the first place of the agenda~$\rhd$,~$p$ cannot be the successive winner of $(B'\cup \{p\}, V)$ with respect to~$\rhd$.

$(\Leftarrow)$ Suppose there exists a subset $B' \subseteq B$ such that $\abs{B'} \leq k$ and $p$ is not the successive winner of $(B' \cup \{p\}, V)$ with respect to $\rhd$. We aim to show that $B'$ dominates $R$ in $G$.

Assume, for the sake of contradiction, that there exists $r \in R$ such that $r \not\in N_G(B')$. Observe that the distinguished candidate $p$ is ranked before~$B'$ in $\succ_r$. 
Together with the votes from $V \setminus V_R$, this results in at least $\abs{R}$ votes ranking $p$ in the first position in the election $(B' \cup \{p\}, V)$. Since $p$ is the first candidate in the agenda $\rhd$, it follows that $p$ must be the successive winner of $(B' \cup \{p\}, V)$---a contradiction.
\end{proof}

Next, we show that the same problem is tractable if the distinguished candidate has a back seat in the agenda.

\begin{theorem}
\label{thm-dcac-suc-fpt-successors}
{\prob{DCAC}}-Successive can be solved in time~$\bigos{2^{\ell}}$, where~$\ell$ is the number of successors of the distinguished candidate in the agenda.
\end{theorem}

\begin{proof}
Let $I=(C, p, D, V, \rhd, k)$ be an instance of {\prob{DCAC}}-Successive. Define~$D'$ as the set of successors of~$p$ contained in~$D$ with respect to the agenda $\rhd$, and let $\ell = \abs{D'}$. We enumerate all $S\subseteq D'$ of up to~$k$ candidates and, for each enumerated~$S$, proceed as follows: 
If~$p$ is not the successive winner of $(C\cup S, V)$ with respect to ${\rhd}$, we conclude that~$I$ is a {\yesins}. 
Otherwise, two cases arise:
\begin{description}
    \item[Case~1:] $\abs{S}=k$. \hfill 
    
    In this case, we discard $S$.
    
    \item[Case~2:] $\abs{S}<k$. \hfill 
    
    In this case, we check if there is a candidate $d\in D\setminus D'$ such that~$d$ is the successive winner of $(C\cup S\cup \{d\}, V)$ with respect to  ${\rhd}$.  If such a candidate exists,  we conclude that~$I$ is a {\yesins}; otherwise, we discard $S$.
\end{description}

Since there are at most~$2^{\ell}$ possible  choices for~$S$, and for each~$S$, the above algorithm  can be executed  in polynomial time,  the overall running time of the algorithm is $\bigos{2^{\ell}}$.
\end{proof}

Now, we turn to the destructive control by deleting candidates. 
The following corollary is straightforward.

\begin{corollary}
\label{cor-suc-immune-dcdc-first}
Successive is immune to {\prob{DCDC}} if the distinguished candidate is the first in the agenda.
\end{corollary}

Corollary~\ref{cor-suc-immune-dcdc-first} indicates that an election controller cannot take any effective action if they are limited to performing the candidate deletion operation and the distinguished candidate they aim to prevent from winning is the first in the agenda. However, if the distinguished candidate has at least one predecessor, the destructive election controller has an opportunity to influence the outcome. Nevertheless, as established in the following two theorems, the controller faces an intractable problem.

\begin{theorem}
\label{thm-dcdc-suc-wah-dis-last-solution}
{\prob{DCDC}}-{\memph{Successive}} is {\memph\wah} with respect to the number of deleted candidates, as long as the distinguished candidate is not the first one in the agenda.
\end{theorem}

\begin{proof}
We prove the theorem via a reduction from the {\prob{Clique}} problem. Let $(G, \kappa)$ be an instance of {\prob{Clique}}, where $G=(\vset, \eset)$. Let $m=\abs{\eset}$ be the number of edges in~$G$. Without loss of generality, we assume that $m\geq \frac{\kappa \cdot (\kappa-1)}{2}>0$. We create an instance $(C, p, V, \rhd, k)$ of {\prob{DCDC}}-Successive as follows. For each vertex in~$G$, we create one candidate denoted by the same symbol for notational brevity. In addition, we create two candidates~$q$ and~$p$, where~$p$ is the distinguished candidate.  Let $C=\vset\cup \{q, p\}$. Let $\rhd$ be an agenda on~$C$ where~$q$ is the first candidate. We create the following votes.
\begin{itemize}
\item First, we create a multiset $V_1$ of $m-\frac{\kappa \cdot (\kappa-1)}{2}+1$ votes, each with the preference
\[q\Succ p\Succ \overrightarrow{\vset}.\]
\item Second, we create a multiset $V_2$ of $\frac{\kappa \cdot (\kappa-1)}{2}$ votes, each with the preference
\[p\Succ \overrightarrow{\vset} \Succ q.\]
\item Third, for each edge $\ede\in \eset$, we create one vote $\succ_{\ede}$ with the preference
\[\overrightarrow{\ede} \Succ q\Succ p\Succ (\overrightarrow{\vset}\setminus \ede).\]
Let $V_{\eset}=\{\succ_{\ede}\, \setmid \ede\in \eset\}$ be the set of the votes corresponding to the edges in~$G$. We have that $\abs{V_{\eset}}=m$.
\end{itemize}
Let $V=V_1\cup V_2\cup V_{\eset}$ be the multiset of all the $2m+1$ votes created above. Let $k=\kappa$. 
%The instance of {\prob{DCDC}}-Successive is $(C, p, V, \rhd, k)$. 
%It is easy to check that the distinguished candidate~$p$ is the successive winner of $(C, V)$ with respect to~$\rhd$.
We prove the correctness of the reduction below.

$(\Rightarrow)$ Suppose that there is a clique~$K$ of~$\kappa$ vertices in~$G$. Defining $\elec=(C\setminus K, V)$, we claim that~$q$ is the successive winner of $\elec$ with respect to $\rhd$.  Let~$\eset(K)$ be the set of edges whose both endpoints are in~$K$, and let $V_{\eset(K)}=\{\succ_{\ede}\, \setmid \ede\in \eset(K)\}$ be the set of votes corresponding to~$\eset(K)$. Clearly, $\abs{V_{\eset(K)}}=\abs{\eset(K)}=\frac{\kappa \cdot (\kappa-1)}{2}$. According to the construction of the votes, after deleting candidates from~$K$, each vote in $V_{\eset(K)}$ ranks~$q$ in the first place. Therefore, in the election~$\elec$ there are at least $\abs{V_{\eset(K)}}+\abs{V_1}=m+1$ votes ranking~$q$ in the first place, implying that~$q$ is the successive winner of~$\elec$.

$(\Leftarrow)$ Suppose that there is a subset $C'\subseteq C\setminus \{p\}$ of at most~$k$ candidates such that~$p$ is not the successive winner of $(C\setminus C', V)$  with respect to $\rhd$. Let $\elec=(C\setminus C', V)$. Observe that~$p$ majority-dominates~$\vset$, the set containing all possible successors of~$p$. It follows that $q$ is the successive winner of~$\elec$, and $C'\subseteq U$. Then, as~$q$ is the first candidate in the agenda~$\rhd$ and $\abs{V}=2m+1$, there are at least $m+1$ votes ranking~$q$ in the first place in~$\elec$. This implies that, in the election~$\elec$, at least $\frac{\kappa \cdot (\kappa-1)}{2}$ votes in~$V_{\eset}$ rank~$q$ in the top. Let~$\eset'$ be the set of the edges corresponding to any arbitrary $\frac{\kappa  \cdot (\kappa-1)}{2}$ votes in~$\elec$ ranking~$q$ in the top. By the construction of the votes, we know that both endpoints of each edge in~$\eset'$ are contained in~$C'$. As $\abs{\eset'}=\frac{\kappa\cdot (\kappa-1)}{2}$ and $\abs{C'}\leq \kappa$, this is possible only if~$C'$ forms a clique of size~$\kappa$ in~$G$.
\end{proof}

\begin{theorem}
\label{thm-dcdc-suc-wah-remaining-one-predecessor}
{\prob{DCDC}}-{\memph{Successive}} is {\memph\wah} with respect to the remaining candidates. This result holds as long as the distinguished candidate is not the first in the agenda.
\end{theorem}

\begin{proof}
We prove the theorem by a reduction from {\prob{Biclique}}. Let $(G, \kappa)$ be an instance of {\prob{Biclique}}, where~$G$ is a bipartite graph with the bipartition $(X, Y)$. Let $m=\abs{X}$ and let $n=\abs{Y}$. Without loss of generality, we assume that $\min\{m, n\}>\kappa>1$. We create an instance of {\prob{DCDC}}-Successive as follows. First, for each vertex $x\in X$, we create one candidate denoted by the same symbol for notational brevity. Additionally, we introduce two candidates,~$q$ and~$p$, where~$p$ is the distinguished candidate. Let $C = X \cup \{q, p\}$. There are $m + 2$ candidates in total. Let~$\rhd$ be an agenda on~$C$ such that~$q$ is the first candidate. We create the following votes.
\begin{itemize}
\item First, we create a multiset~$V_1$ of~$\kappa$ votes, each with the preference
\[p\Succ \overrightarrow{X}\Succ q.\]
\item Second, we create a multiset~$V_2$ of~$n-\kappa+1$ votes, each with the preference
\[q\Succ p\Succ \overrightarrow{X}.\]
\item Third, for each vertex $y\in Y$, we create one vote~$\succ_y$ with the preference
\[\left(\overrightarrow{X}\setminus N_G(y)\right) \Succ q\Succ p\Succ \left(\overrightarrow{X}[N_G(y)]\right).\]
\end{itemize}
For a given $Y'\subseteq Y$, we use $V_{Y'}=\{\succ_y\, \setmid y\in Y'\}$ to denote the multiset of the votes corresponding to~$Y'$.
Let $V=V_1\cup V_2\cup V_Y$ be the multiset of all the~$2n+1$ votes created above. 
Let $k=m-\kappa$. The instance of {\prob{DCDC}}-Successive is $(C, p, V, \rhd, k)$. 

In the following, we prove the correctness of the reduction, i.e., the {\prob{Biclique}} instance is a {\yesins} if and only if there is a subset $C'\subseteq C$ such that $p\in C'$, $\abs{C'}\geq \kappa+2$, and~$p$ is not the successive winner of $(C', V)$ with respect to~${\rhd}$.

$(\Rightarrow)$ Assume that there are $X'\subseteq X$ and $Y'\subseteq Y$ such that $\abs{X'}=\abs{Y'}=\kappa$, and $X'\cup Y'$ induces a complete bipartite subgraph of~$G$. Let $C'=X'\cup \{p, q\}$ and let $\elec=(C', V)$. We claim that~$q$ is the successive winner of~$\elec$. Let $y$ be a vertex in~$Y'$. By the definition of~$\succ_y$,~$q$ is ranked before $X'\cup \{p\}$ in~$\succ_y$. Therefore, in total, there are at least $\abs{Y'}+\abs{V_2}=n+1$ votes ranking~$q$ before $X'\cup \{p\}$ in~$\elec$. Since $\abs{V}=2n+1$ and $X'\cup \{p\}$ is the set of successors of~$q$ in~$\elec$, we conclude that~$q$ is the successive winner of~$\elec$.

$(\Leftarrow)$ Observe that~$p$ majority-dominates the set of all its successors in $(C, V)$. Therefore, if~$p$ is not the winner after some candidates are deleted, it must be the case that~$q$ becomes the winner of the resulting election. As a consequence, let us assume that there is a subset $X'\subseteq X$ of at least~$\kappa$ candidates such that~$q$ is the successive winner of $(X'\cup \{q, p\}, V)$ with respect to ${\rhd}$. Similar to the proof for the~$\Rightarrow$ direction, we know that there are at least~$\kappa$ votes in~$V_Y$ ranking~$q$ before~$X'\cup \{p\}$. Let~$V_{Y'}$, where $Y'\subseteq Y$, be the multiset of such votes. 
By the construction of the votes, it holds that $X'\subseteq N_G(y)$ for all $y\in Y'$. This implies that $G[X'\cup Y']$ is a complete bipartite graph. As $\abs{X'}\geq \kappa$ and $\abs{Y'}\geq \kappa$, the instance of {\prob{Biclique}} is a {\yesins}.
\end{proof}

We have obtained numerous intractability results ({\nphns}, {\wahns}, or {\wbhns}) for the special cases where the distinguished candidate has only a constant number of predecessors or successors. These hardness results lead to the following corollaries:

\begin{corollary}
The following problems are {\memph{\paranph}} with respect to the number of predecessors of the distinguished candidate:
\begin{itemize}
\item {\prob{CCAV}}-$\tau$ and {\prob{CCDV}}-$\tau$ for each $\tau\in \{\text{Amendment}, \text{\Famend}\}$;
\item {\prob{X}}-Successive for each $\prob{X}\in \{\prob{CCAC}, \prob{DCAC}, \prob{CCDC}, \prob{DCDC}\}$; and
\item {\prob{X}}-Full-Amendment for each $\prob{X}\in \{\prob{DCAV}, \prob{DCDV}\}$.
\end{itemize}
\end{corollary}

\begin{corollary}
The following problems are {\memph{\paranph}} with respect to the number of successors of the distinguished candidate:
\begin{itemize}
\item {\prob{CCAV}}-$\tau$ and {\prob{CCDV}}-$\tau$ for each $\tau\in \{\text{Amendment}, \text{\Famend}, \text{Successive}\}$;
\item {\prob{CCAC}}-$\tau$ for each $\tau\in \{\text{\Famend}, {\text{Successive}}\}$;
\item {\prob{DCAV}}-$\tau$ and {\prob{DCDV}}-$\tau$ for each $\tau\in \{\text{Amendment}, \text{\Famend}\}$; and
\item {\prob{DCDC}}-Successive.
\end{itemize}
\end{corollary}

We also note that if a problem is shown to be intractable when the distinguished candidate~$p$ has~$\ell$ predecessors (respectively, successors) for some constant~$\ell$, the same result can be extended to the case where~$p$ has $\ell+1$ predecessors (respectively, successors). This can be achieved by introducing a new candidate who is immediately before (respectively, after)~$p$ in the agenda and is ranked last in all votes. 

\section{Algorithmic Lower Bounds}
\label{sec-lowerbounds}
In this section, we discuss how the reductions established in Section~\ref{sec-main-results} offer lower bounds on kernelization algorithms, exact algorithms, and approximation algorithms. The following notations are used in our discussions:
\begin{itemize}
\item $m_{\text{r}}$: number of registered candidates, i.e., $m_{\text{r}}=\abs{C}$.
\item $m_{\text{u}}$: number of unregistered candidates, i.e., $m_{\text{u}}=\abs{D}$.
\item $m$: total number of candidates, i.e., $m=\abs{C}+\abs{D}$.
\item $n_{\text{r}}$: number of registered votes, i.e., $n_{\text{r}}=\abs{V}$.
\item $n_{\text{u}}$: number of unregistered votes, i.e., $n_{\text{u}}=\abs{W}$.
\item $n$: total number of votes, i.e., $n=\abs{V}+\abs{W}$.
\item $k$: a given upper bound on the cardinality of a desired feasible solution.
\end{itemize}

\subsection{Kernelization Lower Bounds}
It is known that {\prob{RBDS}} does not admit any polynomial kernels with respect to both~$\abs{R}+\kappa$ and~$\abs{B}$, assuming {$\textsf{PH}\neq \Sigma_{\textsf{P}}^3$} (see, e.g., the book chapter edited by Cygan~et~al.~\cite{Cygan2015}, and the paper by Dom, Lokshtanov, and Saurabh~\cite{DBLP:journals/talg/DomLS14}). In light of this result, many of our reductions imply the nonexistence of polynomial kernels for various parameters, as these reductions are, in fact, polynomial parameter transformations. To clarify our discussion, we reiterate the formal definitions of kernelization and polynomial parameter transformations.

\begin{definition}[Kernelization]
Let~$P$ be a parameterized problem. A kernelization algorithm (or kernelization) for~$P$ is an algorithm that takes as input an instance $(X, \kappa)$ of~$P$ and outputs an instance $(X',\kappa')$ of $P$ such that the following four conditions hold simultaneously.
\begin{itemize}
\item The algorithm runs in polynomial time in the size of $(X,\kappa)$.
\item $(X, \kappa)$ is a {\yesins} of~$P$ if and only if $(X', \kappa')$ is a {\yesins} of $P$.
\item $\abs{X'}\leq f(\kappa)$ for some computable function~$f$  depending only on~$\kappa$.
\item $\kappa'\leq g(\kappa)$ for some computable function~$g$  depending only on~$\kappa$.
\end{itemize}
\end{definition}

The new instance $(X', \kappa')$ in the above definition is called a {kernel}, and the size of~$X'$ is referred to as the {size of the kernel}. If~$f(\kappa)$ is a polynomial function of~$\kappa$, we say that $(X', \kappa')$ is a {polynomial kernel}, and the problem~$P$ is said to {admit a polynomial kernel}.

\begin{definition}[Polynomial Parameter Transformation]
Let~$P$ and~$Q$ be two parameterized problems. A polynomial parameter transformation from~$P$ to~$Q$ is an algorithm that takes as input an instance $(X, \kappa)$ of~$P$ and outputs an instance $(X', \kappa')$ of~$Q$ such that the following three conditions are satisfied:
\begin{itemize}
    \item The algorithm runs in polynomial time with respect to the size of $(X, \kappa)$.
    \item $(X, \kappa)$ is a {\yesins} of~$P$ if and only if $(X', \kappa')$ is a {\yesins} of~$Q$.
    \item $\kappa' \leq g(\kappa)$, where $g$ is a polynomial function of~$\kappa$.
\end{itemize}
\end{definition}

\begin{lemma}[\cite{DBLP:conf/esa/BodlaenderTY09,Cygan2015,DBLP:journals/talg/DomLS14}]
Let~$P$ and~$Q$ be two parameterized problems such that the unparameterized version of~$P$ is {\memph\npc} and the unparameterized version of~$Q$ is in {\memph\np}. Then, if there exists a polynomial parameter transformation from~$P$ to~$Q$, it follows that~$Q$ admitting a polynomial kernel implies that~$P$ also admits a polynomial kernel.
\end{lemma}

With these notions in mind, and recalling that 
(1) the unparameterized version of {\prob{RBDS}} is {\npc}, 
(2) all election control problems considered in the paper are in {\np}, 
(3) {\prob{RBDS}} does not admit any polynomial kernels with respect to both~$\abs{R}+\kappa$ and~$\abs{B}$ assuming {$\textsf{PH}\neq \Sigma_{\textsf{P}}^3$}, and 
(4) many of our reductions are polynomial parameter  transformations, 
we can summarize the nonexistence of polynomial kernels for various election control problems in Table~\ref{tab-kernelization-lower-bounds}.

\begin{table}
\renewcommand{\tabcolsep}{1.65mm}
\captionsetup{singlelinecheck=off}
\caption{Nonexistence of polynomial kernels with respect to various parameters, assuming {$\textsf{PH} \neq \Sigma_{\textsf{P}}^3$}. Entries marked with ``$\backslash$'' indicate that the corresponding problems are solvable in polynomial time, while entries marked with ``?'' signify that our reductions do not provide a conclusive result regarding the nonexistence of polynomial kernels for these problems, given the current state of known lower bound techniques.}
\centering
{
\begin{tabular}{lllll} \toprule
&{\prob{CCAV}}-$\tau$
&{\prob{CCDV}}-$\tau$
&{\prob{CCAC}}-$\tau$
&{\prob{CCDC}}-$\tau$
\\ \midrule

{amendment}
& $n$, $m+n_{\text{r}}+k$ (Thm.~\ref{thm-ccav-amd-np})
&  $m+k$ (Thm.~\ref{thm-ccdv-amd-np})
&$\backslash$
&$\backslash$
\\

&
&$n$, $m+n-k$ (Thm.~\ref{thm-ccdv-amd-wbh-remaining-votes})
&
& \\  \midrule

{\famend}
&?
&$m+k$ (Thm.~\ref{thm-ccdv-suc-np})
&$m_{\text{u}}$, $m_{\text{r}}+k$ (Thm.~\ref{thm-ccac-suc-np})
&$\backslash$
\\

&
&$n$, $m+n-k$ (Thm.~\ref{thm-ccdv-famend-wbh-remaining})
&
& \\ \midrule

{successive}
&?
& $n$, $m+k$ (Thm.~\ref{thm-ccdv-suc-wah-last})
& $m$, $n+k$ (Thm.~\ref{thm-ccac-suc-wbh-solution-size})
& ?
\\

&
& $m+n-k$ (Thm.~\ref{thm-ccdv-suc-wah-remainning-last})
&
&  \\ \bottomrule \toprule

%%%%%%%%%%%%%%%%%%%%%%%%%%%%
%%%%%%%%%%%%%%%%%%%%%%%%%%
%%%%%%%%%%%%%%%%%%%%%%%%%%
&\prob{DCAV}-$\tau$
&\prob{DCDV}-$\tau$
&\prob{DCAC}-$\tau$
&\prob{DCDC}-$\tau$
\\ \midrule

{amendment}
& ?
&$m+k$ (Thm.~\ref{thm-dcdv-amd-np})
&$\backslash$
&$\backslash$ \\

&
&$n$, $m+n-k$ (Thm.~\ref{thm-dcdv-amd-wbh-remaining})
&
& \\  \midrule

{\famend}
&?
&$n$, $m+k$ (Thm.~\ref{thm-dcdv-suc-np})
&$\backslash$
&$\backslash$ \\ \midrule

{successive}
&$\backslash$
&$\backslash$
& $m$, $n+k$ (Thm.~\ref{thm-dcac-suc-wbh})
& ? \\ \bottomrule
\end{tabular}
}
\label{tab-kernelization-lower-bounds}
\end{table}

\subsection{Exact Algorithm Lower Bounds}
Assuming the Strong Exponential Time Hypothesis (SETH), it is known that {\prob{RBDS}} cannot be solved in time $\bigos{(2-\epsilon)^{\abs{B}}}$ (see, e.g., the book chapter by Cygan et al.~\cite{Cygan2015}). Moreover, unless the Exponential Time Hypothesis (ETH) fails, {\prob{Clique}} cannot be solved in time $f(\kappa) \cdot z^{\smallo{\kappa}}$ (where $z$ denotes the number of vertices), and {\prob{RBDS}} cannot be solved in time $f(\kappa) \cdot \abs{B}^{\smallo{\kappa}}$ for any computable function $f$ in $\kappa$ (see, e.g., the work of Chen et al.~\cite{DBLP:journals/jcss/ChenHKX06}). These lower bounds, combined with several of our reductions, suggest that brute-force-based algorithms for many election control problems are essentially optimal. A summary of these results can be found in Table~\ref{tab-algorithm-lower-bounds}.

\begin{table}[ht!]
\renewcommand{\tabcolsep}{1.4mm}
\captionsetup{singlelinecheck=off}
\caption{Lower bounds for exact algorithms solving election control problems. The lower bound $\bigos{(2-\epsilon)^{n_{\text{u}}}}$ for {\prob{CCAV}}-Amendment is based on the SETH, while all other lower bounds are based on the ETH\@. Entries filled with ``$\backslash$'' indicate that the corresponding problems are solvable in polynomial time. Entries filled with ``?'' signify that, based on current lower bound techniques, our reductions do not imply lower bounds matching the running times of brute-force-based algorithms.
}
\centering
{
\begin{tabular}{lllll} \toprule
&{\prob{CCAV}}-$\tau$
&{\prob{CCDV}}-$\tau$
&{\prob{CCAC}}-$\tau$
&{\prob{CCDC}}-$\tau$
\\ \midrule

{amendment}
& $\bigos{(2-\epsilon)^{n_{\text{u}}}}$ (Thm.~\ref{thm-ccav-amd-np})
&  $f(k)\cdot n^{\smallo{k}}$ (Thm.~\ref{thm-ccdv-amd-np})
&$\backslash$
&$\backslash$
\\

& $f(k)\cdot n_{\text{u}}^{\smallo{k}}$ (Thm.~\ref{thm-ccav-amd-np})
&  $f(n-k)\cdot n^{\smallo{n-k}}$ (Thm.~\ref{thm-ccdv-amd-wbh-remaining-votes})
&
& \\  \midrule

{full-}
&?
& $f(k)\cdot n^{\smallo{k}}$ (Thm.~\ref{thm-ccdv-suc-np})
& $f(k)\cdot {m_{\text{u}}}^{\smallo{k}}$ 
&$\backslash$ \\ 

amendment&
& $f(n-k)\cdot n^{\smallo{n-k}}$ (Thm.~\ref{thm-ccdv-famend-wbh-remaining})
& (Thm.~\ref{thm-ccac-suc-np})
& \\
\midrule

{successive}
&?
& $f(k)\cdot n^{\smallo{k}}$ (Thm.~\ref{thm-ccdv-suc-wah-last})
& $f(k)\cdot {m_{\text{u}}}^{\smallo{k}}$ 
& $f(k)\cdot m^{\smallo{k}}$ \\ 

&
& $f(n-k)\cdot n^{\smallo{n-k}}$ (Thm.~\ref{thm-ccdv-suc-wah-remainning-last})
& (Thm.~\ref{thm-ccac-suc-wbh-solution-size})
& (Thm.~\ref{thm-ccdc-suc-wah-solution-size}) \\
\bottomrule \toprule

%%%%%%%%%%%%%%%%%%%%%%%%%%%%
%%%%%%%%%%%%%%%%%%%%%%%%%%
%%%%%%%%%%%%%%%%%%%%%%%%%%
&\prob{DCAV}-$\tau$
&\prob{DCDV}-$\tau$
&\prob{DCAC}-$\tau$
&\prob{DCDC}-$\tau$
\\ \midrule

{amendment}
& ?
& $f(k)\cdot n^{\smallo{k}}$ (Thm.~\ref{thm-dcdv-amd-np})
&$\backslash$
&$\backslash$ \\

& 
& $f(n-k)\cdot n^{\smallo{n-k}}$ (Thm.~\ref{thm-dcdv-amd-wbh-remaining})
&
& \\ \midrule

{\famend}
&?
& $f(k)\cdot n^{\smallo{k}}$ (Thm.~\ref{thm-dcdv-suc-np})
&$\backslash$
&$\backslash$ \\ \midrule

{successive}
&$\backslash$
&$\backslash$
&$f(k)\cdot m^{\smallo{k}}$  
& $f(k)\cdot m^{\smallo{k}}$ \\ 

&
&
& (Thm.~\ref{thm-dcac-suc-wbh})
& (Thm.~\ref{thm-dcdc-suc-wah-dis-last-solution}) \\ \bottomrule
\end{tabular}
}
\label{tab-algorithm-lower-bounds}
\end{table}

\subsection{Inapproximability Results}

\begin{table}[ht!]
\captionsetup{singlelinecheck=off}
\caption{Lower bounds for approximation algorithms based on the assumption $\poly\neq \np$. Entries filled with~``$\backslash$'' mean the corresponding problems are polynomial-time solvable. Entries filled with ``?'' mean that our reductions established in the paper do not imply meaningful inapproximability results.}
\centering
{
\begin{tabular}{lllll} \toprule
&{\prob{CCAV}}-$\tau$
&{\prob{CCDV}}-$\tau$
&{\prob{CCAC}}-$\tau$
&{\prob{CCDC}}-$\tau$
\\ \midrule

{amendment}
& $(1-\epsilon)\ln m$ (Thm.~\ref{thm-ccav-amd-np})
&   $(1-\epsilon)\ln m$ (Thm.~\ref{thm-ccdv-amd-np})
&$\backslash$
&$\backslash$ \\  \midrule

{\famend}
&?
&  $(1-\epsilon)\ln m$ (Thm.~\ref{thm-ccdv-suc-np})
& $(1-\epsilon)\ln m_{\text{r}}$ (Thm.~\ref{thm-ccac-suc-np})
&$\backslash$\\ \midrule

{successive}
&?
&  $(1-\epsilon)\ln m$ (Thm.~\ref{thm-ccdv-suc-wah-last})
&  $(1-\epsilon)\ln n$ (Thm.~\ref{thm-ccac-suc-wbh-solution-size})
& ? \\ \bottomrule \toprule

%%%%%%%%%%%%%%%%%%%%%%%%%%%%
%%%%%%%%%%%%%%%%%%%%%%%%%%
%%%%%%%%%%%%%%%%%%%%%%%%%%
&\prob{DCAV}-$\tau$
&\prob{DCDV}-$\tau$
&\prob{DCAC}-$\tau$
&\prob{DCDC}-$\tau$
\\ \midrule

{amendment}
& ?
&  $(1-\epsilon)\ln m$ (Thm.~\ref{thm-dcdv-amd-np})
&$\backslash$
&$\backslash$ \\  \midrule

{\famend}
&?
&  $(1-\epsilon)\ln m$ (Thm.~\ref{thm-dcdv-suc-np})
&$\backslash$
&$\backslash$ \\ \midrule

{successive}
&$\backslash$
&$\backslash$
& ?% $(1-\epsilon)(\ln n/2)$ (Thm.~\ref{thm-dcac-suc-wbh})
& ? \\ \bottomrule
\end{tabular}
}
\label{tab-poly-inapp}
\end{table}

This section is dedicated to the inapproximability consequences of the previously established reductions. For each of the eight standard election control problems, we examine its optimal version, where the goal is to add or delete the minimum number of voters or candidates such that the distinguished candidate becomes the winner (in the case of constructive control problems) or is not the winner (in the case of destructive control problems).

Let {\prob{Optimal RBDS}} be the optimal version of {\prob{RBDS}}, where the objective is to select a minimum number of blue vertices that dominate all red vertices.
It has long been known that unless $\np\subseteq \dtime(n^{\log\log n})$, {\prob{Optimal RBDS}} cannot be approximated in polynomial time within a factor of $(1-\epsilon)\ln \abs{R}$ (see, e.g., the work of Feige~\cite{DBLP:journals/jacm/Feige98}). This result has been recently improved by Dinur and Steurer~\cite{DBLP:conf/stoc/DinurS14} who showed that the same lower bound holds under the assumption $\poly\neq \np$.
Building on these results and our previous reductions, we obtain numerous inapproximability results for the optimal versions of the election control problems, as summarized in Table~\ref{tab-poly-inapp}.

\section{Conclusion}
\label{sec-conclusion}
We have investigated eight standard control problems under the $h$-amendment procedure and the successive procedure. The $1$-amendment procedure and the successive procedure are currently two of the most important sequential voting procedures used in practical parliamentary and legislative decision-making. The full-amendment procedure is a special case of $h$-amendment where~$h$ equals the number of candidates minus one. 
An advantage of the full-amendment procedure is that for the same election and the same agenda, if the full-amendment winner does not coincide with the amendment (respectively, successive) winner, then the former beats the latter. However, a disadvantage of the full-amendment procedure is that in the worst case it may need the head-to-head comparisons among all candidates to determine the winner, while the amendment procedure only needs voters to compare $m-1$ pairs of candidates, where~$m$ is the number of candidates.

Our study offers a comprehensive understanding of the parameterized complexity of these problems. In particular, we obtained both parameterized intractability results ({\wahns} results or {\wbhns} results) for the special cases where the distinguished candidate is the first candidate or the last candidate in the agenda, 
and many tractability results including some {\fpt}-algorithms and several polynomial-time algorithms. For a summary of our concrete results over the amendment procedure, the full-amendment procedure, and the successive procedure, we refer to Table~\ref{tab-resulst-summary}. Our study also yields a lot of algorithmic lower bounds, as summarized in Tables~\ref{tab-kernelization-lower-bounds}--\ref{tab-poly-inapp}.

Overall, our investigation conveys the following message.
\begin{itemize}
\item The amendment procedure and the successive procedure behave quite differently regarding their resistance to these control types: the amendment procedure demonstrates greater resistance to voter control, whereas the successive procedure proves more resistant to candidate control.
\item From a computational complexity perspective, the amendment procedure and the full-amendment procedure exhibit similar behavior concerning their resistance to the eight standard control types: both are resistant to the four voter control types and vulnerable to most candidate control types, with the exception that {\prob{CCAC}}-{\Famend} is computationally hard.
\item From the parameterized complexity perspective, the full-amendment procedure outperforms the amendment procedure in the sense that election control problems for the full-amendment procedure are at least as difficult to solve as the same problems for the amendment procedure. For instance, we showed that destructive control by adding/deleting voters for the amendment procedure is {\fpt} with respect to the number of predecessors of the distinguished candidate (Corollary~\ref{cor-dcav-dcdv-amd-fpt}). However, the same problems for the full-amendment procedure are
    {\nph}, even when the distinguished candidate has a constant number of predecessors (Theorem~\ref{thm-dcav-suc-np} and Theorem~\ref{thm-dcdv-suc-np}).
\item The position of the distinguished candidate on the agenda has a significant impact on the parameterized complexity of the problems. Specifically, from a complexity-theoretic perspective, most of these problems are more challenging to solve when the distinguished candidate~$p$ has a back position on the agenda than when~$p$ has a front position.
\end{itemize}

A summary of our results and those by Black~\cite{Black1958}, Farquharson~\cite{Farquharson1969}, and Bredereck~et~al.~\cite{DBLP:journals/jair/BredereckCNW17} reveals that there is no clear-cut conclusion on which procedure is the best one because these procedures all have their own advantages regarding resistance to different strategic behaviors. Nevertheless, in many practical applications, not all types of strategic behavior are likely to occur. If only one or two types of strategic behavior are likely to happen, our study might provide persuasive guidance on the choice of the procedures.

Finally, we touch upon some prominent topics for future research.
\begin{itemize}
\item First, as our study solely focuses on worst-case analysis, it is important to conduct experimental investigations to examine if the standard control problems are genuinely difficult to solve in practice.

\item Second, beyond the eight control types discussed in this paper, it is worthwhile to investigate additional control problems, such as constructive and destructive control by partitioning voters or candidates. For further insights on these topics, see the book chapter by Baumeister and Rothe~\cite{BaumeisterR2016}, the work of Erd\'{e}lyi, Hemaspaandra, and Hemaspaandra~\cite{DBLP:conf/aldt/ErdelyiHH15}, and the book chapter by Faliszewski and Rothe~\cite{handbookofcomsoc2016Cha7FR}.

\item Third, exploring the parameterized complexity of the problems with respect to structural parameters, such as the single-peaked width and the single-crossing width, is another avenue of investigation. Definitions of these structural parameters and related results can be found in the papers of Cornaz, Galand, and Spanjaard~\cite{DBLP:conf/ecai/CornazGS12,DBLP:conf/ijcai/CornazGS13}, and Yang and Guo~\cite{DBLP:conf/atal/YangG14a}. We note that since all $h$-amendment procedures are Condorcet-consistent, election control problems under these procedures can be solved in polynomial time when restricted to single-peaked or single-crossing domains, given certain mild restrictions (e.g., assuming an odd number of voters). This result leverages the techniques developed by Brandt~et~al.~\cite{DBLP:journals/jair/BrandtBHH15} and Magiera and Faliszewski~\cite{DBLP:journals/aamas/MagieraF17}.

\item Another interesting topic is investigating whether the {\poly}-results for destructive control problems can be extended to {\fpt}-results for the corresponding resolute control problems proposed first by Yang and Wang~\cite{DBLP:conf/atal/YangW17}. In this setting, there are multiple distinguished candidates whom the election controller would like to make nonwinners, and the parameter here is the number of the given distinguished candidates. 

\item Finally, numerous open problems remain to be addressed following our study. Beyond the unresolved questions concerning the parameterized complexity of various candidate control problems with respect to the number of voters, as discussed in Subsection~\ref{subsec-parameters}, the complexity of the following problems for constant $h \geq 2$ also remains open: {\prob{X}}-$h$-Amendment for $X \in \{\textsc{CCAC}, \textsc{CCDC}, \textsc{DCAC},$ $\textsc{DCDC}\}$, and {\prob{X}}-$(m-h)$-Amendment for $X \in \{\prob{CCDC}, \prob{DCDC}\}$. 
\end{itemize}

%although we have developed several polynomial-time algorithms for the candidate control problems under the $h$-Amendment rule and the $(m-h)$-Amendment rule for the specific value $h = 1$, it remains an open question whether these results extend to other constant values of $h$. In particular, 

\section*{Acknowledgment}
The author sincerely thanks the anonymous reviewers of AAMAS 2022 for their valuable feedback on an earlier version of this paper. Additionally, the author extends heartfelt gratitude to the anonymous reviewer of Algorithmica for their insightful and constructive comments, which significantly enhanced the quality of this work. In particular, the author acknowledges the anonymous reviewer of Algorithmica for their invaluable feedback on Reduction Rule~\ref{rule-2}. Their input helped correct a previously flawed statement and led to the development of the counterexample presented in Example~\ref{ex-2}.

%\section*{Declarations}

%\noindent{\bf Ethical Approval}
%
%not applicable
%\medskip

%\noindent{\bf Conflict of interest} The author declares no competing interests.
 %\medskip
 
%\noindent{\bf Authors' contributions}
%
%The author is solely responsible for this manuscript. All aspects of the research were conducted independently by the author.
 %\medskip
 %
%\noindent{\bf Funding}
%
%The work is not supported by any third-party funding.
%\medskip

 %
%\noindent{\bf Availability of data and materials}

%not applicable

%\section*{\refname}

%\bibliographystyle{plain}
%\bibliography{sociachoiceref,graphref}

%%%%%%%%%%%%%%%%%%%%%%%%%%%%%%%%%%%%%%%%%%%%%%%%%%%%%%%%%%%%%%%%%%%%%%%%

\end{document}